\documentclass[12pt]{iopart}
\usepackage{amssymb,amsmath,bm,hyperref} 
\usepackage{mathrsfs}
\usepackage{cases}
\usepackage{graphicx}
\usepackage{subfigure}
\usepackage{here}

\usepackage[normalem]{ulem}
\makeatletter
\def \Vec#1{{\bm #1}}
\def \be {\begin{equation}}
\def \ee {\end{equation}}
\newcommand{\Exp}[1]{\,\mathrm{e}^{\mbox{\footnotesize$#1$}}}
\newcommand{\I}{\mathrm{i}}

\newcommand{\ket}[1]{|#1\rangle}
\newcommand{\bra}[1]{\langle#1|}
\renewcommand{\tr}[1]{\mathrm{tr}\left(#1\right)}
\renewcommand{\Tr}[1]{\mathrm{Tr}\left\{#1\right\}}

\newcommand{\expectn}[1]{\langle#1\rangle}
\def \ds{\displaystyle}

\def \Re{\mathrm{Re}\,}
\def \Im{\mathrm{Im}\,}

\newcommand{\vin}[2]{\langle#1,#2\rangle}

\def \del{\partial}

\def \lra{\Leftrightarrow}

\def \argmin{\mathop{\rm argmin}}

\def \cB{{\cal B}}

\def \cM{{\cal M}}
\def \cH{{\cal H}}
\def \cX{{\cal X}}
\def \sofh{{\cal S}({\cal H})}

\def \bbr{{\mathbb R}}
\def \bbc{{\mathbb C}}

\def \sofc2{{\cal S}({\mathbb C}^2)}

\def \lofh{{\cal L}({\cal H})}
\def \lofhh{{\cal L}_h({\cal H})}
\def \cD#1{{\cal D}_{#1}}

\def \Eof#1{{\mathrm{E}}_{{\bm{\theta}}}^{(n)}[#1]}

\def \dI{d_{\mathrm{I}}}
\def \dN{d_{\mathrm{N}}}
\newcommand{\sldin}[2]{\langle#1,#2\rangle_{\rho_{{\bm{\theta}}}}^{\mathrm{S}}}
\newcommand{\rldin}[2]{\langle#1,#2\rangle_{\rho_{{\bm{\theta}}}}^{\mathrm{R}}}
\def \sldQFI{J_{\bm{\theta}}^{\mathrm{S}}}
\def \rldQFI{J_{\bm{\theta}}^{\mathrm{R}}}
\def \sldQFIinv#1{J_{\bm{\theta}}^{\mathrm{S};#1}}
\def \rldQFIinv#1{J_{\bm{\theta}}^{\mathrm{R};#1}}
\def \sldqfi#1{J_{#1}^{\mathrm{S}}}
\def \rldqfi#1{J_{#1}^{\mathrm{R}}}
\def \SLD#1{L_{{\bm{\theta}};#1}^{\mathrm{S}}}
\def \RLD#1{L_{{\bm{\theta}};#1}^{\mathrm{R}}}
\def \SLDdual#1{L_{{\bm{\theta}}}^{\mathrm{S};#1}}
\def \RLDdual#1{L_{{\bm{\theta}}}^{\mathrm{R};#1}}

\newtheorem{theorem}{Theorem}[section]
\newtheorem{lemma}[theorem]{Lemma}
\newtheorem{proposition}[theorem]{Proposition}
\newtheorem{corollary}[theorem]{Corollary}
\newtheorem{remark}[theorem]{Remark}

\newenvironment{proof}[1][Proof:]{\begin{trivlist}
\item[\hskip \labelsep {\bfseries #1}]}{\end{trivlist}
\hfill$\square$}

\newenvironment{definition}[1][Definition]{\begin{trivlist}
\item[\hskip \labelsep {\bfseries #1}]}{\end{trivlist}}

\newcommand{\qed}{\nobreak \ifvmode \relax \else
      \ifdim\lastskip<1.5em \hskip-\lastskip
      \hskip1.5em plus0em minus0.5em \fi \nobreak
      \vrule height0.75em width0.5em depth0.25em\fi}

\begin{document}

\title{Quantum state estimation with nuisance parameters}
\date{\today}
\author{Jun Suzuki}
\ead{junsuzuki@uec.ac.jp}
\address{Graduate School of Informatics and Engineering, The University of Electro-Communications, Tokyo, 182-8585 Japan}

\author{Yuxiang Yang}
\ead{yangyu@phys.ethz.ch}
\address{Institute for Theoretical Physics, ETH Z\"urich, 8093 Z\"urich, Switzerland}

\author{Masahito Hayashi}
\ead{masahito@math.nagoya-u.ac.jp}
\address{Graduate School of Mathematics, Nagoya University, Nagoya, 464-8602, Japan}
\address{Shenzhen Institute for Quantum Science and Engineering,
	Southern University of Science and Technology,
	Shenzhen,
	518055, China}
\address{Center for Quantum Computing, Peng Cheng Laboratory, Shenzhen 518000, China}
\address{Centre for Quantum Technologies, National University of Singapore, 117542, Singapore}

\begin{abstract}
In parameter estimation, nuisance parameters refer to parameters that are not of interest but nevertheless affect the precision of estimating other parameters of interest. For instance, the strength of noises in a probe can be regarded as a nuisance parameter. Despite its long history in classical statistics, the nuisance parameter problem in quantum estimation remains largely unexplored.
The goal of this article is to provide a systematic review of quantum estimation in the presence of nuisance parameters, and to supply those who work in quantum tomography and quantum metrology with tools to tackle relevant problems. 
After an introduction to the nuisance parameter and quantum estimation theory,
we explicitly formulate the problem of quantum state estimation with nuisance parameters. We extend quantum Cram\'{e}r-Rao bounds to the nuisance parameter case and provide a parameter orthogonalization tool to separate the nuisance parameters from the parameters of interest.  In particular, we put more focus on the case of one-parameter estimation in the presence of nuisance parameters, as it is most frequently encountered in practice.
\end{abstract}
 
\noindent{\it Keywords\/}: quantum state estimation, nuisance parameter, quantum Cram\'{e}r-Rao bounds

\submitto{J. Phys.\ A: Math.\ Theor.}

\maketitle 
%=====================================================================

%%%%%%%%%%%%%%%%%%%%%%%%%%%%%%%%%%%%%%%%%%
\section{Introduction}

The nuisance parameter problem, first pointed out by Fisher \cite{fisher35}, is one of the 
practical issues when dealing with parameter estimation problems. 
A parametric family of probability distributions is usually specified by multiple parameters, 
yet one  might be interested in only some of them.  A typical example is 
when one cares only about the expectation value and the variance of a random variable. 
Nuisance parameters are those that appear in the model but are not of interest. 
In principle, one can always try to estimate all parameters, including the nuisance parameters. 
However, in practice, this may be expensive or even impossible sometimes. 
One then wishes to explore more efficient strategies to estimate the parameters of interest 
by suppressing effects of the nuisance parameters. 
In  classical statistics, studies on the nuisance parameters problem have a long history; see, for example, books \cite{amari,lc,bnc,ANbook} 
and some relevant papers \cite{basu77,ka84,rc87,ak88,bs94,zr94}.
On the other hand, few studies on the nuisance parameter problem have been carried out so far in the quantum estimation theory.

Nuisance parameters are not merely a statistical concept. In fact, they persist in many physically relevant tasks of quantum estimation. 
Consider, as a simple example, the task of estimating a time parameter $t$ using identical copies of a two-level atom with Hamiltonian $\sigma_z/2=-(1/2)(\ket{0}\bra{0}-\ket{1}\bra{1})$ in a Ramsey interferometry. Ideally, each of the atoms would be in the pure qubit state $\ket{\psi_t}:=(1/\sqrt{2})(\ket{0}+e^{-it}\ket{1})$ at time $t$. Nevertheless, the atom's evolution is often affected by noise, and thus its state becomes mixed. A typical type of noise is dephasing, which causes the qubit to evolve under the master equation \cite{gardiner2004quantum,huelga1997improvement,ych17} 
$\partial\rho/\partial t=(i/2)[\sigma_z,\rho]+(\gamma/2)(\sigma_z\rho\sigma_z-\rho)$
where $\sigma_z:=\ket{0}\bra{0}-\ket{1}\bra{1}$ is the Pauli matrix and $\gamma\ge 0$ is the decay rate. For instance, $\gamma$ corresponds to the inverse of the relaxation time $T_2$ in Nuclear Magnetic Resonance (NMR), which can be pinpointed to a narrow interval $\mathcal{I}_\gamma\subset\mathbb{R}$ via benchmark tests. Under the dephasing evolution, the qubit state at time $t$ will be in the state
\begin{align}
\rho_{t,\gamma}=e^{-\gamma t}\ket{\psi_t}\bra{\psi_t}+\left(1-e^{-\gamma t}\right)\frac{I}{2}.
\end{align}
For the state $\rho_{t,\gamma}$, the symmetric logarithmic derivative (SLD) quantum Fisher information of the parameter $t$ can be calculated as 
\begin{align}
J^{\rm S}_t=e^{-2\gamma t}+\frac{\gamma^2}{e^{2\gamma t}-1}.
\end{align}
Naively, one might expect that the minimum estimation error $V_t$, quantified by the mean square error (MSE)  of the optimal unbiased estimator, would be $(J^{\rm S}_t)^{-1}$, as predicted by the well-known SLD quantum Cram\'{e}r-Rao (CR) bound $V_t\ge (J^{\rm S}_t)^{-1}$. However, this is not true. As we will show later in Section\ \ref{subsec-example-qubit}, the optimal estimation has an error
\begin{align}
\min\ V_{t}=e^{2\gamma t},
\end{align}
which is strictly larger than $(J^{\rm S}_t)^{-1}$.

The above example showcases that the single-parameter CR bound $V_t\ge (J^{\rm S}_t)^{-1}$ may not be tight even if there is only a single parameter $t$ of interest. The reason behind is that the state is determined not only by  $t$ but also by the unknown noise parameter $\gamma$. As a consequence, $\gamma$, which is not a parameter of interest but nonetheless affects the precision of estimating $t$, should be treated as a nuisance parameter. 
One can see from the example that nuisance parameters arise naturally in estimation problems concerning multiple parameters. As multiparameter quantum metrology \cite{Demkowicz-Dobrzanski2020,Albarelli2020} are prospering, 
the demand for a theory that treats the nuisance parameter problem in the quantum regime is also increasing.
The main purpose of this review is to provide a systematic overview of nuisance parameters in quantum state estimation and tools of determining ultimate precision limits in situations like this example.

In this review article, we provide a systematic overview of quantum estimation in the presence of nuisance parameters. 
Our primary aim is to review some facts in the nuisance parameter problem in classical statistics and then to provide a full survey on this problem in the quantum estimation theory. We stress that there are still many open problems on quantum estimation with nuisance parameters, and we list some of them at the end of this review.

We begin with a brief introduction of nuisance parameter in classical estimation theory (Section \ref{sec:Cnui}), guiding the readers through essential concepts like parameter orthogonalization. We also quickly review quantum estimation theory (Section \ref{sec:Qnui1}), focusing on the multiparameter case since nuisance parameters appear only when the model contains more than one parameter. We then proceed to discuss the nuisance parameter problem in quantum estimation (Section \ref{sec:Qnui2}). We explicitly formulate the problem and extend precision bounds to the nuisance-parameter case. We  provide a parameter orthogonalization tool to separate the nuisance parameters from the parameters we want to estimate (the parameters of interest) (Sections \ref{sec:QpoLocal} and  \ref{sec:QpoGlobal}). Since it is the fundamental and most frequently considered case, we put more focus on the case when there is only one parameter to estimate (Section \ref{sec:1para}). We illustrate the results for nuisance parameters with a couple of examples (Section \ref{sec:example}). 
Finally, we conclude by listing some open questions (Section \ref{sec:conclusion}).

%\newpage

\section{Nuisance parameter problem in classical statistics}\label{sec:Cnui}
This section summarizes the nuisance parameter problem in classical statistics. 
More details can be found in books \cite{amari,lc,bnc,ANbook} 
and relevant papers for this subject \cite{basu77,ka84,rc87,ak88,bs94,zr94}.

\subsection{Cram\'er-Rao inequality in the presence of nuisance parameters}
Let $p_{{\bm{\theta}}}(x)$ be a $d$-parameter family of probability distributions on a real-valued set $\cX$, 
where the $d$-dimensional real vector ${\bm{\theta}}=(\theta_1,\theta_2,\dots,\theta_d)$ 
takes values in an open subset of $d$-dimensional Euclidean space $\Theta\subset\bbr^d$. 
The $d$-parameter family: 
\be \label{c-npara-model}
\cM:=\{p_{\bm{\theta}}\,|\,{\bm{\theta}}\in\Theta\}
\ee
is called a {\it statistical model} or a {\it model} in short. 
To avoid mathematical difficulties, in this review we only consider regular models\footnote{
A statistical model is called regular when 
it satisfies the smoothness of the model (i.e.\ regarding differentiability) and
its Fisher information matrix is strictly positive.}, 
requiring that the mapping ${\bm{\theta}}\mapsto p_{\bm{\theta}}$ is one-to-one and  $p_{\bm{\theta}}$ can be differentiated sufficiently many times. 
This is because regularity conditions simplify several mathematical derivations. For example, 
the variations about different parameters are assumed to be linearly independent such that the Fisher information matrix is not singular. 
More technical regularity conditions can be found in the standard textbook \cite{lc} and also the book on this subject \cite{at95}. 

The standard problem of classical statistics is to find a good estimator that minimizes a given cost (risk) function under a certain condition.
An estimator is a mapping from the data set of   sample size $n$ to the parameter set. 
Let $\hat{{\bm{\theta}}}:\ \cX^n\to\Theta$ be an estimator, and assume that a given $n$-sequence of the observed data 
$x^n=x_1x_2\dots x_n$ is drawn according to an independently and identically distributed (i.i.d.) 
source $p^{(n)}_{{\bm{\theta}}}(x^n)=\Pi_{t=1}^n p_{{\bm{\theta}}}(x_t)$. 
An estimator is called {\it unbiased} if 
\be
\Eof{{\hat{\theta}}_i(X^n)}:=\sum_{x^n\in\cX^n}p_{\bm{\theta}}^{(n)}(x^n) {\hat{\theta}}_i(x^n)=\theta_i \quad(\forall i=1,2,\dots,d), 
\ee
holds for all parameter values ${\bm{\theta}}\in\Theta$. 
It is known that this condition of unbiasedness is often too strong and 
there may not be such an estimator. 
To relax the condition, we consider the Taylor expansion of the above unbiasedness condition. 
An estimator is called {\it locally unbiased} at ${\bm{\theta}}$ if 
\begin{align*}
\Eof{{\hat{\theta}}_i(X^n)}&=\sum_{x^n\in\cX^n}p_{\bm{\theta}}^{(n)}(x^n) {\hat{\theta}}_i(x^n)=\theta_i, \\
\del_j\Eof{{\hat{\theta}}_i(X^n)}&=\sum_{x^n\in\cX^n}\del_j p_{\bm{\theta}}^{(n)}(x^n) {\hat{\theta}}_i(x^n)=\delta_{i,j}, 
\end{align*}
holds for all $i,j=1,2,\dots,d$. Here $\del_j=\del/\del\theta_j$ denotes the $j$th partial derivative 
and $\delta_{i,j}$ is the Kronecker delta. The local unbiasedness condition requires 
the above conditions at the true parameter value ${\bm{\theta}}$. 
Clearly, if $\hat{{\bm{\theta}}}$ is unbiased, then $\hat{{\bm{\theta}}}$ is locally unbiased at any point. 
The converse statement is also true. 

The estimation error is quantified by the mean square error (MSE) matrix, defined as
\begin{equation} \nonumber
V_{{\bm{\theta}}}^{(n)}[\hat{{\bm{\theta}}}]
=\left[ E_{\bm{\theta}}^{(n)}\big[(\hat{\theta}_i(X^n)-\theta_i)(\hat{\theta}_j(X^n)-\theta_j) \big] \right]. 
\end{equation}
It is well-known that the following Cram\'er-Rao (CR) inequality holds for any locally unbiased estimator:
\be\label{crineqfull}
V_{{\bm{\theta}}}^{(n)}[\hat{{\bm{\theta}}}]\ge \frac{1}{n} \big( J_{\bm{\theta}}\big)^{-1}. 
\ee
Here, $J_{\bm{\theta}}=\big[J_{{\bm{\theta}};i,j}\big]$ denotes the Fisher information matrix about the model $\cM$, whose the $(i,j)$ component is defined by
\begin{align}
J_{{\bm{\theta}};i,j}=\sum_{x\in\cX}  p_{\bm{\theta}}(x) \frac{\del \ell_{\bm{\theta}}(x)}{\del\theta_i}\frac{\del \ell_{\bm{\theta}}(x)}{\del\theta_j} %\nonumber\\
=E_{\bm{\theta}}\Big[   \frac{\del \ell_{\bm{\theta}}(X)}{\del\theta_i}\frac{\del \ell_{\bm{\theta}}(X)}{\del\theta_j}\Big], \label{Def:CFI}
\end{align}
with $\ell_{\bm{\theta}}(x):=\log p_{\bm{\theta}}(x)$ being the logarithmic likelihood function and 
$E_{\bm{\theta}}[f(X)]$ being the expectation value of a random variable $f(X)$ with respect to $p_{\bm{\theta}}$. 
(See also the generalized CR inequality in Appendix \ref{sec:AppCstat1}.)

Suppose we are interested in estimating the values of a certain subset of parameters 
$\bm{\theta}_{\mathrm{I}}=(\theta_1,\theta_2,\dots,\theta_{\dI})$ ($\dI<d$), whereas the remaining set of $\dN=d-\dI$ parameters 
$\bm{\theta}_{\mathrm{N}}=(\theta_{\dI+1},\theta_{\dI+2},\dots,\theta_d)$ are not of interest. 
This kind of situation often occurs in various statistical inference problems 
and is of great importance for applications of statistics.  
We denote this partition as ${\bm{\theta}}=(\bm{\theta}_{\mathrm{I}},\bm{\theta}_{\mathrm{N}})$ and assume the similar partition 
for the parameter set $\Theta=\Theta_{\mathrm{I}}\times\Theta_{\mathrm{N}}$. 
In statistics \footnote{There exist several terminologies in statistics. In this paper, we only consider statistical models parametrized by a fixed number of parameters. A nuisance parameter is a certain subset of these parameters of no interest. In some literature, an {\it incident parameter} is also used as synonym.  See, for example, a review \cite{lancaster}.}, the parameters in $\bm{\theta}_{\mathrm{I}}$ are called the {\it parameters of interest} 
and the parameters in $\bm{\theta}_{\mathrm{N}}$ are referred to as the {\it nuisance parameters}.  
Here, an estimator for the parameters of interest returns a parameter value $\bm{\theta}_{\mathrm{I}}$ when given an $n$-sequence of the observed data 
$x^n=x_1x_2\dots x_n$, which is drawn according to i.\,i.\,d.\ source $p^{(n)}_{{\bm{\theta}}}(x^n)=\Pi_{t=1}^n p_{{\bm{\theta}}}(x_t)$. 
Mathematically, it is a map from $\cX^n$ to $\Theta_{\mathrm{I}}$. 
Let $\hat{{\bm{\theta}}}_{\mathrm{I}}=(\hat{{\bm{\theta}}}_1,\hat{{\bm{\theta}}}_2,\dots,\hat{{\bm{\theta}}}_{\dI})$ be an estimator, 
and the MSE matrix for the parameters of interest is defined by
\begin{align} \nonumber
V_{{\bm{\theta}}}^{(n)}[\hat{{\bm{\theta}}}_{\mathrm{I}}]
&=\left[ \sum_{x^n\in\cX^n} p^{(n)}_{{\bm{\theta}}}(x^n)(\hat{\theta}_i(x^n)-\theta_i)(\hat{\theta}_j(x^n)-\theta_j)  \right]\\
&=\left[ E_{\bm{\theta}}^{(n)}\big[(\hat{\theta}_i(X^n)-\theta_i)(\hat{\theta}_j(X^n)-\theta_j) \big] \right], 
\end{align}
where the matrix index takes values in the index set of parameter of interest, i.e., $i,j\in\{1,2,\dots,\dI\}$. 
By definition, the MSE matrix is a $\dI\times \dI$ real positive semidefinite matrix. 

It is important to find a precision bound for the parameter of interest. 
There are two different scenarios: one is when the nuisance parameters $\bm{\theta}_{\mathrm{N}}$ 
are completely known, and hence, $\bm{\theta}_{\mathrm{N}}$ are fixed parameters.  
The other is when we do not have prior knowledge on $\bm{\theta}_{\mathrm{N}}$, yet they appear in the statistical model.  
The former is a $\dI$-parameter problem whose model is 
\be \label{c-mpara-model}
\cM':=\{p_{\bm{\theta}_{\mathrm{I}}}|\bm{\theta}_{\mathrm{I}}\in\Theta_{\mathrm{I}}\}, 
\ee
and the latter is a $d$-parameter problem; $\cM$. 
The well-established result in classical statistics proves the following CR inequality: 
\begin{numcases}{V_{\bm{\theta}}^{(n)}
[\hat{{\bm{\theta}}}_{\mathrm{I}}]
\ge \frac{1}{n}  \!\times\!}
 \ds(J_{\bm{\theta}; \mathrm{I} ,\mathrm{I}} )^{-1}&\hspace{-12pt} ($\bm{\theta}_{\mathrm{N}}$ is known) \label{ccrineq1} \\[1ex]
 \ds J_{\bm{\theta}}^{\mathrm{I},\mathrm{I}}& \hspace{-12pt} ($\bm{\theta}_{\mathrm{N}}$ is not known)\label{ccrineq2}
\end{numcases}
 for estimators satisfying suitable local unbiasedness conditions, which we will discuss in more details in \eqref{lu_cond_class}.

In the above formula, two matrices $J_{\bm{\theta}; \mathrm{I},\mathrm{I}}$ and $J_{\bm{\theta}}^{\mathrm{I},\mathrm{I}}$ 
are defined as the block matrices of the Fisher information matrix $J_{\bm{\theta}}$ and the inverse Fisher information matrix $J_{\bm{\theta}}^{-1}$ 
according to the partition ${\bm{\theta}}=(\bm{\theta}_{\mathrm{I}},\bm{\theta}_{\mathrm{N}})$;
\be\label{fisherblock} 
J_{\bm{\theta}}=\left(\begin{array}{cc}
J_{\bm{\theta};\mathrm{I},\mathrm{I}} 
 & J_{\bm{\theta};\mathrm{I},\mathrm{N}}  \\[0.1ex] 
J_{\bm{\theta};\mathrm{N},\mathrm{I}} 
& 
J_{\bm{\theta};\mathrm{N},\mathrm{N}} 
\end{array}\right), \quad 
J_{\bm{\theta}}^{-1}=\left(\begin{array}{cc}
J_{\bm{\theta}}^{\mathrm{I},\mathrm{I}}& 
J_{\bm{\theta}}^{\mathrm{I},\mathrm{N}}
\\[0.1ex] 
J_{\bm{\theta}}^{\mathrm{N},\mathrm{I}}
&
J_{\bm{\theta}}^{\mathrm{N},\mathrm{N}}
\end{array}\right).   
\ee
We will make frequent use of this notation throughout the review.

The sub-block matrix in inequality \eqref{ccrineq2}, 
\be\label{cpFI}
(J_{\bm{\theta}}^{\mathrm{I},\mathrm{I}})^{-1}
=J_{\bm{\theta};{\mathrm{I}},{\mathrm{I}}}
-J_{\bm{\theta};{\mathrm{I}},{\mathrm{N}}}
(J_{\bm{\theta};{\mathrm{N}},{\mathrm{N}}})^{-1}
J_{\bm{\theta};{\mathrm{N}},{\mathrm{I}}}=:
J_{\bm{\theta}}({\mathrm{I}}|{\mathrm{N}}),
\ee
is known as the {\it partial Fisher information matrix} for the parameters of interest $\bm{\theta}_{\mathrm{I}}$, 
and it accounts the amount of information for $\bm{\theta}_{\mathrm{I}}$ that can be extracted from a given datum. 
Note that equality \eqref{cpFI} follows from well-known Schur's complements in matrix analysis; see, for example, \cite{bhatia}. 

%%%%
Here, we have four remarks concerning the classical CR inequalities \eqref{ccrineq1} and \eqref{ccrineq2}. 
First, the nuisance parameters here are treated as non-random variables. 
When they are random as in the Bayesian setting, the above CR inequality in the presence of the nuisance parameters needs to be 
replaced by the Bayesian version; see, for example, \cite{rm97}.  

Second, the nuisance parameters are defined up to an arbitrary reparameterization, since they are of no interest. 
It will be shown that a different representation of the nuisance parameters does not affect the CR bound for the parameter of interest. 
Consider the following change of parameters  
\be\label{nui_change0}
{\bm{\theta}}=(\bm{\theta}_{\mathrm{I}},\bm{\theta}_{\mathrm{N}})\mapsto{\bm{\xi}}=({\bm{\xi}}_{\mathrm{I}},{\bm{\xi}}_{\mathrm{N}})\mbox{ s.t. }{\bm{\xi}}_{\mathrm{I}}=\bm{\theta}_{\mathrm{I}}. 
\ee 
The condition ${\bm{\xi}}_{\mathrm{I}}=({\xi}_1,{\xi}_2,\dots,{\xi}_{\dI})=\bm{\theta}_{\mathrm{I}}$ ensures that 
the parameters of interest are unchanged while the nuisance parameters can be changed arbitrary. 
Details accounting on this additional degree of freedom will be discussed in Section \ref{sec:NuiPO}. 

Third, the case when the number of nuisance parameters are more than allowed by the regularity condition. 
When a statistical model is defined on the finite set $\cX$ with $|\cX|=D$, the total number of parameters should 
be at most $D-1$. Otherwise, the model is not regular. Now, suppose we have (possibly infinitely) many nuisance parameters violating 
this condition. (The number of parameters of interest $\dI$ should be less than $D-1$.) 
Even in this case, we can still derive the CR inequality for the parameters of interest using the concept of the partial Fisher information matrix \eqref{cpFI} \cite{ka84,ak88}. The detailed exposition of this procedure 
is postponed to at the end of section \ref{sec:AppCstat2-geo}, since we need additional definitions.  

Last, we adapt the unbiasedness conditions to the case when there are nuisance parameters.  
An estimator $\hat{{\bm{\theta}}}_{\mathrm{I}}$ for the parameter of interest is called {\it unbiased} for $\bm{\theta}_{\mathrm{I}}$, if the condition 
\be
E_{\bm{\theta}}^{(n)}[{\hat{\theta}_i}(X^{(n)})]=\theta_i \quad (\forall i=1,2,\dots,\dI),
\ee
holds for all ${\bm{\theta}}\in\Theta$. 
We next introduce the concept of locally unbiasedness for the parameter of interest as follows\footnote{Unbiased estimators are commonly discussed in standard textbooks, but locally unbiased estimators are not touched in introductory textbooks. 
We find the latter concept important when discussing the nuisance parameter problem in quantum estimation theory. 
To our knowledge, the concept of locally unbiasedness for the parameter of interest was introduced in \cite{jsNuipaper}.}. 
\begin{definition}%\label{def:lucond_thetaI}
An estimator $\hat{{\bm{\theta}}}_{\mathrm{I}}$ for the parameter of interest is called {\it locally unbiased for the parameters of interest} at ${\bm{\theta}}$, 
if, for $\forall i\in \{1,\dots,\dI\}$ and $\forall j\in \{1,\dots,d\}$, 
\be \label{lu_cond_class}
E_{\bm{\theta}}^{(n)}[{\hat{\theta}_i}(X^{(n)})]=\theta_i\ \mathrm{and}\  \frac{\del}{\del\theta_j}E_{\bm{\theta}}^{(n)}[{\hat{\theta}_i}(X^{(n)})]=\delta_{i,j}
\ee
are satisfied at a given point ${\bm{\theta}}$. 
\end{definition}
What is important here is an additional requirement that
$\frac{\del}{\del\theta_j}E_{\bm{\theta}}^{(n)}[{\hat{\theta}_i}(X^{(n)})]=0$ for $i=1,2,\dots,\dI$ and $j=\dI+1,\dI+2,\dots,d$. 
This condition can be trivially satisfied if a probability distribution is independent of the nuisance parameters. 
It is clear that if the estimator $\hat{{\bm{\theta}}}_{\mathrm{I}}$ is unbiased for the parameters interest, 
it is locally unbiased for the parameters of interest at any point. 

At first sight, the above definition \eqref{lu_cond_class} depends on the nuisance parameters explicitly.  
One might then expect that the concept of locally unbiased estimator for $\bm{\theta}_{\mathrm{I}}$ 
is not invariant under reparametrization of the nuisance parameters of the form \eqref{nui_change0}. 
The following lemma shows the above definition, in fact, does not depend on parametrization of the nuisance parameters. 
Its proof is given in Appendix \ref{sec:AppPr1}.
\begin{lemma}\label{lem_lucond}
If an estimator $\hat{{\bm{\theta}}}_{\mathrm{I}}$ is locally unbiased for $\bm{\theta}_{\mathrm{I}}$ at ${\bm{\theta}}$, 
then it is also locally unbiased for the new parametrization defined by an arbitrary transformation of the form \eqref{nui_change0}. 
That is, if two conditions \eqref{lu_cond_class} are satisfied, then the following conditions also hold. 
\be \label{lu_cond2} 
E_{\bm{\xi}}^{(n)}[{\hat{\theta}_i}(X^{(n)})]={\xi}_i\ \mathrm{and}\  \frac{\del}{\del{\xi}_j}E_{\bm{\xi}}^{(n)}[{\hat{\theta}_i}(X^{(n)})]=\delta_{i,j},
\ee
for $\forall i\in \{1,\dots,\dI\}$ and $\forall j\in \{1,\dots,d\}$. 
\end{lemma}

Here we present a sketched proof of inequalities \eqref{ccrineq1} and \eqref{ccrineq2}, leaving the detailed derivation to  Appendix \ref{sec:AppCstat1}.
When $\bm{\theta}_{\mathrm{N}}$ is known, the model $\cM$ is reduced 
to a $\dI$-dimensional model $\cM'$ without any nuisance parameter. 
Hence, we can apply the standard CR inequality to get inequality \eqref{ccrineq1}. 
When $\bm{\theta}_{\mathrm{N}}$ is not completely known, on the other hand, 
the model is $d$-dimensional.  
Consider an estimator ${\hat{{\bm{\theta}}}}$ for the all parameters ${\bm{\theta}}=(\bm{\theta}_{\mathrm{I}},\bm{\theta}_{\mathrm{N}})$ 
and denote its MSE matrix by $V^{(n)}_{\bm{\theta}}[{\hat{{\bm{\theta}}}}]$, 
then the CR inequality \eqref{crineqfull} for this $d$-parameter model holds for any locally unbiased estimator 
${\hat{{\bm{\theta}}}}=(\hat{{\bm{\theta}}}_{\mathrm{I}},\hat{{\bm{\theta}}}_{\mathrm{N}}):\ \cX^n\rightarrow \Theta=\Theta_{\mathrm{I}}\times\Theta_{\mathrm{N}}$. 
Let us decompose the MSE matrix as 
\be\label{Cmseblock} 
V^{(n)}_{\bm{\theta}}[{\hat{{\bm{\theta}}}}]=
\left(\begin{array}{cc}
V^{(n)}_{\bm{\theta};\mathrm{I},\mathrm{I}} & 
V^{(n)}_{\bm{\theta};\mathrm{I},\mathrm{N}} \\[0.1ex] 
V^{(n)}_{\bm{\theta};\mathrm{N},\mathrm{I}} 
& 
V^{(n)}_{\bm{\theta};\mathrm{N},\mathrm{N}} 
\end{array}\right). 
\ee
Then, applying the projection onto the subspace $\bm{\theta}_{\mathrm{I}}$ to the above matrix inequality, we obtain the desired result \eqref{ccrineq2}.

\subsection{Discussions on the classical Cram\'er-Rao inequality}
We discuss the above result concerning the CR inequalities \eqref{ccrineq1} and \eqref{ccrineq2} in detail. 
First, it is important to emphasize that two different scenarios 
deal with two different statistical models. In the presence of nuisance parameters, 
the best we can do is to estimate all parameters and hence the precision 
bound is set by the standard CR inequality for the $d$-parameter model. 

Second, when there exist nuisance parameters, the precision bound 
$J_{\bm{\theta}}^{\mathrm{I},{\mathrm{I}}}$ still depends on the unknown values of $\bm{\theta}_{\mathrm{N}}$. 
It is then necessary to eliminate the nuisance parameter from this expression. 
There are several strategies known in classical statistics; see, for example, \cite{basu77}.  
The simplest one is to marginalize the effect of nuisance parameter by 
taking expectation value of $J_{\bm{\theta}}^{\mathrm{I},{\mathrm{I}}}$ with respect to 
some prior distribution for the nuisance parameter $\bm{\theta}_{\mathrm{N}}$. 
The other is to adopt the worst case by calculating 
$\max_{\bm{\theta}_{\mathrm{N}}\in\Theta_\mathrm{N}}
J_{\bm{\theta}}^{\mathrm{I},{\mathrm{I}}}$. 

Third, the existence of a sequence of estimators attaining the 
equality in the asymptotic limit follows from the standard argument. 
When no nuisance parameter exists, the maximum likelihood estimator (MLE) 
for the parameter of interest $\bm{\theta}_{\mathrm{I}}$ saturates the bound. 
If we have some nuisance parameters in the model, we can 
also apply the MLE for the all parameters ${\bm{\theta}}=(\bm{\theta}_{\mathrm{I}},\bm{\theta}_{\mathrm{N}})$. 
This asymptotically saturates the CR inequality \eqref{crineqfull} as well as inequality \eqref{ccrineq2}. 

Fourth, we have the (asymptotically) achievable precision bound for the MSE matrix given by \eqref{ccrineq2}, 
but this bound is not practically useful in general. This is because one has to 
estimate {\it all parameters} in order to achieve it asymptotically by using MLE. 
It is usually very expensive to solve the likelihood equation in general. 
In particular, this is the case when the number of nuisance parameters are large compared to that of parameters of interest. 
Thus, there remain many problems to find efficient estimators in the presence of nuisance parameters. 
For example, \cite{basu77} lists ten different methods of dealing with this problem. 

Fifth, there exist several different derivations of the CR inequality \eqref{ccrineq2} in the presence of nuisance parameters. 
Based on each individual proof, we can give different interpretations of this result.  
In Appendix \ref{sec:AppCstat2}, we give two alternative proofs. A nontrivial part of this fact is that 
all three different methods lead to the same precision bound. 

Last, it is well known that the following matrix inequality holds. 
\begin{align}%\nonumber
J_{\bm{\theta}}^{\mathrm{I},\mathrm{I}}
= \left(J_{\bm{\theta};\mathrm{I},\mathrm{I}}
-J_{\bm{\theta};\mathrm{I},\mathrm{N}}
(J_{\bm{\theta};{\mathrm{N}},{\mathrm{N}}})^{-1}
J_{\bm{\theta};{\mathrm{N}},{\mathrm{I}}}\right)^{-1}
\ge ( J_{\bm{\theta};{\mathrm{I}},{\mathrm{I}}})^{-1}. \label{nuiineq}
\end{align}
Here the equality holds if and only if the off-diagonal block matrix vanishes, i.e., $J_{\bm{\theta};{\mathrm{I}},{\mathrm{N}}} =0$.  
When $J_{\bm{\theta};{\mathrm{I}},{\mathrm{N}}} =0$ holds at ${\bm{\theta}}$, we say that two sets of parameters $\bm{\theta}_{\mathrm{I}}$ and $\bm{\theta}_{\mathrm{N}}$ are {\it locally orthogonal with respect 
to the Fisher information matrix} at ${\bm{\theta}}$ or simply $\bm{\theta}_{\mathrm{I}}$ and $\bm{\theta}_{\mathrm{N}}$ are orthogonal at ${\bm{\theta}}$. 
When $J_{\bm{\theta};{\mathrm{I}},{\mathrm{N}}} =0$ holds for all ${\bm{\theta}} \in\Theta$, $\bm{\theta}_{\mathrm{I}}$ and $\bm{\theta}_{\mathrm{N}}$ are 
called {\it globally orthogonal}. In the next subsection, we discuss some of the consequences of parameter orthogonality. 

In summary, 
the MSE becomes worse in the presence of nuisance parameters when compared 
with the case of no nuisance parameters. We can regard the difference of two the bounds 
as the loss of information due to nuisance parameters. 
This quantity is defined by 
\be \label{losscinfo}
\Delta J_{\bm{\theta}}^{-1}:= 
J_{\bm{\theta}}^{\mathrm{I},\mathrm{I}}- 
( J_{\bm{\theta}; \mathrm{I},\mathrm{I}})^{-1}. 
\ee
When the values of $\Delta J_{\bm{\theta}}^{-1}$ is large (in the sense of matrix inequality), 
the effect of nuisance parameters is more noticeable. From the above mathematical 
fact, we have that no loss of information is possible if and only if two sets of 
parameters are globally orthogonal, i.e.,
\be
\Delta J_{\bm{\theta}}^{-1}=0 \ \Leftrightarrow\ 
J_{\bm{\theta};{\mathrm{I}},{\mathrm{N}}} =0
\ee
for all values of ${\bm{\theta}}=(\bm{\theta}_{\mathrm{I}},\bm{\theta}_{\mathrm{N}})\in\Theta$. 

\subsection{Parameter orthogonalization}\label{sec:NuiPO}
\subsubsection{Local orthogonalization}
For a given statistical model with nuisance parameter(s), 
the form of the precision bound appears different when the parameter of interest and 
the nuisance parameter are not orthogonal to each other, i.e., $\Delta J_{\bm{\theta}}^{-1}\neq 0$. 
Therefore, this orthogonality condition is a key ingredient 
when discussing parameter estimation problems with nuisance parameters. 
This was pointed out in the seminal paper by Cox and Reid whose result is 
briefly summarized below \cite{rc87,amari_comment}. 

Denote the $i$th partial derivative of the logarithmic likelihood function $\ell_{\bm{\theta}}(x)=\log p_{\bm{\theta}}(x)$, 
which is known as the {\it score function}, by 
\be\label{cscorefn}
u_{{\bm{\theta}};i}(x):=\frac{\del}{\del\theta_i}\ell_{\bm{\theta}}(x).
\ee 
Here after, we set the sample size $n=1$ to simplify notation. 
Then, the ($i,j$) component of the Fisher information matrix can be expressed as
$J_{{\bm \theta};i,j}=E_{\bm{\theta}}[u_{{\bm{\theta}};i}(X)u_{{\bm{\theta}};j}(X)]$.  
The local orthogonality condition 
$J_{{\bm \theta};i,j}=0$ for $i=1, \ldots, \dI$ and $j=\dI+1,\ldots, d$ 
is equivalent to 
the statistical independence of the two sets of random variables 
$u_{{\bm{\theta}};\mathrm{I}} (X)=(u_{{\bm{\theta}};1} (X),u_{{\bm{\theta}};2} (X),\dots,u_{{\bm{\theta}};\dI} (X))$ and 
$u_{{\bm{\theta}};\mathrm{N}} (X)=(u_{{\bm{\theta}};\dI+1} (X),u_{{\bm{\theta}};\dI+2} (X),\dots,u_{{\bm{\theta}};d} (X))$. As an example, consider a two-parameter model 
with $\theta_2$ a nuisance parameter. 
When $\theta_1$ and $\theta_2$ are orthogonal, two MLEs
$\hat\theta_1$ and $\hat{\theta}_2$ become independent when the experiment is repeated for $n\to\infty$ times.\  
As a consequence, the asymptotic error for $\theta_1$ becomes independent of knowing the true value of $\theta_2$ or not. 
A familiar example of this phenomenon is  the problem of estimating the mean value of a normal distribution 
without knowing its variance \cite{amari,lc,bnc,ANbook}. 

It is well known that any two sets of parameters can be made orthogonal at each point locally 
by an appropriate smooth invertible map from a given parametrization to the new parametrization: 
\be\nonumber%\label{nui_change}
{\bm{\theta}}=(\bm{\theta}_{\mathrm{I}},\bm{\theta}_{\mathrm{N}})\mapsto{\bm{\xi}}=({\bm{\xi}}_{\mathrm{I}},{\bm{\xi}}_{\mathrm{N}})\mbox{ s.t. }{\bm{\xi}}_{\mathrm{I}}=\bm{\theta}_{\mathrm{I}}. 
\ee 
Here, we stress that although the equation ${\bm{\xi}}_{\mathrm{I}}=\bm{\theta}_{\mathrm{I}}$ holds,
$u_{{\bm{\theta}};i}$ does not necessarily equal to $u_{{\bm{\xi}};i}$ even for $i=1, \ldots, \dI$. 
 That is, 
the partial derivative $\frac{\partial}{\partial \theta_i}|_{\bm{\theta}_{\mathrm{N}}}$
does not necessarily equal 
the partial derivative $\frac{\partial}{\partial \xi_i}|_{\bm{\xi}_{\mathrm{N}}}$
for $i=1, \ldots, \dI$. (As an example, see the transformation law for the partial derivatives \eqref{localortho_pd} below.) 
This statement about the local orthogonalization holds for an arbitrary model with any number of parameters \cite{amari}. 
For example, consider the following new parametrization for the nuisance parameters $\bm{\theta}_{\mathrm{N}}\mapsto{\bm{\xi}}_{\mathrm{N}}$: 
\begin{align} \label{localortho}
{\bm{\xi}}_{\mathrm{I}}&=\bm{\theta}_{\mathrm{I}},\\ \nonumber
{\bm{\xi}}_{\mathrm{N}}&=
\bm{\theta}_{\mathrm{N}}+
\left(J_{\bm{\theta}_0;{\mathrm{N}, \mathrm{N}}}\right)^{-1}
J_{\bm{\theta}_0;\mathrm{N},\mathrm{I}}
%\big|_{\bm{\theta}_{\mathrm{I}}=\bm{\theta}_{\mathrm{I}}(0)}
(\bm{\theta}_{\mathrm{I}}-\bm{\theta}_{\mathrm{I},0}), 
\end{align}
 where $\bm{\theta}_0=(\bm{\theta}_{\mathrm{I},0},\bm{\theta}_{\mathrm{N},0})$  is an arbitrary reference point. 
Under this coordinate transformation, 
we can work out that the partial derivatives $\frac{\partial}{\partial \theta_i}$ are transformed as follows.
\begin{align} \label{localortho_pd}
\frac{\partial}{\partial \bm{\xi}_{\mathrm{I}}}&=\frac{\partial}{\partial \bm{\theta}_{\mathrm{I}}}
-J_{\bm{\theta};\mathrm{I},\mathrm{N}}\left(J_{\bm{\theta};{\mathrm{N}, \mathrm{N}}}\right)^{-1}
\frac{\partial}{\partial \bm{\theta}_{\mathrm{N}}},\\
\frac{\partial}{\partial \bm{\xi}_{\mathrm{N}}}&=\frac{\partial}{\partial \bm{\theta}_{\mathrm{N}}},
\end{align}
where $\frac{\partial}{\partial \bm{\theta}_{\mathrm{I}}}=(\frac{\partial}{\partial \theta_1},
\frac{\partial}{\partial \theta_2},\ldots,\frac{\partial}{\partial \theta_{\dI}})^{\mathrm T}$ 
and $\frac{\partial}{\partial \bm{\theta}_{\mathrm{N}}}=(\frac{\partial}{\partial \theta_{\dI+1}},\frac{\partial}{\partial \theta_{\dI+2}},\ldots,\frac{\partial}{\partial \theta_d})^{\mathrm T}$. 
$\frac{\partial}{\partial \bm{\xi}_{\mathrm{I}}}$ and $\frac{\partial}{\partial \bm{\xi}_{\mathrm{N}}}$ are defined similarly. 
With this new parametrization, $\bm{\xi}_{\mathrm{I}}$ and $\bm{\xi}_{\mathrm{N}}$  are orthogonal at this point, 
i.e., $E_{\bm{\xi}}[u_{{\bm{\xi}};i}(X)u_{{\bm{\xi}};j}(X)]=0$ holds for $i=1,2,\ldots,\dI$ and $j=\dI+1,\ldots, d$. 

\subsubsection{Geometrical picture}\label{sec:AppCstat2-geo} 
It is worth emphasizing a simple geometrical picture of this local orthogonalization procedure \cite{ka84,amari,ak88,zr94,ANbook}, 
since we can immediately extend it to the quantum case. 
We define the tangent space of a statistical manifold $\cM$ at $\theta$, spanned by the score functions, by
\be
T_{\bm{\theta}}(\cM):=\mathrm{span}\{ u_{{\bm{\theta}};i}\}_{i=1}^d.
\ee
We introduce an inner product for the elements of the tangent space by 
\be
{\vin{u}{v}}_{\bm{\theta}}:=E_{\bm{\theta}}[u(X)v(X)],\quad u,v\in T_{\bm{\theta}}(\cM). 
\ee
Naturally, the Fisher information matrix $J_{{\bm \theta};i,j}$ can be regarded as 
a metric tensor of a Riemannian metric on $\cM$, since 
$J_{{\bm \theta};i,j}={\vin{u_{{\bm{\theta}};i}}{u_{{\bm{\theta}};j}}}_{\bm \theta}$ holds. 
In fact, Chentsov proved that the Fisher information matrix is the only unique Riemannian metric, 
which is invariant under the Markov mapping. (See \cite{ANbook}.)  
In the following, we will denote the $(i,j)$ component of the inverse of the Fisher information matrix by $J_{{\bm{\theta}}}^{i,j}$. 

Consider the set of score functions $(u_{{\bm{\theta}};\mathrm{I}} (X),u_{{\bm{\theta}};\mathrm{N}} (X))$, 
and introduce the linear subspace spanned by the score functions for the nuisance parameters by
\be
T_{\bm{\theta};\rm{N}}(\cM):=\mathrm{span}\{ u_{{\bm{\theta}};i}\}_{i=\dI+1}^d\subset T_{\bm{\theta}}(\cM).
\ee
Let $u_{{\bm{\theta}};\rm{N}}^i:=\sum_{j=\dI+1}^d  (J_{{\bm \theta};\rm{N},\rm{N}}^{-1})_{j,i}  u_{{\bm{\theta}};j}$ 
($i=\dI+1,\dI+2,\ldots,d$) be the dual basis for the tangent subspace $T_{\bm{\theta};\rm{N}}(\cM)$. 
The canonical projection ${\cal P}$ onto the tangent space at $\rm{\theta}$ for the nuisance parameters is 
given by 
\[
u\mapsto {\cal P}(u)=\sum_{i=\dI+1}^{d} \vin{u_{{\bm{\theta}};\rm{N}}^i}{u}_{\bm{\theta}}\,  u_{{\bm{\theta}};i} .
\]
By definition, the projection onto the orthogonal complement of the tangent space of the nuisance parameters 
is expressed as $u\mapsto u-{\cal P}(u)$ for $u\in T_{\bm{\theta}}(\cM)$. 
Therefore, this orthogonal projection of the score functions for the parameters of interest is 
\be\label{eq:effective_score}
\tilde{u}_{{\bm{\theta}};i}= u_{{\bm{\theta}};i}-\sum_{j,k=\dI+1}^{d} J_{\bm{\theta};i,j}\left(\left(J_{{\bm \theta};\rm{N},\rm{N}}\right)^{-1}\right)_{j,k} u_{{\bm{\theta}};k} 
\quad (i=1,2,\ldots,d_{\dI}). 
\ee
This projected score functions $\tilde{u}_{{\bm{\theta}};\mathrm{I}}
=\big( \tilde{u}_{{\bm{\theta}};1},\tilde{u}_{{\bm{\theta}};2},\ldots,\tilde{u}_{{\bm{\theta}};\dI}  \big)$ 
can be regarded as the effective score functions for the parameters of interest. 
It is worth noting that $\tilde{u}_{{\bm{\theta}};\mathrm{I}}$ can also be calculated directly by the coordinate transformation \eqref{localortho} as \eqref{localortho_pd}. 

The partial Fisher information matrix is nothing but 
the Fisher information matrix calculated by this effective score functions about the parameters of interest:
\[
J_{\bm{\theta}}({\mathrm{I}}|{\mathrm{N}})
=\Big[{\vin{\tilde{u}_{\bm{\theta};i} }{\tilde{u}_{\bm{\theta};j} }}_{\bm{\theta}} \Big], 
\] 
for $i,j=1,2,\ldots,\dI$. As will be demonstrated in section \ref{sec:QpoLocal}, we can construct the partial quantum Fisher information matrix for the parameters of interest in the same procedure.  

As an important application of geometrical picture, we discuss the case when 
there are (possibly infinitely) many nuisance parameters \cite{ka84,ak88}. 
In this case, tangent vectors for the nuisance parameters $u_{{\bm{\theta}};\mathrm{N}} (X)$ are no longer linearly independent. 
Nevertheless, we can derive the CR inequality for the parameters of interest as follows. 
Note that the maximum number for the nuisance parameters is $|\cX|-\dI-1$ for a regular statistical model on $\cX$. 
Take any linearly independent $|\cX|-\dI-1$ tangent vectors so that they form a basis for 
the tangent space $T_{\bm{\theta};\rm{N}}(\cM)$ for the nuisance parameters. 
We next calculate a Fisher information matrix for the nuisance parameters by using only these linearly independent tangent vectors. 
Let us denote it by $\widetilde{J}_{{\bm \theta};\rm{N},\rm{N}}$. Likewise, we also define the matrix 
$\widetilde{J}_{{\bm \theta};\rm{I},\rm{N}}$. Note that expressions of $\widetilde{J}_{{\bm \theta};\rm{N},\rm{N}}$ and 
$\widetilde{J}_{{\bm \theta};\rm{I},\rm{N}}$ depend on a particular choice of a set of tangent vectors. 
Now, we can define the effective score function by the formula \eqref{eq:effective_score} 
with $\widetilde{J}_{{\bm \theta};\rm{N},\rm{N}}$ and $\widetilde{J}_{{\bm \theta};\rm{I},\rm{N}}$. 
This then leads to the partial Fisher information matrix for the parameters of interest. 
\[
\widetilde{J}_{\bm{\theta}}({\mathrm{I}}|{\mathrm{N}})
=\Big[{\vin{\tilde{u}_{\bm{\theta};i} }{\tilde{u}_{\bm{\theta};j} }}_{\bm{\theta}} \Big]. 
\]
Due to non-uniqueness of the choice of tangent vectors, and hence $\widetilde{J}_{{\bm \theta};\rm{N},\rm{N}}$ and 
$\widetilde{J}_{{\bm \theta};\rm{I},\rm{N}}$ are not uniquely detemined. 
However, the partial Fisher information matrix is uniquely defined, 
since the second term in \eqref{eq:effective_score} is also determined by the canonical projection 
on to the tangent space for the nuisance parameters.  
It is now immediate to derive the desired CR inequality as before by $\widetilde{J}_{\bm{\theta}}({\mathrm{I}}|{\mathrm{N}})$. 

\subsubsection{Global orthogonalization}\label{subsubsec-globalorth}
Although local orthogonalization is always possible as was demonstrated above, 
it is impossible to find a globally orthogonal parametrization in general unless the model satisfies some conditions. 
A well-known exceptional case for such a globally orthogonal parametrization is the case 
when the number of parameter of interest is one ($\dI=1$), and the other parameters are all nuisance, 
that is, $\bm{\theta}_{\mathrm{I}}=\theta_1, \bm{\theta}_{\mathrm{N}}=(\theta_2,\dots,\theta_d)$.   

Consider a model with $d$ parameters and introduce a new parametrization ${\bm{\xi}}=({\xi}_1,{\xi}_2,\dots,{\xi}_d)$ 
such that $\theta_1={\xi}_1, \theta_2=\theta_2({\xi}_1,{\xi}_2), \theta_3=\theta_3({\xi}_1,{\xi}_2,{\xi}_3),\dots,  \theta_d=\theta_d({\bm{\xi}})$. 
Then, the Fisher information matrix in the new parametrization is 
\be \nonumber
J_{\bm{\xi}}=T_{\bm{\xi}} J_{\bm{\theta}} T_{\bm{\xi}}^{\mathrm{T}}\mbox{ with\ }T_{\bm{\xi}}= \left[ \frac{\del \theta_j}{\del {\xi}_{\alpha}}\right]_{j,{\alpha}\in\{1,2,\dots,d\}},
\ee
where the greek index is used for the new parametrization ${\bm{\xi}}$. 
From our assumption, the transformation matrix $T_{\bm{\xi}}$ takes the form of the upper triangle matrix:  
\be \nonumber
T_{\bm{\xi}}=\left(\begin{array}{cccc}
1 & t^2_1& \dots & t^d_1 \\0 & t^2_2 & \dots & t^d_2 \\ \vdots &  & \ddots & \vdots \\ 0 & 0 & \dots & t^d_d
\end{array}\right),
\ee
with $t^i_{\alpha}= \del \theta_i/\del {\xi}_{\alpha}$. 
We impose the orthogonality condition between ${\xi}_1=\theta_1$ and the rest ${\bm{\xi}}_{\mathrm{N}}=({\xi}_2,{\xi}_3,\dots,{\xi}_d)$ by setting 
\be \nonumber
J_{{\bm{\xi}};1,{\alpha}}=\sum_{i,j=1}^d t^i_1 
J_{{\bm{\theta}};i,j}t^j_{\alpha}=0\mbox{ for all }{\alpha}=2,3,\dots,d . 
\ee
Owing to the assumption of $t^i_{\alpha}=0$ for $i<{\alpha}$ and 
the smooth one-to-one mapping between ${\bm{\theta}}$ and ${\bm{\xi}}$, this is met by solving the $d-1$ coupled differential equations: 
\be \label{diffeq}
J_{{\bm{\theta}};1,i}+\sum_{j=2}^d J_{{\bm{\theta}};i,j}\frac{\del \theta_j}{\del {\xi}_1}=0\mbox{ for all }i=2,3,\dots,d. 
\ee
These equations in turn determine the forms of ${\del \theta_j}/{\del {\xi}_1}$ as functions of ${\bm{\xi}}$ through 
the original Fisher information matrix $J_{{\bm{\theta}}}$. 
Although the solution is not uniquely determined in general, we can always find a 
new parametrization that leads to the orthogonality between $\theta_1$ and the rest. 
When the parameters of interest can be made orthogonal to the nuisance parameters globally 
with a suitable reparametrization of nuisance parameters, 
we call this procedure as a {\it parameter orthogonalization method}. 
From above discussion, it works with certainty when there is only a single parameter of interest, i.e., $\dI=1$. 

Global parameter orthogonalization in statistics is appreciated when dealing with the nuisance parameter problem. 
One of the main advantages is that this enables us to construct an efficient estimator for the parameters 
of interest using the method of conditional inference with an ancillary statistic. 
When a given model mets a certain condition, it can be shown that the maximum likelihood equations for 
the parameters of interest can be separated from those of the nuisance parameters. Hence, we can 
completely ignore estimating the nuisance parameters without loosing any information. 
We refer to the original paper \cite{rc87} for more detail discussion and examples. 
In the next section, we provide a simple example to demonstrate the advantage of the parameter orthogonalization method. 

\subsection{Example}\label{sec:AppCstat3}
Let us consider the model of a random dice with three outcomes $\Omega=\{1,2,3\}$. 
We examine the following parametrization of this model;
\be \label{cmodel2}
\cM=\{ p_{\bm{\theta}}=(\theta_1,\theta_2,1-\theta_1-\theta_2)\,|\,{\bm{\theta}}=(\theta_1,\theta_2)\in\Theta \}. 
\ee
For convenience of notations, we define $\theta_3:=1-\theta_1-\theta_2$. 
The parameter region $\Theta$ is any open subset of $\Theta_0:=\{(\theta_1,\theta_2)\,|\,\theta_1>0,\theta_2>0,\theta_1+\theta_2<1 \}$. 
We take $\theta_2$ to be the nuisance parameter of this model and $\theta_1$ is the one of our interest, 
i.e., ${\theta}_{\mathrm{I}}=\theta_1$, ${\theta}_{\mathrm{N}}=\theta_2$.  
In this case, one cannot use the bound \eqref{ccrineq1} derived from the Fisher information for $\theta_1$ 
as the achievable bound. The correct one is the bound \eqref{ccrineq2} instead. 

When both parameters $\theta_1,\theta_2$ are unknown, we deal with the 
two-parameter model. If, on the other hand, the value of $\theta_2$ is known, 
the model is reduced to a single parameter model as 
\be \label{cmodel1}
\cM'=\{ p_{\bm{\theta}}=(\theta_1,\theta_2,1-\theta_1-\theta_2)\,|\,\theta_1\in\Theta_1 \subset\bbr\}. 
\ee
The Fisher information matrix and its inverse for the two-parameter model \eqref{cmodel2} are 
\begin{align*}
J_{\bm{\theta}}&=\frac{1}{\theta_1\theta_2\theta_3}\left(\begin{array}{cc}\theta_2(1-\theta_2) & \theta_1{ \theta_2} \\[1ex]
\theta_1{\theta_2} & {\theta_1}(1-{\theta_1})\end{array}\right)
,\\
J_{\bm{\theta}}^{-1}&=\left(\begin{array}{cc}\theta_1(1-\theta_1) & -\theta_1{ \theta_2} \\[1ex]
-\theta_1{\theta_2} & {\theta_2}(1-{\theta_2})\end{array}\right). 
\end{align*}
Therefore, the CR bound for estimating $\theta_1$ in the presence of the nuisance parameter $\theta_2$ is 
\be \label{ccrnui}
J_{\theta}^{\mathrm{I},\mathrm{I}}=\left( (J_{\bf \theta}^{-1}) \right)_{1,1}
=\theta_1(1-\theta_1), 
\ee
whereas, when $\theta_2$ is fixed, it becomes 
\be\label{ccrnonui}
(J_{\bm{\theta};{\mathrm{I}},{\mathrm{I}}})^{-1}
={(J_{{\bm{\theta}};1,1})^{-1}}=\frac{\theta_1\theta_3}{1-\theta_2}. 
\ee
The information loss due to the presence of this nuisance parameter $\theta_2$ 
is calculated as
\be \label{lossex}
\Delta J_{\bm{\theta}}^{-1}
=J_{\bm{\theta}}^{\mathrm{I},{\mathrm{I}}}
-(J_{\bm{\theta};{\mathrm{I}},{\mathrm{I}}})^{-1}=\frac{(\theta_1)^2{\theta_2}}{1-{\theta_2}}, 
\ee
and this is strictly positive. 

We list a few observations on this model. 
First, the CR bound \eqref{ccrnui} is independent of the value of the nuisance parameter $\theta_2$. 
This is a very special case and a model-dependent result. 
Second, it is obvious from \eqref{lossex} that the effect of the nuisance parameter is present 
since $\theta_1$ and $\theta_2$ are not orthogonal to each other. 
Third, the information loss $\Delta J_{\bm{\theta}}^{-1}$ is strictly positive since 
$\theta_1\theta_2\neq0$. It becomes larger as $\theta_2$ gets closer to $1$. 
%This corresponds to the case in which the occurrence of the outcome $1$ becomes less often. 
Last, this example fits into the application of the parameter orthogonalization described before 
and $\theta_1$ can be made globally orthogonal to the nuisance parameter by introducing 
a new parametrization. Thus, we can eliminate the effect of nuisance parameter. 
We shall work this out below. 

We introduce a new parameterization ${\bm{\xi}}=({\xi}_1,{\xi}_2)$ 
and assume that $\theta_1={\xi}_1$ and $\theta_2$ is a function of ${\bm{\xi}}$. 
The parameter orthogonalization condition \eqref{diffeq} is 
\be\nonumber
J_{{\bm{\theta}};1,2}+J_{{\bm{\theta}};2,2}\frac{\del\theta_2}{\del{\xi}_1}=0
\Leftrightarrow \theta_2+(1-\theta_1)\frac{\del\theta_2}{\del{\xi}_1}=0. 
\ee
A solution to this differential equation is found, for example, as 
\be\nonumber
\theta_2({\xi}_1,{\xi}_2)=(1-{\xi}_1) c({\xi}_2),
\ee
with $c(x)$ any smooth differentiable function that is not constant. 
We also assume that its derivative does not vanish for all ${\xi}_2$. 
The inverse of the Fisher information matrix in the new parametrization  becomes 
diagonal as
\begin{equation*}
J_{\bm{\xi}}^{-1}=\left(\begin{array}{cc}{\xi}_1(1-{\xi}_1)& 0 \\[1ex]
0&\ds \frac{c({\xi}_2)\big(1-c({\xi}_2)\big)}{(1-{\xi}_1)\dot{c}({\xi}_2)^2}  \end{array}\right), 
\end{equation*}
with $\dot{c}({\xi}_2)=d c({\xi}_2)/d{\xi}_2$. 
Since $\theta_1={\xi}_1$ by assumption, the corresponding CR bound in the ${\bm{\xi}}$ parametrization is 
\be
J_{\xi}^{\mathrm{I},{\mathrm{I}}}=(J_{\xi;{\mathrm{I}},{\mathrm{I}}})^{-1}={\xi}_1(1-{\xi}_1)=\theta_1(1-\theta_1). 
\ee

A practical advantage using the parameter orthogonalization is 
when one tries to solve the MLE equation. 
For a given string  of data $x^n=x_1x_2\dots x_n$, let us denote by 
$n_k$ ($k=1,2,3$) the number of data with value $x_k$. By definition, $n=n_1+n_2+n_3$. 
In the original parametrization ${\bm{\theta}}=(\theta_1,\theta_2)$, 
one has to solve the coupled MLE equations $\del/\del\theta_i [\sum_{k=1,2,3}n_k \log p_{\bm{\theta}}(k)]=0$ 
for $i=1,2$. If we work in the new parametrization ${\bm{\xi}}$, which diagonalizes the Fisher information matrix, 
one only needs to solve a single MLE equation $\del/\del{\xi}_1[ \sum_{k=1,2,3}n_k \log p_{\bm{\xi}}(k)]=0$ 
to infer the value of ${\xi}_1=\theta_1$. This is because this equation is independent 
of the nuisance parameter ${\xi}_2$. In other words, one can completely ignore 
the value of ${\xi}_2$. 
This simple example shows that the parameter orthogonalization procedure 
provides an efficient way of constructing a good estimator that concerns only the parameter of interest. 
When the model contains many nuisance parameters, we can greatly simplify 
the MLE equation to obtain the MLE for the parameter of interest. 
The parameter orthogonalization method plays a pivotal role in parameter estimation problems 
in the presence of nuisance parameters. 

\section{Quantum multi-parameter estimation problem} \label{sec:Qnui1}
\subsection{Single copy setting}\label{sec:Qnui1-1}
In this subsection, we shall briefly summarize the result of quantum state estimation theory. 
We refer readers to books \cite{helstrom,holevo,ANbook,hayashi,petz} for more details. 

A {\it quantum system} is represented by a Hilbert space $\cH$. 
Let $\lofh$ be the set of all linear operators on $\cH$. 
A {\it quantum state} $\rho$ is a positive semi-definite matrix on $\cH$ with unit trace. 
The set of all quantum states on $\cH$ is denoted by $\sofh:=\{\rho\,|\,\rho\ge0,\tr{\rho}=1 \}$.
In particular, a state in a $d$-dimensional Hilbert space is often referred to as a \emph{qudit}.

A {\it measurement} $\Pi$ on a given quantum state $\rho$ is a nonnegative operator-valued function 
on $(\Omega,\cB)$ with $\cB$ a Borel set on $\Omega$. 
Let $\Pi$ be a function from $\cB$ to $\lofh$ such that 
\begin{align*}
&i)\ \Pi(\Omega)=I,\\
&ii)\ \Pi(B)\ge 0,\,\forall B\in\cB,\\
&iii)\ \Pi\big(\bigcup B_j\big)=\sum_j   \Pi(B_j)\mbox{ for any mutually disjoint $B_j\in\cB$},
\end{align*}
where $I$ is the identity operator on $\cH$. 
$\Pi$ is usually referred to as the positive operator-valued measure (POVM). 
When considering measurements with finite outcomes, 
we use $\cX=\Omega=\{1,2,\dots, |\cX|\}$. The corresponding POVM 
is a set of nonnegative matrices $\Pi=\{\Pi_x\}_{x\in\cX}$ satisfying the condition $\sum_{x\in\cX}\Pi_x=I$.   
For the continuous measurement case ($\cX=\bbr$), $\Pi_x$ satisfies $\int_{\cX}\Pi_x dx=I$ as a practical working rule. 
When a POVM consists of mutually orthogonal projectors, we call it a projection valued measure (PVM) 
or simply a projection measurement. 
The probability of getting an outcome $x$ when a POVM $\Pi$ is performed on $\rho$ is given by the \emph{Born rule}
\begin{equation}\label{born}
p_\rho(x|\Pi)=\tr{\rho\Pi_x}.
\end{equation}

A {\it quantum statistical model} or simply a {\it model} is defined by a parametric family of quantum states on $\cH$: 
\be
\cM:=\{\rho_{\bm{\theta}}\,|\,{\bm{\theta}}\in\Theta\}\subset \sofh, 
\ee
where $\Theta\subset\bbr^d$ is an open subset. 
As in the standard statistical problem, we implicitly assume necessary regularity conditions \footnote{To avoid mathematical subtleties, we need to impose regularity 
conditions for quantum statistical models. For example, 
a mapping $\theta\mapsto \rho_\theta$ is one-to-one and smoothness 
so that we can differentiate $\rho_\theta$ sufficiently many times. 
$\del\rho_\theta/\del\theta_i$ are also assumed to be linearly independent. 
We also need to be careful about 
the rank of quantum states. For the sake of clarity, 
we only consider full-rank states in this article. 
For problems in the pure-state model, see, for example, \cite{fn95,fn99}.}. 
Indeed, when the parametric space and the state family have a common group covariant symmetry,
the state estimation can be formulated based on the group symmetry \cite{holevo,H-group}.
In this review, we consider a different and more general formulation,
which also works without symmetry. 

A set of a measurement $\Pi$ and an estimator $\hat{\bm{\theta}}$, 
$\hat{\Pi}=\Pi \circ \hat{{\bm{\theta}}}^{-1}$\footnote{
In this notation, $\hat{\Pi}$ describes a POVM over 
$\Theta\subset\bbr^d$ so that 
$\hat{\Pi}(B)=\Pi (\hat{{\bm{\theta}}}^{-1}(B))$ for a subset $B \subset \Theta$.}, 
is called a {\it quantum estimator} or simply an {\it estimator}. 
We define the MSE matrix for the estimator $\hat{\Pi}$ by 
\begin{align}\nonumber
V_{{\bm{\theta}}}[\hat{\Pi}]
&=\left[ \sum_{x\in\cX} \tr{\rho_{\bm{\theta}}\Pi_x}({\hat{\theta}_i}(x)-\theta_i)({\hat{\theta}_j}(x)-\theta_j)  \right]\\
&=\left[ E_{\bm{\theta}}\big[({\hat{\theta}_i}(X)-\theta_i)({\hat{\theta}_j}(X)-\theta_j)|\Pi\big]  \right].
\end{align}
where $E_{\bm{\theta}}[f(X) |\Pi]$ is the expectation value of a random variable $f(X)$ 
with respect to the distribution $p_{\rho_{\bm{\theta}}}(x|\Pi)=\tr{\rho_{\bm{\theta}}\Pi_x}$. 
The aim of quantum parameter estimation is to 
find an optimal estimator $\hat{\Pi}=\Pi \circ \hat{{\bm{\theta}}}^{-1}$ such that the MSE matrix 
approaches the minima allowed by the laws of quantum theory and statistics.

We note that it is in general not possible to minimize the MSE matrix over all possible measurements 
in the sense of a matrix inequality. This kind of situation often happens in the theory of optimal design of experiments, 
where one wishes to minimize the inverse of the Fisher information matrix over design variables. 
See \cite{fedorov,pukelsheim,fh97,pp13,fl14} and \cite{gns19} in the context of quantum estimation theory. 
One of possible approaches to find the precision bound for the MSE matrix is to minimize the weighted trace of the MSE matrix: 
\be
V_{\bm{\theta}}[\hat{\Pi}|W,\cM]:= \Tr{WV_{{\bm{\theta}}}[\hat{\Pi}]}, 
\ee
for a given positive matrix $W$. Here, the matrix $W$ is called a {\it weight matrix} 
(also called a {\it utility matrix} or {\it loss matrix} in statistical literature)
and it represents a trade-off relation upon estimating different vector component of the parameter ${\bm{\theta}}$. 
For instance, the case $W=I_d$ (the $d\times d$ identity matrix) corresponds 
to minimizing the averaged variance of estimators. In the language of optimal design of experiments, 
this optimality is called the A-optimal design. We can similarly define other optimality functions 
to define optimal estimators \cite{fedorov,pukelsheim,fh97,pp13,fl14}. 

In passing, we should not forget other possible formulations of parameter estimation problems 
in a quantum system. A general formulation of the quantum decision theory was developed by 
Holevo \cite{holevo1973statistical} and Ozawa \cite{ozawa1980optimal}. 
Prior to Holevo's work, a quantum Bayesian estimation theory appeared in \cite{personick71}, 
and its content was reviewed in \cite[section 7.5]{whth07}. 
Over the last two decades, quantum Bayesian theory became popular in applications to quantum metrology. 
See, for examples, \cite{tanaka2006,teklu2009bayesian,teklu2010phase,brivio2010experimental,blume2010optimal,christandl2012reliable,tsang2012ziv,koyama2017minimax,teo2018bayesian,oh2018bayesian,quadeer2019minimax} for concrete cases. 

We mainly consider strictly positive weight matrices, i.e., $W>0$, 
although it is also possible to formulate the problem with a nonnegative weight matrix. 
As we will discuss in this paper, the role of the weight matrix is important 
when discussing the nuisance parameter problem in the quantum case. 

One of the main interests in the quantum estimation theory is to find the precision bound under a certain condition on estimators. 
An estimator $\hat{\Pi}$ is called {\it unbiased}, if the following condition holds for all ${\bm{\theta}}\in\Theta$:
\[
E_{\bm{\theta}}\big[{\hat{\theta}_i}(X)|\Pi\big]=\sum_{x\in\cX} {\hat{\theta}_i}(x) \tr{\rho_{{\bm{\theta}}}\Pi_x}=\theta_i \quad(\forall i=1,2,\dots,d).
\]
Usually, such an unbiased estimator does not exist. To relax the unbiasedness condition, 
we impose this condition on the neighborhood of a given point. 
An estimator $\hat{\Pi}$ is called {\it locally unbiased} at ${\bm{\theta}}$, if 
\begin{align}\label{local-unbiased-cond1}
E_{\bm{\theta}}\big[{\hat{\theta}_i}(X)|\Pi\big]&=\sum_{x\in\cX} {\hat{\theta}_i}(x) \tr{\rho_{{\bm{\theta}}}\Pi_x}=\theta_i, \\ 
\frac{\del}{\del\theta_j}E_{\bm{\theta}}\big[{\hat{\theta}_i}(X)|\Pi\big]&=\sum_{x\in\cX} {\hat{\theta}_i}(x)\tr{\frac{\del}{\del\theta_j}\rho_{{\bm{\theta}}}\Pi_x}=\delta_{i,j}, \label{local-unbiased-cond2}
\end{align}
are satisfied at ${\bm{\theta}}\in\Theta$ for all parameter indices $i,j\in\{1,2,\dots,d\}$. 
Note that this condition is to require the usual unbiasedness condition at a point ${\bm{\theta}}$ up to 
the first order in the Taylor expansion. 

As a quantum version of the score function,
we often focus on the SLD $L_{{\bm{\theta}};i}^{\rm S}$, which is defined as a Hermitian matrix to satisfy
\begin{align}
\frac{\del}{\del\theta_i}\rho_{{\bm{\theta}}}= 
\frac{1}{2}
\big(L_{{\bm{\theta}};i}^{\rm S}\rho_{{\bm{\theta}}}+ \rho_{{\bm{\theta}}}L_{{\bm{\theta}};i}^{\rm S}\big).
\label{DEFSLD}
\end{align}
The SLD Fisher information matrix $J_{\bm{\theta}}^{\rm S}$ is defined as
\begin{align}
J_{{\bm{\theta}};i,j}^{\rm S}:=\frac{1}{2}
\tr{L_{{\bm{\theta}};i}^{\rm S}
\big(L_{{\bm{\theta}};j}^{\rm S}\rho_{{\bm{\theta}}}+ \rho_{{\bm{\theta}}}L_{{\bm{\theta}};j}^{\rm S}\big)}.
\label{DEFSLDF}
\end{align}
Here, when $\rho_{\bm{\theta}}$ is strictly positive, 
the choice of Hermitian matrix $L_{{\bm{\theta}};i}^{\rm S}$ is unique.
Otherwise, it is not unique.
However, the definition of 
the SLD Fisher information matrix $J_{\bm{\theta}}^{\rm S}$ in \eqref{DEFSLDF}
does not depend on the choice of Hermitian matrix $L_{{\bm{\theta}};i}^{\rm S}$
under the condition \eqref{DEFSLD}.
Under the locally unbiasedness condition at ${\bm{\theta}}$, 
we have SLD CR inequality \cite{helstrom}
\begin{align}
V_{\bm{\theta}}[\hat{\Pi}]
\ge (J_{{\bm{\theta}}}^{\rm S})^{-1}.
\label{CRSLDF}
\end{align}
The proof is reviewed in Appendix \ref{PfCR}.
When we can choose SLDs $L_{{\bm{\theta}};i}^{\rm S}$ for $i=1, \ldots, d$ 
such that these SLDs $L_{{\bm{\theta}};i}^{\rm S}$ are commutative with each other,
the equality in \eqref{CRSLDF} can be achieved by a local unbiased estimator constructed by their simultaneous spectral decomposition.
In the choice of SLDs $L_{{\bm{\theta}};i}^{\rm S}$, 
extending the Hilbert space is allowed.
Otherwise, the equality in \eqref{CRSLDF} cannot be achieved.
Indeed, for a strictly positive density matrix $\rho_{\bm{\theta}}$, it is sufficient to check the commutativity of 
SLDs $L_{{\bm{\theta}};i}^{\rm S}$ without the extension of Hilbert space.
For the detail see Appendix \ref{PfCR}. 

As a typical case, we consider this case when
the SLD Fisher information matrix $J_{\bm{\theta}}^{\rm S}$ is diagonal with $d=2$. 
This condition can be satisfied at one point when 
we change the coordinate.  
If the measurement is chosen by using the spectral decomposition of 
$L_{{\bm{\theta}};1}^{\rm S}$, the first diagonal element of 
$V_{\bm{\theta}}[\hat{\Pi}]$ can attain the lower bound in \eqref{CRSLDF},
but the second diagonal element cannot attain the lower bound in general.
That is, the first and second diagonal elements satisfy a trade-off relation.
To handle this trade-off, 
we introduce 
the fundamental precision limit by
\be\label{qcrbound}
C_{\bm{\theta}}[W,\cM]:=\min_{\hat{\Pi}\mathrm{\,:l.u.at\,}{\bm{\theta}}}\Tr{WV_{\bm{\theta}}[\hat{\Pi}]}, 
\ee
where the minimization is carried out for all possible estimators under the locally unbiasedness condition, 
which is indicated by l.u.~at ${\bm{\theta}}$. 
In this paper, any bound for the weighted trace of the MSE matrix is referred to as the {\it CR type bound}. 
When a CR type bound equals to the fundamental precision limit $C_{\bm{\theta}}[W,\cM]$ as in \eqref{qcrbound}, it is called 
{\it most informative (MI)} in our discussion.  
In the following, we discuss some of CR type and MI bounds.
Taking weighted trace in \eqref{CRSLDF}, we have the following bound.
\begin{itemize}
\item The SLD CR bound, which is the MI bound for any one-parameter model \cite{helstrom}:
\be \label{Eq:sld_crbound}
C_{\bm{\theta}}^{\rm S}[W,\cM]:=\Tr{W ({\sldQFI})^{-1}},
\ee
where $\sldQFI$ denotes the SLD Fisher information matrix about the model $\cM$.
\end{itemize}

To characterize the non-commutativity, we introduce 
the right logarithmic derivative (RLD)
$L_{{\bm{\theta}};i}^{\rm R}$, which is defined as a matrix to satisfy
\begin{align}
\frac{\del}{\del\theta_i}\rho_{{\bm{\theta}}}= 
\rho_{{\bm{\theta}}}L_{{\bm{\theta}};i}^{\rm R}.
\label{DEFRLD}
\end{align}
The RLD Fisher information matrix $J_{\bm{\theta}}^{\rm R}$ is defined as
\begin{align}
J_{{\bm{\theta}};i,j}^{\rm{R}}:=
\tr{ (L_{{\bm{\theta}};i}^{\rm R})^\dagger
\rho_{{\bm{\theta}}}L_{{\bm{\theta}};j}^{\rm R}}.
\label{DEFRLDF}
\end{align}
Here, when $\rho_{\bm{\theta}}$ is strictly positive, 
the choice of the RLD $L_{{\bm{\theta}};i}^{\rm R}$ is unique.
Otherwise, it is not unique.
However, the definition of 
the RLD Fisher information matrix $J_{\bm{\theta}}^{\rm R}$ in \eqref{DEFRLDF}
does not depend on the choice of the RLD $L_{{\bm{\theta}};i}^{\rm R}$
under the condition \eqref{DEFRLD}.
Although the RLD Fisher information matrix $J_{\bm{\theta}}^{\rm R}$ is Hermitian,
it has imaginary off-diagonal elements beacuse the RLD $L_{{\bm{\theta}};i}^{\rm R}$ is not necessarily Hermitian. 
Under the locally unbiasedness condition at ${\bm{\theta}}$, 
we have the RLD CR inequality \cite{yl73}
\begin{align}
V_{\bm{\theta}}[\hat{\Pi}]
\ge (J_{{\bm{\theta}}}^{\rm R})^{-1}.
\label{CRRLDF}
\end{align}
The proof is reviewed in Appendix \ref{PfCR}. 
Handling the imaginary components of $J_{\bm{\theta}}^{\rm R}$ efficiently, 
we have the following bound.
\begin{itemize}
\item The RLD CR bound, which is MI for a Gaussian shift model \cite{yl73,holevo}: 
\be \label{Eq:rld_crbound}
C_{\bm{\theta}}^{\rm R}[W,\cM]:=\Tr{W\Re({\rldQFI})^{-1}}+\Tr{|W^{\frac12} {\Im} ({\rldQFI})^{-1} W^{\frac12} |},
\ee
where $|X|=\sqrt{X^\dagger X}$, $\Re X=(X+X^\dagger)/2$, and $\Im X=(X-X^*)/2\I$ denote 
the absolute value, the real, and the imaginary part of a linear operator $X\in\lofh$, respectively. 
Here, $\rldQFI$ denotes the RLD Fisher information matrix about the model $\cM$.
\end{itemize}

As a tighter bound than both bounds, we often consider the following bound; 
\begin{itemize}
\item The Holevo bound (also known as the Holevo CR bound) \cite{holevo}: 
\be\label{holevo-bound}
C_{\bm{\theta}}^H[W,\cM]:=
\min_{{\bm X}=(X_1,\ldots, X_d)}\Tr{W\Re Z_{\bm{\theta}}({\bm X})}+\Tr{|W^{\frac12} {\Im} Z_{\bm{\theta}}({\bm X}) W^{\frac12} |},
\ee
where the minimization takes the vector of Hermitian matrices ${\bm X}=(X_1,\ldots, X_d)$
to satisfy the condition
$\tr{\frac{\del}{\del\theta_j}\rho_{{\bm{\theta}}}X_i}=\delta_{i,j}$ for 
$i,j=1, \ldots, d$,
and $Z_{\bm{\theta}}({\bm X})$ is the Hermitian matrix whose ($i,j$) component is
$\tr {X_i \rho X_j}$.
For readers' convenience, we give the proof for the inequality
\be\label{Ho-ine}
C_{\bm{\theta}}[W,\cM]\ge C_{\bm{\theta}}^H[W,\cM]
\ee
in Appendix \ref{AC-4}.
Notice that the minimum \eqref{holevo-bound} is achieved 
when the vector of Hermitian matrices ${\bm X}$ satisfies
the condition $\tr{\rho_{{\bm{\theta}}}X_i}=0$ for $i=1, \ldots, d$.
When the model is composed of pure states,
the equality in inequality \eqref{Ho-ine} holds \cite{KM}.

\if0
\item The Nagaoka bound \cite{nagaoka89}, which is given only in the case with $d=2$ and is tighter than the Holevo bound.
\be\label{nagaoka-bound}
C_{\bm{\theta}}^N[W,\cM]:=
\min_{{\bm X}=(X_1, X_2)}
W_{1,1}\tr {X_1 \rho X_1}+W_{2,2}\tr {X_2 \rho X_2}
+2 W_{1,2}\tr{| \rho^{1/2} [X_1,X_2] \rho^{1/2} |},
\ee
where the minimization takes the vector of Hermitian matrices 
${\bm X}=(X_1, X_2)$ under the same condition as the Holevo bound.
\fi
\end{itemize}

Note this bound \eqref{qcrbound} in general depends on the value of parameter ${\bm{\theta}}$ and the choice of the weight matrix $W$. 
Let $\hat{\Pi}_{\mathrm{opt}}:=\argmin_{\hat{\Pi}}  \Tr{WV_{\bm{\theta}}[\hat{\Pi}]}$ be an optimal estimator attaining the minimum 
of the most informative bound \eqref{qcrbound}, 
then it is clear that this $\hat{\Pi}_{\mathrm{opt}}$ represents the best measurement and the estimator in the sense of 
the above optimization. That is, if somebody specifies the weight matrix $W$, 
we can always construct the best estimator $\hat{\Pi}_{\mathrm{opt}} $ that minimizes 
the weighted trace of the MSE. 

When considering positive semi-definite weight matrices, the most informative bound cannot be 
attained explicitly in general. In this case, we have
\be\label{qcrbound-2}
\underline{C}_{\,{\bm{\theta}}}[W,\cM]:=\inf_{\hat{\Pi}\mathrm{\,:l.u.at\,}{\bm{\theta}}}\Tr{WV_{\bm{\theta}}[\hat{\Pi}]}
\ee
for $W\ge0$. 
The difference from the bound \eqref{qcrbound} is that 
an optimal estimator may not be locally unbiased at ${\bm{\theta}}$ for low-rank matrices $W$.

Before we move to the discussion on the multiple-copy setting, 
we show an alternative expression for the most informative bound \eqref{qcrbound}, which is due to Nagaoka \cite{nagaoka89}. 
He proved that the above bound can alternatively be expressed as the following optimization.  
For a given quantum statistical model $\cM=\{\rho_{\bm{\theta}}|{\bm{\theta}}\in\Theta\}$, 
let us fix a POVM $\Pi=\{\Pi_x\}_{x\in\cX}$. Then, the probability 
distribution determined by measurement outcomes $p_{\bm{\theta}}(x|\Pi)=\tr{\rho_{\bm{\theta}}\Pi_x}$ 
defines a classical statistical model:  
\be
\cM(\Pi):=\{p_{\bm{\theta}}(\cdot|\Pi)\,|\,{\bm{\theta}}\in\Theta\}. 
\ee
If the resulting classical model is regular, we can calculate 
the Fisher information matrix $J_{\bm{\theta}}[\Pi]$ about this model, and the CR inequality 
holds for the MSE matrix. Therefore, one can show that \cite[Theorem 2]{nagaoka89} 
\be\label{MICRbound}
C_{\bm{\theta}}[W,\cM]=\min_{J \in {\cal J}_{\bm{\theta}}} \Tr{W J^{-1} }, 
\ee
where ${\cal J}_{\bm{\theta}}$ is the collection of all Fisher information matrices $J_{\bm{\theta}}[\Pi]$ that are associated with POVMs $\Pi$. 
It is important to note that the statistical model $\cM(\Pi)$ can violate regularity conditions for some POVM. 
Since $\rho_\theta$ satisfies a certain regularity condition,
the statistical model $\cM(\Pi)$ satisfies the differentiability.
But the Fisher information matrix $J_{\bm{\theta}}[\Pi]$ might not be full rank, i.e., might be singular. 
In this case, one cannot calculate the inverse directly. 
A standard treatment is to use the generalized inverse with some care \cite{sm01}. 
Alternatively, regularization techniques are often used in literature. 
In the above optimization in \eqref{MICRbound}, due to the positivity assumption of the weight matrix, 
we can automatically exclude 
POVMs with singular Fisher information matrices
because of the following reason.
If the Fisher information matrix $J$ is singular, 
$J^{-1}$ will be unbounded.
Since $W>0$, we have $\Tr{J^{-1}W}\to\infty$ and can be excluded from the minimization. 
That is, we minimize the weighted trace of the inverse of Fisher information matrix 
associated only to POVMs with non-singular Fisher information matrix, 
and thus their statistical models are regular.    

As an alternative way to see the precision limit, Gill and Masser \cite{GM00} considered maximization of the quantity $\max_{\Pi\mathrm{: POVM}} \Tr{(\sldQFI)^{-1} J_{\bm{\theta}}[\Pi] }$, which captures how close the measurement induced Fisher information matrix $J_{\bm{\theta}}[\Pi] $ is to $\sldQFI$.\
They showed that
\be\label{GMbound2}
\max_{\Pi\mathrm{: POVM}} \Tr{(\sldQFI)^{-1} J_{\bm{\theta}}[\Pi] }
\le \dim {\cal H}-1.
\ee
The above bound immediately implies that one can measure at most $\dim{\cal H}-1$ nontrivial observables simultaneously without disturbing each other: Indeed, when observables are measured without mutual disturbance, we have $J_{\bm{\theta}}[\Pi]=\sldQFI$. 
Combing this inequality with \eqref{MICRbound}, they showed that
\be\label{qcrbound3}
C_{\bm{\theta}}[W,\cM] \ge \frac{ (\Tr{( W^{-1/2} \sldQFI W^{-1/2})^{-1/2} })^2}{\dim {\cal H}-1}. 
\ee
In the qubit case, 
the lower bound \eqref{qcrbound3} equals the bound obtained by 
Nagaoka \cite{nagaoka89} for a two-parameter model ($d=2$) and the bound obtained by 
Hayashi \cite{hayashi97} for a three-parameter model ($d=3$).

\subsection{Multiple-copy setting}\label{sec:Qnui1-2}
An important remark regarding this ``optimal estimator" is 
that it depends on the unknown parameter value ${\bm{\theta}}$ in general, 
due to the structure of the above optimization problem. 
In other words, one has to perform these measurements to estimate 
unknown parameters by using unknown values ${\bm{\theta}}$. 
This contradictory fact creates a major opponent against the use of (locally) unbiased estimators in 
classical statistics. 
Here, we stress that methods of statistical inference 
provide an additional ingredient to overcome such a difficulty and to achieve 
bound \eqref{qcrbound} asymptotically.

To resolve this problem, we consider the multiple copy setting, where
one is given states of the $n$-fold form $\rho_{\bm{\theta}}^{\otimes n}$. 
That is, we consider the state family $\{\rho_{\bm{\theta}}^{\otimes n}| {\bm{\theta}}\in \Theta\}$.
In this case, our measurement is given as a POVM on the $n$-fold tensor product system ${\cal H}^{\otimes n}$.
In this case, we can consider three types of settings.
\begin{description}
\item[A1]
Repetitive strategy:
In the first setting, 
we can repeat the same measurement on each of the $n$ subsystems ${\cal H}$ in ${\cal H}^{\otimes n}$.
\item[A2]
Adaptive strategy:
In the second setting, 
we make individual measurements on each of the $n$ subsystems,
but each measurement can depend on previous measurement outcomes.
\item[A3]
Collective strategy:
In the third setting, any POVM on ${\cal H}^{\otimes n}$ is allowed.
Such a measurement is often called a collective measurement.
\end{description}

In these settings, the MSE matrix behaves as $O(1/n)$. 
In the first setting A1, once we fix the measurement $\Pi$ to be repetitively applied, 
the problem can be handled as the statistical inference under the probability distribution family 
$\{ p_{\bm \theta}(\cdot|\Pi)| {\bm{\theta}}\in \Theta \}$. 
In this case, we can optimize the classical data processing. 
It is known that the MLE has the optimal performance in the large $n$ asymptotics, 
where $n$ times of MSE matrix asymptotically equals the inverse of the Fisher information matrix of 
the above probability distribution family \cite{lc,bnc,ANbook,pp13}. 
However, from the practical viewpoint, the MLE requires large calculation complexity \cite{PRA.75.042108}, 
they often require the linear inversion method \cite{PRA.90.012115,npj.3.44} and other methods \cite{QSEbook}.

When the MLE is assumed as our classical data processing method, what remains is the optimization of the POVM. 
For the setting A1, the papers \cite{H98,LFGKC,ZH18} focus on the maximization \eqref{GMbound2}.
A POVM $\Pi$ is said to be Fisher symmetric when 
it attains the maximization \eqref{GMbound2} for any ${\bm{\theta}} \in \Theta$. 
They discussed the case when the state family is composed of all pure states on ${\cal H}$. 
The paper \cite{H98} showed the existence of a Fisher symmetric POVM, but the paper \cite{ZH18} showed 
the non-existence of a Fisher symmetric POVM when the number of outcomes is restricted to be finite.

In the second setting A2, an adaptive choice of measurement is allowed and 
such a choice is considered as an estimator represented by a POVM $\hat{\Pi}^{(n)}$ on ${\cal H}^{\otimes n}$, where 
the output is an element of $\Theta\subset\bbr^d$.  
For a given sequence of estimators $\{\hat{\Pi}^{(n)}\}_{n=1}^{\infty}$
and a weighted matrix $W$,
we focus on the rescaled error
\begin{align}
\lim_{n \to \infty} n \Tr{W V_{{\bm{\theta}}}[\hat{\Pi}^{(n)}]}.\label{MSESE}
\end{align}

In the context of quantum state estimation,
this setting was first addressed by Nagaoka \cite{nagaoka89-2}. 
He proposed a concrete method to choose the measurement in each step
by using the likelihood.
However, since its analysis is complicated,
the papers \cite{GM00,HM98,BNG00} focus on the two-step estimation method.
In this method, we divided the $n$ given copies into two groups.
Then, we apply ${\bm{\theta}}$-independent separate measurements on states in the first group. For qubits, for example, we can measure each of the three Pauli observables using one third of the copies in the first group.
Based on these outcomes, we get a tentative estimate of ${\bm{\theta}}$. Finally, based on the tentative estimate, we apply the optimal single-copy measurement to 
all copies in the second group.
We refer to this measurement as the two-step measurement with single-copy optimality.
It was shown \cite{GM00,HM98,BNG00} that 
the rescaled error \eqref{MSESE} of this estimator equals the most informative CR-type bound \eqref{qcrbound}.
Later, Fujiwara \cite{fujiwara06} showed the same fact when the 
sequence of estimators is given by the Nagoka method.
Other various types of adaptive schemes have been intensively studied recently. 
See, for example, \cite{stm12,oioyift12,mrdfbks13,ksrhhk13,hzxlg16,ooyft17} 
and a review paper \cite{zlwjn17} and references therein. 
When applying an adaptive measurement ${\Pi}^{(n)}$, we denote 
the Fisher information of the resulting classical model by $J_{\bm{\theta}}[{\Pi}^{(n)}]$.
Then, it was shown in \cite[Chapter 6]{H-book} that the normalized Fisher information matrix 
belongs to the set ${\cal J}_{\bm{\theta}}$, i.e.,
\begin{align}
J_{\bm{\theta}}[{\Pi}^{(n)}]/n\in {\cal J}_{\bm{\theta}}
\label{JinJ}
\end{align}
When we take the normalization into account,
this fact shows that the adaptive choice cannot improve the maximization in RHS of \eqref{MICRbound}.

In the above way, several adaptive strategies globally achieve \eqref{qcrbound}.
However, there is no guarantee whether they satisfy the locally unbiasedness conditions (\ref{local-unbiased-cond1}) and (\ref{local-unbiased-cond2}), which are used to derive the bound $V_{\bm{\theta}} [\hat{\Pi}|W,{\cal M}] \ge C_{\bm{\theta}}[W,{\cM}]$. 
To resolve this problem, we focus on the limiting distribution 
for a sequence of estimators $\{\hat{\Pi}^{(n)}\}_{n=0}^{\infty}$ as an alternative formulation. 
The limiting distribution family $\{P_{{\bm t},\bm{\theta}_0}\}_{{\bm t}\in  \mathbb{R}^d}$ at $\bm{\theta}_0$
is defined as 
\begin{align}
P_{\bm{\theta}_0,{\bm t}}(B):=\lim_{n\to \infty}
\tr {\rho^{\otimes n}_{\bm{\theta}_0+\frac{{\bm t}}{\sqrt{n}}} \hat{\Pi}^{(n)}( \{ \hat{{\bm{\theta}}} | 
(\hat{{\bm{\theta}}}-\bm{\theta}_0)\sqrt{n}- {\bm t} \in B
 \}  ) }
\end{align}
for any $B\subset\mathbb{R}^d$. Intuitively, as $(\hat{{\bm{\theta}}}-\bm{\theta}_0)$ is proportional to $1/\sqrt{n}$, $P_{{\bm t},\bm{\theta}_0}$ characterises the asymptotic behaviour of the proportionality constant in a local region near ${\bm t}$.
Then, we impose a covariance condition requiring $P_{\bm{\theta}_0,{\bm t}}$ to be invariant under tiny shifts.
Technically, the condition requires $P_{\bm{\theta}_0,{\bm t}}=P_{\bm{\theta}_0,{\bm 0}} $
for any ${\bm t} \in \mathbb{R}^d$,
which is called the local asymptotic covariance condition at $\bm{\theta}_0$
for a sequence of estimators $\{\hat{\Pi}^{(n)}\}_{n=0}^{\infty}$. 
It is difficult to evaluate the quantity \eqref{MSESE} under the local asymptotic covariance condition.
Instead, we focus on the covariance matrix of the limiting distribution $P_{\bm{\theta}_0,{\bm t}} $, 
which is denoted by $V_{\bm{\theta}_0}[\{\hat{\Pi}^{(n)}\}_{n=0}^{\infty}]$.
Using \cite[Lemma 20]{YCH18}, as shown in Appendix \ref{AC-3},
we can show that the covariance matrix 
is lower bounded by the limit of the normalized Fisher information of the resulting classical model, i.e.\ $V_{\bm{\theta}_0}[\{\hat{\Pi}^{(n)}\}_{n=0}^{\infty}]\ge \lim_{n\to\infty}J_{\bm{\theta}}[{\Pi}^{(n)}]/n$.
Combining it with the relation \eqref{JinJ}, we have the inequality
\begin{align}\label{WVCW}
\Tr{W V_{\bm{\theta}_0}[\{\hat{\Pi}^{(n)}\}_{n=0}^{\infty}] }
\ge C_{\bm{\theta}_0}[W,\cM].
\end{align}
Also, the sequence of the two-step measurements with single-copy optimality
satisfies the local asymptotic covariance condition.
Hence, under the framework of the local asymptotic covariance condition,
$C_{\bm{\theta}_0}[W,\cM]$ is optimal and
there exists a sequence of estimators to attain this bound at any point ${\bm{\theta}}$.
That is, 
$C_{\bm{\theta}_0}[W,\cM]$ is the optimal bound in the setting A2 (Adaptive strategy).

However, it requires  additional cost to realize an arbitrary measurement in the setting A2 like 
the two-step estimation method.
Indeed, such a measurement requires the choice of measurement based on the previous outcomes.
That is, such an adaptive control of measurement devices
needs a feedback control, which requires additional devices.
To avoid such an additional cost,
we often adopt the setting A1, which does not require such an adaptive choice of our measurements.

Next, we consider the third setting, in which any measurement ${\Pi}^{(n)}$ on
${\cal H}^{\otimes n}$ is allowed.
In this setting, when a sequence of estimators $\{\hat{\Pi}^{(n)}\}_{n=0}^{\infty}$
satisfies the local asymptotic covariance condition at ${\bm{\theta}}$, 
for any weighted matrix $W$,
we have the inequality \cite{YCH18}
\begin{align}\label{WVCW}
\Tr{W V_{{\bm{\theta}}}[\{\hat{\Pi}^{(n)}\}_{n=0}^{\infty}] }
\ge
C_{\bm{\theta}}^H[W ,\cM].
\end{align}
Further, combining the above idea of two-step method, 
the paper \cite{YCH18} showed the following under a suitable regularity condition for
a state family $\{\rho_{\bm{\theta}}| {\bm{\theta}}\in \Theta\}$.
For any family of weighted matrices $\{W_{\bm{\theta}}\}_{\bm{\theta}}$,
  there exists a sequence of estimators 
$\{\hat{\Pi}^{(n)}\}_{n=0}^{\infty}$ 
such that 
the relation $\Tr{W_{\bm{\theta}} V_{{\bm{\theta}}}[\{\hat{\Pi}^{(n)}\}_{n=0}^{\infty}] }
=C_{\bm{\theta}}^H[W_{\bm{\theta}} ,\cM]$ holds with any ${\bm{\theta}} \in \Theta$.
This fact shows that 
the Holevo bound $C_{\bm{\theta}}^H[W ,\cM]$ expresses the ultimate precision bound in the state estimation.
That is, 
the Holevo bound $C_{\bm{\theta}}^H[W ,\cM]$ is the optimal bound in the setting A3 (Collective strategy).

\subsection{Model characterization for quantum parametric models}\label{sec:model_class}
Before we move to the discussion on the nuisance parameter problem for quantum parametric models, 
we briefly discuss the characterization of models in the quantum case. 
As we emphasize in this review, the Holevo bound on the MSE matrix, which is optimal in many cases, involves an optimization and is not expressed directly in terms of information theoretic quantities like quantum Fisher informations. 
It is then important to find some conditions enabling us to write down the achievable bound with an explicit expression. 
Traditionally, there were several sufficient conditions known to derive the closed expression for the precision limit. 
In past, there were a few progresses in deriving several necessary and sufficient conditions together with 
geometrical characterizations of quantum parametric models \cite{hayashi}. 
In the recent paper \cite{js18_clmodel}, one of the authors developed a systematic and unified methodology to address the problem. 

First, let us introduce the super-operator $\cD{\rho}$ for a given state $\rho$, whose action on any $X\in\lofh$ is defined by the following operator equation: 
\begin{equation}\label{def:CommOp}
\rho X-X \rho= \I \rho \cD{\rho}(X)+\I \cD{\rho}(X) \rho.
\end{equation} 
The solution is unique if the state is full rank. 
This super-operator is called the commutation operator \cite{holevo}, which is defined at $\rho$. 

Second, given a quantum parametric model $\cM=\{\rho_{\bm{\theta}}\,|\,{\bm{\theta}}\in\Theta\}$, 
let us introduce the SLD tangent space spanned by the SLD operators:
\begin{equation}\label{sldT}
T_{\bm{\theta}}(\cM)=\mathrm{span}_\bbr \{\SLD{i}\}_{i=1}^d. 
\end{equation}
Clearly, $T_{\bm{\theta}}(\cM)$ is a vector subspace of $\lofhh$ containing only Hermitian operators. 
Holevo investigated a special class of models, known as the D-invariant model. 
A model is said to be \emph{D-invariant} at ${\bm{\theta}}$, if $T_{\bm{\theta}}(\cM)$ is an invariant subspace of $\cD{\rho_{\bm{\theta}}}$. 
Equivalently, $\cD{\rho_{\bm{\theta}}}(\SLD{i})\in T_{\bm{\theta}}(\cM)$ holds for all $i=1,2,\dots,d$. 

The seminal result is the following fact: 
When the model is D-invariant at all ${\bm{\theta}}$, then the Holevo bound is reduced to the RLD CR bound \eqref{Eq:rld_crbound}. 
In other words, the RLD CR bound is achievable. 
In fact, the converse statement is also true and hence we have \cite{js16}: 
\begin{lemma} 
The Holevo bound is identical to the RLD CR bound, if and only if the model is D-invariant.
\end{lemma}
This result established the statistical meaning of the D-invariant model. 
We can also derive several equivalent characterizations of the D-invariant models. 
We list some of these conditions in Appendix \ref{sec:Qstat_model}. 
Two important examples for the D-invariant models are: 
The quantum Gaussian shift model \cite{holevo,yl73} and the full-parameter model on finite-dimensional Hilbert spaces. 
The latter model is parametrized by $d=(\dim\cH)^2-1$ parameters. 

The property of D-invariance is useful even when our model is not D-invariant. 
Consider a D-invariant model ${\cal M}'$ that includes the original model ${\cal M}$. 
Although the Holevo bound in the original model ${\cal M}$ is given as the minimum value of 
$\Tr{W\Re Z_{\bm{\theta}}({\bm X})}+\Tr{|W^{\frac12} {\Im} Z_{\bm{\theta}}({\bm X}) W^{\frac12} |}$,
the choice of 
${\bm X}=(X_1,\ldots, X_d)$ can be restricted to the case when each $X_i$ is given as a linear sum of 
SLD operators of the D-invariant model ${\cal M}'$ \cite{HM08}.
That is, in order to calculate the Holevo bound, it is sufficient to consider the minimization under the D-invariant model ${\cal M}'$.

We next turn our attention to the SLD CR bound. 
It is clear that the SLD CR bound cannot be saturated in the single copy setting 
in general due to the non-commutativity of the SLD operators. 
One exceptional case is when all SLDs commute with each other. 
That is, there exists a set of SLDs $\{L_{\bm{\theta};i}\}_{\bm{\theta},i}$ such that 
$[\SLD{i}\,,\,L_{\bm{\theta}';j}^{\rm{S}}]=0$ for all $i,j=1,2,\dots,d$ and 
all ${\bf \theta}, {\bf \theta}'\in\Theta$. 
In this case, we say that the model is quasi-classical \cite{nagaoka87,ANbook}
\footnote{Quasi-classicality for the one-parameter model was first introduced in \cite{nagaoka87}. 
Its generalization to the general model and other equivalent characterization were reported in \cite[Chapter 7.4]{ANbook}}.  
An equivalent characterization of the quasi-classical model is 
the existence of mutually commuting Hermitian operators $M_{\bm{\theta};i}$ ($i=1,2,\ldots,d$) 
such that the family of states is expressed as 
\begin{align}
\rho_{\bm{\theta}}&=N(\bm{\theta})\rho_0N(\bm{\theta}),\\ \nonumber
N(\bm{\theta})&:=\exp\left[{\frac{1}{2}\sum_{i=1}^d\int_{\bm{\theta}_0}^{\bm{\theta}} M_{\bm{\theta}',i}d{\theta'}^i -\frac12\psi(\bm{\theta})}\right],\\ 
&\left[ M_{\bm{\theta};i}\,,\,M_{\bm{\theta}',j} \right]=0\quad\forall i,j,\ \forall \bm{\theta};\bm{\theta}'.  
\end{align}
Here $\bm{\theta}_0$ is an arbitrary reference point and $\psi(\bm{\theta})$ is a scalar function for a normalization of the state. 
As an important class of the quasi-classical model, we have a quantum version of the exponential family of probability distributions.  
Let $F_i$ ($i=1,2,\ldots,d$) be mutually commutative Hermitian operators ($\forall i,j,\ [F_i,F_j]=0$) 
and define the family by
\be \label{Qe-family}
\rho_{\bm{\theta}}=e^{ \frac{1}{2}\sum_{i=1}^d  F_{i}\theta_i -\frac12\psi(\bm{\theta}) }\rho_0 e^{ \frac{1}{2}\sum_{i=1}^d  F_{i}\theta_i -\frac12\psi(\bm{\theta}) },
\ee
where $\psi(\bm{\theta})=\log\left[ \tr{\rho_0 \exp[\sum_{i=1}^d  F_{i}\theta_i ]}\right]$. 
This family of quantum states is called the \emph{quantum exponential family}, 
which is a quantum version of the exponential family 
of probability distributions known in statistics \cite{lc,bnc,ANbook}. 
This quantum exponential family plays an important role 
when studying a geometrical aspect of quantum statistical models \cite{ANbook,hayashi}.

When the model is quasi-classical, we can diagonalize the SLDs simultaneously.  
Hence, there exists a PVM $\Pi$ such that 
the classical Fisher information matrix under the resulting distribution family 
$\{ p_{\bm \theta}(x|\Pi)| {\bm{\theta}}\in \Theta \}$
achieves the SLD Fisher information matrix
at all points in $\Theta$.
%This immediately enables us to construct an optimal PVM attaining the SLD CR at each point ${\bm{\theta}}$. 
Achievability of this bound is then established for 
the repetitive strategy (A1) with the maximum likelihood estimator
%the adaptive strategy 
as discussed in Section \ref{sec:Qnui1-2}. 
Also, this condition implies the existence of Fisher symmetric POVM.  Moreover, the converse statement is also true. The existence of a POVM achieving the SLD Fisher information matrix 
for all points $\bm{\theta}$ implies that the state family is quasi-classical. 

Beside the above quasi-classical model, there is an extreme case when one can saturate 
the SLD CR bound asymptotically. This condition was investigated by several authors \cite{js16,rjdd16,js18_clmodel}. If
%We say that a model is asymptotically classical if 
$\tr{\rho_{\bm{\theta}} [\SLD{i}\,,\,\SLD{j}]}=0$ holds for all $i,j=1,2,\dots,d$
at $\bm{\theta}$,
the SLD CR bound can be achieved in asymptotically in the setting A3. Hence, such a model is called {\it asymptotically classical} at $\bm{\theta}$.
Indeed, this definition does not depend on the choice of SLDs 
$L_{\bm{\theta};i}^{\rm{S}}$ 
because the quantity $\tr{\rho_{\bm{\theta}} 
[L_{\bm{\theta};i}^{\rm{S}},L_{\bm{\theta};j}^{\rm{S}}]}$
 does not depend on this choice. 
We then have the following result \cite{rjdd16}. 
\begin{lemma} 
The Holevo bound is identical to the SLD CR bound, if and only if the model is asymptotically classical.
\end{lemma}
Other equivalent conditions are listed in Appendix \ref{sec:Qstat_model}. 

Note that the D-invariant model and the asymptotically classical model are mutually exclusive in the following sense. 
Suppose that a model is D-invariant and at the same time asymptotically classical. 
Then, we can show that this model is classical, that is, the state $\rho_{\bm{\theta}}$ for ${\bm{\theta}}\in\Theta$ is represented 
by a diagonal matrix in some basis. We can also show that this is also equivalent to 
equivalence of the SLD and RLD Fisher information matrices. 
In \cite{js18_clmodel}, several equivalent characterizations of the classical model were derived. 
For our convenience, we state the following result \cite{js18_clmodel}. 
\begin{proposition}
For a given model $\cM=\{\rho_{\bm{\theta}}|{\bm{\theta}}\in\Theta\}$ composed of strictly positive density matrices, 
$\sldQFI=\rldQFI$ for all ${\bm{\theta}}\in\Theta$ holds if and only if 
the model is D-invariant and asymptotically classical. 
Further, this condition is equivalent to the case when the model is classical. 
\end{proposition}

Finally, when the model is generic in the sense that it is neither D-invariant nor asymptotically classical, 
we need to solve the optimization appearing in the definition of the Holevo bound. 
Although an analytical expression for the Holevo bound might not be derived, it is not so hard to evaluate numerically. 
For example, a semi-definite programing approach was employed to evaluate the Holevo bound numerically in \cite{PhysRevA.97.012106,afd19}. 
In \cite{js16}, a non-trivial closed expression was obtained for any two-parameter qubit model. 
There, the Holevo bound is expressed in terms of both the SLD and RLD Fisher information matrices as follows \cite{js16}. 
\be \label{Hbound_qubit}
C_{\bm \theta}^H[W,\cM]=
\begin{cases}
C_{\bm \theta}^{\rm R}[W,\cM]&\mbox{ for } B_{\bm \theta}[W]\ge0\\[2ex] 
C_{\bm \theta}^{\rm S}[W,\cM] 
+\frac14\ \frac{\left(\Tr{|W^{\frac12}\Im (J^{\rm R}_{\bm \theta})^{-1}W^{\frac12}|}\right)^2}{\Tr{W\left(\Re(J^{\rm R}_{\bm \theta})^{-1}- (J^{\rm S}_{\bm \theta})^{-1}\right)}}&\mbox{ for } B_{\bm \theta}[W]<0  
\end{cases},
\ee
where $B_{\bm \theta}[W]:=\Tr{W\left(\Re (J^{\rm R}_{\bm \theta})^{-1}-(J^{\rm S}_{\bm \theta})^{-1}\right)}
-\frac12 \Tr{|W^{\frac12}\Im (J^{\rm R}_{\bm \theta})^{-1}W^{\frac12}|}$. 

\section{Nuisance parameter problem in the quantum case}\label{sec:Qnui2}

\subsection{Formulation of the problem}
We now introduce a model with nuisance parameters for the quantum case. 
Consider a $d$-parameter model as before and divide the parameters into two groups, 
one consists of parameters of interest $\bm{\theta}_{\mathrm{I}}=(\theta_1,\theta_2,\dots,\theta_{\dI})$ 
and the other consists of nuisance parameters $\bm{\theta}_{\mathrm{N}}=(\theta_{\dI+1},\theta_{\dI+2},\dots,\theta_d)$. 
We thus have a family of quantum states parametrized by two different kinds of parameters:
\be
\cM=\{\rho_{\bm{\theta}}\,|\,{\bm{\theta}}=(\bm{\theta}_{\mathrm{I}},\bm{\theta}_{\mathrm{N}})\in\Theta\subset\bbr^d\}. 
\ee
Our goal is to perform a good measurement and then to infer the values of parameter of interest $\bm{\theta}_{\mathrm{I}}$. 
Let $\hat{\Pi}_{\mathrm{I}}=\Pi\circ\hat{{\bm{\theta}}}_{\mathrm{I}}^{-1}$ be an estimator for the parameter of interest 
and define its MSE matrix for the parameters of interest by 
\begin{align}\nonumber
V_{\bm{\theta};\mathrm{I}}[\hat{\Pi}_{\mathrm{I}}]
&=\left[ \sum_{x\in\cX} \tr{\rho_{\bm{\theta}}\Pi_x}({\hat{\theta}_i}(x)-\theta_i)({\hat{\theta}_j}(x)-\theta_j)  \right]\\
&=\left[ E_{\bm{\theta}}\big[({\hat{\theta}_i}(X)-\theta_i)({\hat{\theta}_j}(X)-\theta_j)|\Pi \big] \right],
\end{align}
where the matrix indices $i,j$ run from $1$ to $\dI$ (instead of $d$). 
Hence, the MSE matrix is a  $\dI\times \dI$ matrix. 
We wish to find the precision bound for the above MSE matrix for the parameter of interest 
under the locally unbiasedness condition. 

Upon dealing with the nuisance parameter problem, it is necessary 
to define the locally unbiasedness for a subset of parameters. 
(See also Appendix \ref{sec:AppCstat1}.)
Let us consider the two sets of parameters ${\bm{\theta}}=(\bm{\theta}_{\mathrm{I}},\bm{\theta}_{\mathrm{N}})$ and 
an estimator $\hat{{\bm{\theta}}}_{\mathrm{I}}=(\hat{{\theta}}_1,\dots,\hat{{\theta}}_{\dI})$ as before. 
An estimator $\hat{\Pi}_{\mathrm{I}}=\Pi\circ\hat{{\bm{\theta}}}_{\mathrm{I}}^{-1}$ for the parameter of interest is called {\it unbiased} 
for $\bm{\theta}_{\mathrm{I}}$, if the condition 
\be
E_{\bm{\theta}}[{\hat{\theta}_i}(X)|\Pi]=\theta_i,
\ee
holds for all $i=1,2,\dots,\dI$ and for all ${\bm{\theta}}\in\Theta$. 
Clearly, this condition of unbiasedness does not concern 
the estimate of the nuisance parameters. 

Next, we introduce the concept of locally unbiasedness for the parameter of interest as follows \cite{jsNuipaper}. 
\begin{definition}
An estimator $\hat{\Pi}_{\mathrm{I}}$ for the parameter of interest is locally unbiased for $\bm{\theta}_{\mathrm{I}}$ at ${\bm{\theta}}$, 
if, for $\forall i\in \{1,\dots,\dI\}$ and $\forall j\in \{1,\dots,d\}$, 
\be \label{lu_cond}
E_{\bm{\theta}}[{\hat{\theta}_i}(X)|\Pi]=\theta_i\ \mathrm{and}\  \frac{\del}{\del\theta_j}E_{\bm{\theta}}[{\hat{\theta}_i}(X)|\Pi]=\delta_{i,j}
\ee
are satisfied at a given point ${\bm{\theta}}$. 
\end{definition}
Just as in the classical case, we stress the importance of the requirement that
$\frac{\del}{\del\theta_j}E_{\bm{\theta}}[{\hat{\theta}_i}(X)|\Pi]=0$ for $i=1,2,\dots,\dI$ and $j=\dI+1,\dI+2,\dots,d$. 
This requirement can be trivially satisfied if a probability distribution from a POVM 
is independent of the nuisance parameters. But this can only happen in special cases. 
In general, a non-vanishing $\frac{\del}{\del\theta_j}E_{\bm{\theta}}[{\hat{{\bm{\theta}}}_{\mathrm{I}}}(X)|\Pi]$ 
(for $j=\dI+1,\dI+2,\dots,d$) affects the MSE bound for the parameters of interest. 
See the general inequality \eqref{mse_genineq} in Appendix \ref{sec:AppCstat1}. 

It is known that for a given regular statistical model, we can always construct a 
locally unbiased estimator at arbitrary point; see expression \eqref{lu_est} in Appendix \ref{sec:AppPr2}. 
We can extend this to the case with nuisance parameters as follows. 
Suppose we fix a POVM whose classical statistical model is not regular. 
In particular, we consider the case when the score functions for the nuisance parameters 
are not linearly independent, i.e., $\{\frac{\del}{\del \theta_i}\log p_{\bm{\theta}}(x)\}_{i=\dI+1,\dots,d}$ are linearly dependent. 
In this case, the Fisher information matrix is singular and is not invertible. 
Nevertheless, the following estimator is locally unbiased for $\bm{\theta}_{\mathrm{I}}=(\theta_1,\dots,\theta_{\dI})$: 
\be
\hat{\theta}_i(x)=\theta_i+\sum_{j=1}^{\dI} \left( (J_{\bm{\theta}}(\mathrm{I}|\mathrm{N})[\Pi]) ^{-1}\right)_{j,i} u_{\bm{\theta}_{\mathrm{I}};j}(x|M_*).
\ee
Here, $J_{\bm{\theta}}(\mathrm{I}|\mathrm{N})[\Pi]$ is the partial Fisher information of \eqref{cpFI} for 
the classical model upon performing a POVM $\Pi$. To evaluate this partial Fisher information, we can use the generalized inverse. 
$u_{\bm{\theta}_{\mathrm{I}},j}(x|M)$ ($j=1,2,\ldots,\dI$) are the effective score functions defined by \eqref{eqApp:effscore} in Appendix \ref{sec:AppCstat2}. $M_*=J_{\bm{\theta};{\mathrm{I}},{\mathrm{N}}} (J_{\bm{\theta};{\mathrm{N}},{\mathrm{N}}})^{-1}$ is a $\dI\times\dN$ matrix, which is an optimal choice.  

Just as in the classical case, the locally unbiasedness here is also robust under the change of variables. 
Following the same logic as in Lemma \ref{lem_lucond}, we can prove the lemma below \cite{jsNuipaper}.  
\begin{lemma}\label{lem_lucond2}
If an estimator $\hat{\Pi}_{\mathrm{I}}$ is locally unbiased for $\bm{\theta}_{\mathrm{I}}$ at ${\bm{\theta}}$, 
then it is also locally unbiased for the new parametrization defined by the transformation \eqref{nui_change0}. 
That is, if two conditions \eqref{lu_cond} are satisfied, then the following conditions also hold. 
\be \label{lu_cond2} 
E_{\bm{\xi}}[{\hat{\theta}_i}(X)|\Pi]={\xi}_i\ \mathrm{and}\  \frac{\del}{\del{\xi}_j}E_{\bm{\xi}}[{\hat{\theta}_i}(X)|\Pi]=\delta_{i,j},
\ee
for $\forall i\in \{1,\dots,\dI\}$ and $\forall j\in \{1,\dots,d\}$. 
\end{lemma}

Having introduced the locally unbiasedness condition for the parameter of interest, 
we define the most informative bound for the parameter of interest by the following optimization: 
\begin{definition}
For a given $\dI\times \dI$ weight matrix $W_{\mathrm{I}}>0$, the most informative bound about the parameter of interest is 
defined by
\be\label{qcrboundnui0}
C_{\bm{\theta};\mathrm{I}}[W_{\mathrm{I}},\cM]
:=\min_{\hat{\Pi}_{\mathrm{I}}\mathrm{\,:l.u.\,at\,}\bm{\theta} \mathrm{\, for \,} \bm{\theta}_{\mathrm{I}}} 
\Tr{W_{\mathrm{I}}V_{\bm{\theta};\rm{I}}[\hat{\Pi}_{\mathrm{I}}]}, 
\ee
where the condition for the minimization is such that estimators $\hat{\Pi}_{\mathrm{I}}$ are 
locally unbiased for $\bm{\theta}_{\mathrm{I}}$ at ${\bm{\theta}}$. 
\end{definition}

By taking into account the nuisance parameters in the derivation of \eqref{MICRbound} 
and the classical CR inequality \eqref{ccrineq2}, 
we can show that the following alternative expression holds \cite{jsNuipaper}. 
For readers' convenience, its derivation is given in Appendix \ref{sec:AppPr2}. 
\be\label{MICRbound2}
C_{\bm{\theta};\mathrm{I}}[W_{\mathrm{I}},\cM]
=\min_{\Pi\mathrm{: POVM}} \Tr{W_{\mathrm{I}} J_{\bm{\theta}}^{\mathrm{I},\mathrm{I}}[\Pi] }, 
\ee
where $J_{\bm{\theta}}^{{\mathrm{I}},{\mathrm{I}}}[\Pi] $ is the block sub-matrix of the inverse of the Fisher information 
matrix about the POVM $\Pi$ [see \eqref{fisherblock}].
In general, the above  minimization \eqref{MICRbound2} may be even harder than the optimization in the case of estimating all parameters. 

Likewise, we have the nuisance parameter version of the Holevo bound \cite{holevo}
for $C_{\bm{\theta};\mathrm{I}}[W_{\mathrm{I}},\cM]$ as follows \cite{YCH18}:
\be\label{N-holevo-bound}
C_{\bm{\theta};\mathrm{I}}^H[W_{\mathrm{I}},\cM]:=
\min_{{\bm X}=(X_1,\ldots, X_{\dI})}
\Tr{W_{\mathrm{I}}\Re Z_{\bm{\theta}}({\bm X})}+\Tr{|W_{\mathrm{I}}^{\frac12} {\Im} Z_{\bm{\theta}}({\bm X}) W_{\mathrm{I}}^{\frac12} |},
\ee
where the minimization takes the vector of Hermitian matrices ${\bm X}=(X_1,\ldots, X_{\dI})$ to satisfy the condition
$\tr{\frac{\del}{\del\theta_j}\rho_{{\bm{\theta}}}X_i}=\delta_{i,j}$
for $i=1,\ldots, \dI$ and $j=1, \ldots, d$, and $Z_{\bm{\theta}}({\bm X})$ is the Hermitian matrix whose ($i,j$) component is
$\tr {X_i \rho X_j}$.
Similar to \eqref{Ho-ine}, as shown in Appendix \ref{AC-4},
we have the inequality
\be \label{Ho-ine-2}
C_{\bm{\theta};\mathrm{I}}[W_{\mathrm{I}},\cM] \ge 
C_{\bm{\theta};\mathrm{I}}^H[W_{\mathrm{I}},\cM].
\ee

Further, we can consider the $n$-fold asymptotic setting similar to 
Section \ref{sec:Qnui1-2}.
In this case, we can consider the settings A2 and A3 in the same way.
The bound 
$C_{\bm{\theta};\mathrm{I}}[W_{\mathrm{I}},\cM] $
is the optimal bound in the setting A2 (Adaptive strategy).
Also, as shown in \cite[Theorem 8]{YCH18},
the Holevo bound $C_{\bm{\theta};\mathrm{I}}^H[W_{\mathrm{I}} ,\cM]$ is the optimal bound in the setting A3 (Collective strategy). 

As discussed in the classical case, we can define the information loss 
due to the presence of nuisance parameters for the quantum case \cite{jsNuipaper}. 
Consider the $d_{\rm I}$-parameter model $\cM'$ that is the submodel of the original $d$-parameter model $\cM$ with all nuisance parameters fixed to be $\bm{\theta}_{\rm N}$. 
Assume that we have a bound $C_{\bm{\theta};\mathrm{I}}[W_{\mathrm{I}},\cM']$ for this model, then the difference
\begin{equation}\label{qinfoloss}
\Delta C^H_{\bm{\theta};\mathrm{I}}[W_{\mathrm{I}}|\bm{\theta}_{\mathrm{N}}]
:=C^H_{\bm{\theta};\mathrm{I}}[W_{\mathrm{I}},\cM]
-C^H_{\bm{\theta};\mathrm{I}}[W_{\mathrm{I}},\cM'], 
\end{equation}
measures how much information we lose for not knowing the nuisance parameters. 
(For the single-copy setting, we can similarly define the information loss by
$ \Delta C_{\bm{\theta};\mathrm{I}}[W_{\mathrm{I}}|\bm{\theta}_{\rm N}]
:=C_{\bm{\theta};\mathrm{I}}[W_{\mathrm{I}},\cM] 
-C_{\bm{\theta};\mathrm{I}}[W_{\mathrm{I}},\cM']$.) 
Unlike the classical case, it is not obvious to derive the condition of $\Delta C^H_{\bm{\theta};\mathrm{I}}[W_{\mathrm{I}}|\bm{\theta}_{\mathrm{N}}]=0$ 
in terms of a given model and weight matrix $W_{\mathrm{I}}$. 
Another difference is that the orthogonal condition does not provide a direct consequence for the zero loss of information. 
Moreover, a precision bound is not in general expressed 
as a simple closed-form in terms of quantum Fisher information. 

%%%%%%%%%%%%%%%%%%%%%%%%%%%%%%%%%%%%%%%%%%
\subsection{Local Parameter orthogonalization in the quantum case} 
\label{sec:QpoLocal}
%\label{sec:Qpo}
In this section we shall examine the effect of local parameter orthogonalization in the quantum case.  To this end, we first rewrite the SLD and RLD Fisher information matrices in terms of inner products. 
We then define the concept of parameter orthogonality with respect to different quantum Fisher informations. 
At last, we derive the CR type bounds for the parameters of interest  
and list several important properties of the  local  parameter orthogonalization method. 

To discuss local parameter orthogonalization, we prepare several notations for   
logarithmic derivatives and quantum Fisher informations. 
For a given smooth family of quantum states $\{\rho_{{\bm{\theta}}}\}$ and any (bounded) linear operators $X,Y$ on $\cH$, 
we define the symmetric and right inner product, respectively, by 
\begin{align} \nonumber
\sldin{X}{Y}&:=\frac12\tr{\rho_{{\bm{\theta}}}(YX^\dagger+X^\dagger Y)},\\
\rldin{X}{Y}&:=\tr{\rho_{{\bm{\theta}}}YX^\dagger}, 
\end{align}
where $X^\dagger$ denotes the Hermitian conjugate of $X$.  

Using the SLDs and RLDs, $\SLD{i}$ and $\RLD{i}$, the SLD and RLD Fisher information matrices are expressed as 
\be\label{sldrld}
\sldQFI= \left[ \sldin{\SLD{i} }{\SLD{j}}\right],\quad 
\rldQFI= \left[ \rldin{\RLD{i} }{\RLD{j}}\right], 
\ee
respectively. 
It is convenient to introduce the following linear combinations of 
the logarithmic derivative operators:
\[ 
\SLDdual{i}:= \sum_{j=1}^d\sldQFIinv{j,i}\SLD{j},\quad 
\RLDdual{i}:=\sum_{j=1}^d\rldQFIinv{j,i}\RLD{j} , 
\] 
where $\sldQFIinv{j,i}$ and $\rldQFIinv{j,i}$ are components of 
the inverse of the SLD and RLD Fisher information matrices, respectively. 

By definition, $\{\SLDdual{1},\SLDdual{2},\dots,\SLDdual{d} \}$ form  a dual basis for 
the inner product space $\sldin{\cdot}{\cdot}$;  
$\sldin{\SLDdual{i}}{\SLD{j}}={\delta_{i,j}}$.
The same statement holds for the RLD case. 
We can also check that the inverses of the SLD and RLD Fisher information matrices are expressed as
\begin{align} \nonumber
(\sldQFI)^{-1}&=[\sldQFIinv{i,j}]\mbox{ with } \sldQFIinv{i,j}=\sldin{\SLDdual{i}}{\SLDdual{j}} ,\\ 
(\rldQFI)^{-1}&=[\rldQFIinv{i,j}]\mbox{ with } \rldQFIinv{i,j}=\rldin{\RLDdual{i}}{\RLDdual{j}}. 
\end{align}

%\subsubsection{Local parameter orthogonality} %\label{sec:Qpo3}
Let us consider the same partition of the parameter ${\bm{\theta}}$ as before, i.e., 
${\bm{\theta}}=(\bm{\theta}_{\mathrm{I}},\bm{\theta}_{\mathrm{N}})$ with $\bm{\theta}_{\mathrm{I}}=(\theta_1,\theta_2,\dots, \theta_{\dI})$ 
and $\bm{\theta}_{\mathrm{N}}=(\theta_{\dI+1},\theta_{\dI+2},\dots, \theta_d)$, and discuss local parameter orthogonality under this parametrization.
When compared with the classical case, we immediately notice that 
the concept of parameter orthogonality is not uniquely defined in the quantum case. 
One may get different orthogonality conditions when considering different quantum Fisher informations.
Interested readers can find in \cite{js15} a qubit model exhibiting this phenomenon. 

Let us first focus on the SLD Fisher information matrix case. 
In the following, we denote the $(i,j)$ components of the SLD Fisher information matrix and its inverse matrix by
\be\nonumber
\sldqfi{{\bm{\theta}};i,j}\mbox{ and }\sldQFIinv{i,j},  
\ee 
respectively.  We remind the readers that we will keep using the following notations.  
\be\nonumber %\label{sldfisherblock}
\sldQFI=\left(\begin{array}{cc}
J^{\mathrm{S}}_{\bm{\theta};\mathrm{I},\mathrm{I}}& 
\sldqfi{\bm{\theta};\mathrm{I},\mathrm{N}} \\[0ex] 
\sldqfi{\bm{\theta};\mathrm{N},\mathrm{I}}& 
\sldqfi{\bm{\theta};\mathrm{N},\mathrm{N}} 
\end{array}\right), \quad 
(\sldQFI)^{-1}=\left(\begin{array}{cc}
J_{\bm{\theta}}^{\mathrm{S};\mathrm{I},\mathrm{I}}&
J_{\bm{\theta}}^{\mathrm{S};\mathrm{I},\mathrm{N}}
 \\[0ex] 
J_{\bm{\theta}}^{\mathrm{S};\mathrm{N},\mathrm{I}}
&
J_{\bm{\theta}}^{\mathrm{S};\mathrm{N},\mathrm{N}}
\end{array}\right) .   
\ee
We say two parameter groups $\bm{\theta}_{\mathrm{I}}$ and $\bm{\theta}_{\mathrm{N}}$ to be {\it locally orthogonal at} ${\bm{\theta}} $ 
with respect to the SLD Fisher information 
if the SLD Fisher information matrix is block diagonal according to 
this parameter partition at ${\bm{\theta}}\in\Theta$, that is 
\be\nonumber
\sldqfi{{\bm{\theta}};i,j}=0,\  \forall i=1,2,\dots,\dI\mbox{ and } \forall j=\dI+1,\dI+2,\dots,d,  
\ee
hold at ${\bm{\theta}}\in\Theta$, or 
equivalently $\sldqfi{\bm{\theta};\mathrm{I},\mathrm{N}}=0$. 
When local orthogonality condition holds for all ${\bf \theta}\in\Theta$, ${\bf \theta}_{\rm I}$ and ${\bf \theta}_{\rm N}$ are said {\it globally orthogonal}. 
Similarly, the local and the global orthogonality with respect to the RLD Fisher information 
can be defined by replacing $\sldqfi{{\bm{\theta}};i,j}$ by $\rldqfi{{\bm{\theta}};i,j}$. 

%\subsubsection{Local parameter orthogonalization}
Following exactly the same manner as in the classical case, we define the 
{\it effective SLD} for the parameters of interest by the orthogonal projection:
\be\label{effSLD}
 \tilde{L}_{\bm{\theta};i}^{\rm{S}}:= \SLD{i}-\sum_{j,k=\dI+1}^{d} J_{\bm{\theta};i,j}^{\rm S} \left( (\sldqfi{\bm{\theta};\mathrm{N},\mathrm{N}})^{-1}  \right)_{j,k}  \SLD{k} 
\quad (i=1,2,\ldots,d_{\dI}), 
\ee
where the second term is the projection onto the SLD tangent space for the nuisance parameters 
with respect to the SLD inner product.  The conversion from 
the SLDs $\SLD{i}$ to 
the SLDs $\tilde{L}_{\bm{\theta};i}^{\rm{S}}$ with $i=1, \ldots, \dI$
is called {\it local parameter orthogonalization.} 
With these projected SLDs, we refer to the $\dI\times\dI$ matrix
\be\label{qpSLDinfo}
J^{\mathrm{S}}_{\bm{\theta}}({\mathrm{I}}|{\mathrm{N}})
:=\left[ \sldin{\tilde{L}_{\bm{\theta};i}^{\rm{S}}   }{ \tilde{L}_{\bm{\theta};j}^{\rm{S}}  } \right]
\ee
for $i,j=1,2,\ldots,\dI$ as the {\it partial SLD Fisher information}. 
As in the classical derivation, we obtain the following relation.  
\be
J^{\mathrm{S}}_{\bm{\theta}}
({\mathrm{I}}|{\mathrm{N}})
=(J_{\bm{\theta}}^{\mathrm{S};\mathrm{I},\mathrm{I}} )^{-1}
=\sldqfi{\bm{\theta};\mathrm{I},\mathrm{I}}-
\sldqfi{\bm{\theta};\mathrm{I},\mathrm{N}} 
(\sldqfi{\bm{\theta};\mathrm{N},\mathrm{N}})^{-1}
\sldqfi{\bm{\theta};\mathrm{N},\mathrm{I}}. 
\ee

It is straightforward to show that the partial SLD Fisher information matrix gives the 
CR inequality for the MSE matrix for the parameters of interest and the corresponding CR type bound. 
\begin{align}
V_{\bm{\theta};\mathrm{I}}[\hat{\Pi}_{\mathrm{I}}]&\ge
 J^{\mathrm{S}}_{\bm{\theta}}({\mathrm{I}}|{\mathrm{N}})^{-1},\\
\Tr{W_{\mathrm{I}} V_{\bm{\theta};\mathrm{I}}
[\hat{\Pi}_{\mathrm{I}}]}
&\ge  C^{\rm S}_{\bm{\theta};\mathrm{I}}[W_{\mathrm{I}},\cM] 
:=\Tr{ W_{\mathrm{I}} 
J^{\mathrm{S}}_{\bm{\theta}}({\mathrm{I}}|{\mathrm{N}})^{-1}}. \label{sldCR_nui}
\end{align}
Likewise, we can also work out the RLD case. 
Define the effective RLD operators by
\be\label{effRLD}
 \tilde{L}_{\bm{\theta};i}^{\rm{R}}:= \RLD{i}-\sum_{j,k=\dI+1}^{d} J_{\bm{\theta};i,j}^{\rm R}\left((\rldqfi{\bm{\theta};\mathrm{N},\mathrm{N}})^{-1}  \right)_{j,k}\RLD{k} 
\quad (i=1,2,\ldots,d_{\dI}), 
\ee
and the partial RLD Fisher information matrix by
\be
J^{\mathrm{R}}_{\bm{\theta}}({\mathrm{I}}|{\mathrm{N}}) 
:=\left[ \rldin{\tilde{L}_{\bm{\theta};i}^{\rm{R}}   }{ \tilde{L}_{\bm{\theta};j}^{\rm{R}}  } \right]. 
\ee
Then, we have
\begin{align}
V_{\bm{\theta};\mathrm{I}}[\hat{\Pi}_{\mathrm{I}}]&\ge 
J^{\mathrm{R}}_{\bm{\theta}}({\mathrm{I}}|{\mathrm{N}})^{-1},\\
 \Tr{W_{\mathrm{I}}V_{\bm{\theta};\mathrm{I}}[\hat{\Pi}_{\mathrm{I}}]}&\ge  
 C^{\rm R}_{\bm{\theta};\mathrm{I}}[W_{\mathrm{I}},\cM]\\
 C^{\rm R}_{\bm{\theta};\mathrm{I} }[W_{\mathrm{I}},\cM]
 &:=\Tr{ 
 W_{\mathrm{I}} \Re 
 J^{\mathrm{R}}_{\bm{\theta}}({\mathrm{I}}|{\mathrm{N}})^{-1} }
+\Tr{
|W_{\mathrm{I}} ^{\frac12} \Im 
J^{\mathrm{R}}_{\bm{\theta}}({\mathrm{I}}|{\mathrm{N}})^{-1} W_{\mathrm{I}}^{\frac12} |
}. \label{rldCR_nui}
\end{align}
It is worth pointing out that here the orthogonal projection to the tangent space for the nuisance parameters 
is defined with respect to the RLD inner product. 
In passing, we note that the method of orthogonal projection was utilized by a recent paper \cite{tsang19} in the context of semiparametric 
estimation of quantum states, where the number of nuisance parameters are infinite.

Regarding the partial SLD Fisher information matrix, the following property is important. 
The proof is given in Appendix \ref{sec:AppPO}. \\
\noindent{\it Property 1: The partial SLD Fisher information matrix under parameter change.}\\
The partial SLD Fisher information defined by \eqref{qpSLDinfo}: 
\[
J^{\mathrm{S}}_{\bm{\theta}}({\mathrm{I}}|{\mathrm{N}})
=\sldqfi{\bm{\theta};\mathrm{I},\mathrm{I}}-
\sldqfi{\bm{\theta};{\mathrm{I}},{\mathrm{N}}} 
{\big(J^{\mathrm{S}}_{\bm{\theta};{\mathrm{N}},{\mathrm{N}}}\big)}^{-1}
\sldqfi{\bm{\theta};{\mathrm{N}},{\mathrm{I}}}
\]
is invariant under any reparametrization of the nuisance parameters of the form \eref{nui_change0}
and is transformed as the same manner as the usual Fisher information matrix. 

\subsection{Estimating a function of parameters}\label{sec:Qfunction}
In this subsection, we show how to apply our formulation to derive the CR-type bound upon estimating a function of parameters. 
(See Subsection \ref{sec:AppCstat3} for the classical case.) 
We note that recent works \cite{YCH18,tsang19,gross2020one} addressed the case of estimating a scalar function of parameters. 
Although the derivation is straightforward, results in this subsection have not been reported in literature to our knowledge. 

Given a vector-valued function $\bm{g}(\bm{\theta}):=\left(g_1(\bm{\theta}), g_2(\bm{\theta}),\ldots, g_K(\bm{\theta})\right)$, 
suppose we are interested in estimating the value of this function. 
For mathematical simplicity, we assume that $K$ should be smaller or equal to the number of parameters $d$. 
$g_k(\bm{\theta})$ for all $k$ are also assumed to be differentiable and continuous. 
We are willing to find a good estimator $\hat{\Pi}_{\bm{g}}$ upon estimating $\bm{g}(\bm{\theta})$. 
[$\hat{\Pi}_{\bm{g}}=(\Pi,\hat{\bm{g}})$: A POVM $\Pi$ and an estimator $\hat{\bm{g}}=(\hat{g}_1, \hat{g}_2,\ldots, \hat{g}_K)$]. 
Let $V_{\bm{\theta}}[\hat{\bm{g}}]:=\left[E_{\bm \theta}[\left(\hat{g}_k(X)-g_k(\bm{\theta})\right) \left(\hat{g}_{k'}(X)-g_{k'}(\bm{\theta})\right) ]  \right]$ be the MSE matrix for estimating the vector-valued function. 
The objective here is to minimize the weighted trace of the MSE matrix,
\[
\Tr{ W_{\bm{g}}V_{\bm{\theta}}[\hat{\bm{g}}] },
\]
under an appropriate condition on the estimator $\hat{\Pi}_{\bm{g}}$. 
We now use the same argument to define the most informative bound \eqref{qcrboundnui0} 
together with the result in the classical case \ref{sec:AppCstat3}. 
We define the most informative bound for $\hat{\Pi}_{\bm{g}}$: 
\be\label{qcrbound_function}
C_{\bm{\theta};\bm{g}}[W_{\bm{g}},\cM]
:=\min_{\hat{\Pi}_{\bm{g}}\mathrm{\,:l.u.\,at\,}\bm{\theta} \mathrm{\, for \,} \bm{g}} 
\Tr{ W_{\bm{g}}V_{\bm{\theta}}[\hat{\bm{g}}] }, 
%\Tr{W_{\bm{g}} G_{\bm \theta} V_{\bm{\theta}}[\hat{\Pi}_{\bm{g}}]\left(G_{\bm \theta}\right)^{\rm T} }.  
\ee
where the weight matrix $W_{\bm g}$ is a $K\times K$ positive matrix. 
%$G_{\bm \theta}$ is the $K\times d$ rectangular matrix defined by
%\be
%G_{\bm \theta}:=\left[ \frac{\del g_k(\bm{\theta})}{\del \theta_i} \right], 
%\ee
%with the row index $k=1,2,\ldots,K$ and the column index $i=1,2,\ldots,d$. 
The minimization in this definition is constrained within the locally unbiased estimator for $\bm{g}(\bm{\theta})$. 
This is defined as follows. 
An estimator $\hat{\Pi}_{\bm{g}}$ for the function $\bm{g}$ is locally unbiased for $\bm{g}(\bm{\theta})$ at $\bm{\theta}$, 
if, for $\forall k\in \{1,\dots,K\}$ and $\forall i\in \{1,\dots,d\}$, 
\be \label{lu_cond_function}
E_{\bm{\theta}}[\hat{g}_k(X)|\Pi]=g_k(\bm{\theta})\ \mathrm{and}\  
\frac{\del}{\del\theta_i}E_{\bm{\theta}}[\hat{g}_k(X)|\Pi]=\frac{\del g_k(\bm{\theta})}{\del \theta_i}
\ee
are satisfied at a given point $\bm{\theta}$. 

With the above formulation of the problem, we can derive the SLD CR bound and the RLD CR bound 
for estimating a vector-valued function $\bm{g}(\bm{\theta})$. 
\begin{align} \label{sldcr_function}
C^{\rm S}_{\bm{\theta};\bm{g}}[W_{\bm g},\cM] 
&:=\Tr{ W_{\bm g} G_{\bm \theta}\,(J^{\mathrm{S}}_{\bm{\theta}})^{-1}\:  \left(G_{\bm \theta}\right)^{\rm T}},\\
C^{\rm R}_{\bm{\theta};\bm{g}}[W_{\bm g},\cM] 
&:=\Tr{ 
 W_{\bm g}G_{\bm \theta}\,\Re 
 (J^{\mathrm{R}}_{\bm{\theta}})^{-1} \:  \left(G_{\bm \theta}\right)^{\rm T}}
+\Tr{
\left|W_{\bm g} ^{\frac12} G_{\bm \theta}\,\Im 
(J^{\mathrm{R}}_{\bm{\theta}})^{-1}\:  \left(G_{\bm \theta}\right)^{\rm T} W_{\bm g}^{\frac12} \right|
},\label{rldcr_function}
\end{align}
where $G_{\bm \theta}$ is the $K\times d$ rectangular matrix defined by
\be
G_{\bm \theta}:=\left[ \frac{\del g_k(\bm{\theta})}{\del \theta_i} \right], 
\ee
with the row index $k=1,2,\ldots,K$ and the column index $i=1,2,\ldots,d$. 
They are lower bounds for the most informative bound, i.e., 
\begin{align}
C_{\bm{\theta};\bm{g}}[W_{\bm{g}},\cM]&\ge C^{\rm S}_{\bm{\theta};\bm{g}}[W_{\bm g},\cM] ,\\
C_{\bm{\theta};\bm{g}}[W_{\bm{g}},\cM]&\ge C^{\rm R}_{\bm{\theta};\bm{g}}[W_{\bm g},\cM],
\end{align}
hold. 

The Holevo bound can also be extended to the case of estimating a vector-valued function. 
Without detailed account on it, we only report the result: 
\be\label{N-holevo-bound_function}
C_{\bm{\theta};\rm{g}}^H[W_{\bm{g}},\cM]:=
\min_{{\bm X}=(X_1,\ldots, X_{K})}
\Tr{W_{\bm{g}}\Re  Z_{\bm{\theta}}({\bm X}) }+\Tr{\left|W_{\bm{g}}^{\frac12} {\Im} Z_{\bm{\theta}}({\bm X}) W_{\bm{g}}^{\frac12} \right|},
\ee
where the minimization takes the vector of Hermitian matrices ${\bm X}=(X_1,\ldots, X_{K})$ to satisfy the condition
$\tr{\frac{\del}{\del\theta_i}\rho_{{\bm{\theta}}}X_k}=\frac{\del g_k(\bm{\theta})}{\del \theta_i}$
for $k=1,\ldots, K$ and $i=1, \ldots, d$. $Z_{\bm{\theta}}({\bm X})$ is the $K\times K$ Hermitian matrix 
%\be
%Z_{\bm{\theta};\bm{g}}({\bm X}):= G_{\bm \theta} Z_{\bm{\theta}}({\bm X}) \left(G_{\bm \theta}\right)^{\rm T}  
%\ee
whose $(k,k')$ component is defined by $ \tr{X_{k}\rho_{\bm \theta} X_{k'} }$ as before. 
The existence of the minimum in \eqref{N-holevo-bound_function} 
will be shown as Remark \ref{remark2} in the end of the next subsection.

%%%% Additional remark added 
\begin{remark} \label{remark1}
We can show the inequality $C_{\bm{\theta};\rm{g}}^H[W_{\bm{g}},\cM]\le
2 C_{\bm{\theta};\rm{g}}^S[W_{\bm{g}},\cM]$
as follows. See \cite{Carollo_2019,Carollo_2020,tsang2019_trivial,albarelli2019upper,tsang_v6} for the related results.  
We choose ${\bm L}:=(\SLDdual{i})$.
Since 
$\Re  Z_{\bm{\theta}}({\bm L}) \ge
- i {\Im}  Z_{\bm{\theta}}({\bm L})$,
we have 
$\Tr{W_{\bm{g}}\Re  Z_{\bm{\theta}}({\bm L}) }
\ge \Tr{\left|W_{\bm{g}}^{\frac12} {\Im} Z_{\bm{\theta}}({\bm L}) W_{\bm{g}}^{\frac12} \right|}$.
Thus, the relation $\Re  Z_{\bm{\theta}}({\bm L}) = 
(\sldQFI)^{-1}$ yields 
\begin{align*}
& C_{\bm{\theta};\rm{g}}^H[W_{\bm{g}},\cM]\le
\Tr{W_{\bm{g}}\Re  Z_{\bm{\theta}}({\bm L}) } 
+ \Tr{\left|W_{\bm{g}}^{\frac12} {\Im} Z_{\bm{\theta}}({\bm L}) W_{\bm{g}}^{\frac12} \right|} \\
\le &
2 \Tr{W_{\bm{g}}\Re  Z_{\bm{\theta}}({\bm L}) }
\le
2 C_{\bm{\theta};\rm{g}}^S[W_{\bm{g}},\cM].
\end{align*}
\end{remark}

\subsection{Model characterization in the presence of nuisance parameters} 
The concepts of D-invariant, quasi-classical and asymptotically classical models in section \ref{sec:model_class} 
can be extended to a quantum statistical model with nuisance parameters  by using the concept of local parameter orthogonalization.
These concepts  provide characterization of the Holevo bound in the presence of  nuisance parameters \eqref{N-holevo-bound}. 
Since we are analyzing the local aspect of the quantum statistical model, 
we will focus on the effective quantum score functions such as the effective SLDs \eqref{effSLD} and RLDs \eqref{effRLD}. 
This is equivalent to analyzing a given model within the new parametrization of the form \eqref{localortho}. 
To our knowledge, results in this subsection are not reported in literature.  

We emphasize that these concepts defined below are independent of choice for parametrization of nuisance parameters 
due to Property 1 in section \ref{sec:QpoLocal}. Furthermore, the effective quantum score functions 
are transformed exactly same manner as the ordinary quantum score functions. Therefore, these definitions do not 
relay on the choice of score functions. 

A quantum model is called D-invariant for the parameters of interest at $\bm{\theta}$ if the SLD tangent subspace 
spanned by the effective SLDs is invariant under the commutation operator at $\bm{\theta}$. 
Mathematically, this condition is expressed as for all $i=1,2,\ldots,\dI$,
\be
{\cal D}_{\rho_{\bm{\theta}}}(\tilde{L}_{\bm{\theta};i}^{\rm{S}} )\in \mathrm{span} \{\tilde{L}_{\bm{\theta};i}^{\rm{S}}\}_{i=1}^{\dI}. 
\ee
When the model is D-invariant for the parameters of interest at any point $\bm{\theta}$, 
we simply say that it is D-invariant for the parameters of interest. 
Once we obtain locally orthogonal parametrization at $\bm{\theta}$,
the calculation of the Holevo bound can be done by ignoring the nuisance parameters, i.e., it is sufficient to discuss only the parameters of interest.
Therefore, applying the proof of Lemma \ref{DRH} to the parameters of 
interest under the locally orthogonal parametrization, we have the following lemma. 

\begin{lemma}\label{DRHN}
A model is D-invariant for the parameters of interest
if and only if
the Holevo bound $C_{\bm{\theta};\mathrm{I}}^H[W_{\mathrm{I}},\cM]$
in the presence of the nuisance parameters \eqref{N-holevo-bound} is identical to the RLD-CR bound $C_{\bm{\theta};\mathrm{I}}^{\rm{R}}[W_{\mathrm{I}},\cM]$\eqref{rldCR_nui}
for any weight matrix $W_{\mathrm{I}}>0$.
\end{lemma}

We next turn our attention to the effective SLDs.  
A quantum model is said quasi-classical for the parameters of interest 
if the effective SLDs commute with each other for any 
$\bm{\theta}$ and
$\bm{\theta}'\in\Theta$, i.e., the condition
\be
\big[\tilde{L}_{\bm{\theta};i}^{\rm{S}}\,,\, \tilde{L}_{\bm{\theta}',j}^{\rm{S}}\big]=0,  
\ee
holds for all $i,j=1,2,\ldots,\dI$ and for all $\bm{\theta},\bm{\theta}'\in\Theta$. 
In this case, we can construct a POVM attaining the partial SLD Fisher information matrix 
by diagonalizing the effective SLDs simultaneously. 

A quantum model is said asymptotically classical for the parameters of interest at $\bm{\theta}$ 
if the effective SLDs commute with each other on the support $\rho_{\bm{\theta}}$ at $\bm{\theta}\in\Theta$: 
\be
\tr{\rho_{\bm{\theta}} \big[\tilde{L}_{\bm{\theta};i}^{\rm{S}}\,,\, \tilde{L}_{\bm{\theta};j}^{\rm{S}}\big] }=0,  
\ee
hold for all $i,j=1,2,\ldots,\dI$. 
A model is said asymptotically classical for the parameters of interest, 
if the model is asymptotically classical at any point. 
Similar to Lemma \ref{DRHN},
applying the proof of Lemma \ref{DSH} to the parameters of 
interest under the locally orthogonal parametrization, we have the following lemma. 

\begin{lemma}\label{DSHN}
A model is asymptotically classical for the parameters of interest
if and only if
the Holevo bound $C_{\bm{\theta};\mathrm{I}}^H[W_{\mathrm{I}},\cM]$
in the presence of the nuisance parameters \eqref{N-holevo-bound} is identical to the SLD-CR bound 
$C_{\bm{\theta};\mathrm{I}}^{\rm{S}}[W_{\mathrm{I}},\cM]$\eqref{rldCR_nui}
for any weight matrix $W_{\mathrm{I}}>0$.
\end{lemma}

\begin{remark}\label{remark2}
The existence of the minimum in \eqref{N-holevo-bound_function} can be shown as follows.
The choice of ${\bm X} $ can be restricted into a compact set in the following way.
Since the objective function is continuous, the minimum exists.

We assume that $W_{\bm{g}}$ is the identity matrix. Otherwise,
we change the coordinate to satisfy this condition.
We choose the minimum D-invariant space including $\SLD{i}$ and additional basis $F_l$ of the minimum D-invariant space
such that $\sldin{F_l}{\SLD{i}}=0$ and $\sldin{F_l}{F_{l'}}=\delta_{l,l'}$,
where  the minimum D-invariant space is given as the orbit of 
the subspace spanned by $\SLD{i}$ with respect to the D operator.
Then, $X_i$ is written as $\SLDdual{i}+ \sum_l a_{l}^i F_l$ using the vector $\bm{a}^i=(a_l^i)$ with $i=1, \ldots, \dI$.
Thus, $(\Re  Z_{\bm{\theta}}({\bm X}))^{i,j}=
\sldQFIinv{i,j}+
\bm{a}^i\cdot \bm{a}^j$.
Hence,
\begin{align*}
&\Tr{W_{\bm{g}}\Re  Z_{\bm{\theta}}({\bm X}) }+
\Tr{\left|W_{\bm{g}}^{\frac12} {\Im} Z_{\bm{\theta}}({\bm X}) W_{\bm{g}}^{\frac12} \right|}
>
\Tr{W_{\bm{g}}\Re  Z_{\bm{\theta}}({\bm X}) } \\
=&\Tr{ (\sldQFI)^{-1}}
+\sum_{i=1}^{\dI}\|\bm{a}^i\|^2 \ge \sum_{i=1}^{\dI}\|\bm{a}^i\|^2 .
\end{align*}
Hence, when $\sum_{i=1}^{\dI}\|\bm{a}^i\|^2 > \Tr{W_{\bm{g}}\Re  Z_{\bm{\theta}}({\bm L}) } 
+ 
\Tr{\left|W_{\bm{g}}^{\frac12} {\Im} Z_{\bm{\theta}}({\bm L}) W_{\bm{g}}^{\frac12} \right|}$,
the vector ${\bm X} $ cannot realize the minimum.
Therefore, the choice of 
${\bm X} $ can be restricted in to the case with
$\sum_{i=1}^{\dI}\|\bm{a}^i\|^2 \le \Tr{W_{\bm{g}}\Re  Z_{\bm{\theta}}({\bm L}) } 
+ 
\Tr{\left|W_{\bm{g}}^{\frac12} {\Im} Z_{\bm{\theta}}({\bm L}) W_{\bm{g}}^{\frac12} \right|}$, which describes a compact set.
\end{remark}

\subsection{Global parameter orthogonalization} \label{sec:QpoGlobal}
 We next examine global parameter orthogonalization.  A parametrization is called {\it globally orthogonal} 
if it is locally orthogonal at any point. 
As discussed in the classical case, the existence of global parameter orthogonalization 
is possible only when a new parametrization allows 
the relation $J_{\bm \xi;\rm{I},\rm{N}}=0$ in a new parametrization under the condition $\bm{\theta}_{\rm I}=\bm{\xi}_{\rm I}$. This is equivalent to finding 
a solution to the coupled partial differential equations similar to (\ref{diffeq}). 
Otherwise,   parameter orthogonalization can only be done locally at each point.
However,  there always exists a globally orthogonal parametrization  when the parameter of interest is a single parameter. 
We demonstrate it for the SLD Fisher information case below. Assume that a $d$-parameter model is given 
and let us introduce a new parametrization of the given quantum state by ${\bm{\xi}}=({\xi}_1,{\xi}_2,\dots,{\xi}_d)$. 
We impose the same conditions as in the classical case: 
$\theta_1={\xi}_1, \theta_2=\theta_2({\xi}_1,{\xi}_2), \theta_3=\theta_3({\xi}_1,{\xi}_2,{\xi}_3),\dots,  \theta_d=\theta_d({\bm{\xi}})$. 
Combining this with the definition of SLD operators \eqref{DEFSLD}, we see that the new set of SLD operators 
is expressed as a linear combination as
\be\label{sldxi}
L^{\mathrm{S}}_{{\bm{\xi}};{\alpha}}=\sum_{i\ge \alpha}^d \frac{\del \theta_i}{\del {\xi}_{\alpha}}L^{\mathrm{S}}_{{\bm{\theta}};i}, 
\ee
where the same index convention is used, i.e., the greek letters for the parameter ${\bm{\xi}}$. 
Then, the SLD Fisher information matrix in the new parameterization becomes 
\be
\sldqfi{{\bm{\xi}}}=T_{\bm{\xi}} \sldQFI T_{\bm{\xi}}^{\mathrm{T}}\mbox{ with\ }T_{\bm{\xi}}= \left[ \frac{\del \theta_j}{\del {\xi}_{\alpha}}\right]_{j,{\alpha}\in\{1,2,\dots,d\}}, 
\ee
which transforms exactly in the same manner as the classical case. 
Imposing the orthogonality condition between $\theta_1={\xi}_1$ and the rest with respect to the SLD Fisher information, 
we have the following conditions:  
\be \label{qdiffeq}
\sldqfi{{\bm{\theta}};1,i}+\sum_{j=2}^d \sldqfi{{\bm{\theta}};i,j}
\frac{\del \theta_j}{\del {\xi}_1}=0\mbox{ for all }i=2,3,\dots,d. 
\ee
By solving these coupled differential equations, we can obtain a new parametrization of 
the state $\rho_{\bm{\xi}}$ in which ${\xi}_1=\theta_1$ is orthogonal to the rest of parameters 
$({\xi}_2,{\xi}_3,\dots,{\xi}_d)$ with respect to the SLD Fisher information matrix. 
The same procedure can be carried out for the RLD Fisher information matrix. 

We now list several properties of the global parameter orthogonalization when $d_{\rm I}=1$.  
The following results are new contributions of this review. 
Proofs are given in Appendix \ref{sec:AppPO}. \\
{\it Property 2:  After the global  parameter orthogonalization, the SLD operator for the parameter of interest in the new parametrization 
is expressed as }
\be\label{sldxi-theta}
L^{\mathrm{S}}_{{\bm{\xi}};1}=(\sldQFIinv{1,1})^{-1}\SLDdual{1}.
\ee

\noindent
{\it Property 3: The partial SLD Fisher information of the parameter of interest after the global parameter orthogonalization is preserved.}\\
Although the parameter orthogonalization method enables us to have the relation 
$J^{\mathrm{S};1,1}_{\bm{\xi}}=(\sldqfi{{\bm{\xi}};1,1})^{-1}$ in the new parameterization, 
it preserves the partial SLD Fisher information for the parameter of interest as 
\be
J^{\mathrm{S};1,1}_{\bm{\xi}}=\sldQFIinv{1,1}. 
\ee
That is, the precision limit for the parameter of interest does not change as should be. 
(See also Theorem \ref{thmAN2} in section \ref{sec:1para2}.)
 
We close this section with a few remarks. 
The parameter orthogonalization method in the quantum case seems to be a natural extension of the classical result. 
Indeed, local parameter orthogonalization presented in this paper are extremely important upon studying the nuisance parameter problem 
in the quantum case. However, benefits of the global parameter orthogonalization method is less visible so far in the quantum case. 
One of the main reasons is that an optimal POVM attaining the most informative bound is $\bm{\theta}$-dependent in general, unless the model satisfies a special condition. 
Therefore, local properties of quantum statistical models are more important than the global aspect. 
In section \ref{sec:1para5}, we will apply this method to discuss the case where we can completely 
ignore the effect of the nuisance parameters.

%%%%%%%%%%%%%%%%%%%%%%%%%%%%%%%%%%%%%%%%%%
\section{One-parameter model with nuisance parameters}\label{sec:1para}
In this section we focus on models with a single parameter of interest in presence of nuisance parameter(s). 
This class of models is important when applying our method to quantum metrology in the presence of noise. 
It happens that this case is rather special, since the MSE bound and the optimal estimator 
have been known in literature for some time. In this section we discuss the general property of this class of problems, 
and then show the precision limit for the parameter of interest in the presence of nuisance parameters.  
In the following discussion, we consider the case for full-rank models on finite-dimensional Hilbert space. 

\subsection{General discussion}\label{sec:1para1}
\subsubsection{One-parameter model}
Let us start with a model with a single parameter, i.e., a scalar parameter $\theta$:
\be\label{1para_Qmodel}
\cM'=\{ \rho_{\theta}\,|\,\theta\in\Theta\subset\bbr\}. 
\ee 
It is known that the achievable MSE bound for the single parameter model is given by 
the SLD CR bound, which is the inverse of the SLD Fisher information, when there is no nuisance parameters. 

Let $\hat{\Pi}$ be a locally unbiased estimator at $\theta$ and denote 
its MSE by $V_{\theta}[\hat{\Pi}]$. The SLD CR bound is 
\be\label{sldCR_1para}
V_{\theta}[\hat{\Pi}]\ge (J^{\rm S}_{\theta})^{-1} ,
\ee
where no weight matrix appears since we are dealing with scalar quantities. 
An optimal estimator that attains the above bound is constructed as follows \cite{young,nagaoka87,bc94}. 
Consider the spectral decomposition of the SLD operator $L^{\mathrm{S}}_{\theta}$; 
\be\nonumber
L^{\mathrm{S}}_{\theta}=\sum_{x\in\cX}\lambda_x E_x,  
\ee
with the projector $E_x$ onto a subspace with the eigenvalue $\lambda_x$. 
We perform the projection measurement $\Pi=\{ E_x\,|\,x\in\cX\}$ 
and make an estimate, which is locally unbiased at $\theta$ by
\be \label{optest}
\hat{\theta}(x)={\theta}+(J^{\rm S}_{\theta})^{-1}\frac{d}{d\theta} \log p_{\theta}(x), 
\ee
where $p_{\theta}(x)=\tr{\rho_{\theta} E_x}$ is a probability distribution for the measurement outcomes. 
It is known that this optimal estimator depends on the unknown parameter $\theta$ in general, 
and hence the achievability of the SLD CR bound needs further discussions.  
Nagaoka derived the necessary and sufficient condition for the existence of 
an efficient estimator attaining the bound \eqref{sldCR_1para} uniformly in $\theta\in\Theta$ \cite{nagaoka87}. 
This condition is expressed as the following theorem \cite[Theorem 1]{nagaoka87}:
\begin{theorem}
For a one-parameter model \eqref{1para_Qmodel}, 
the SLD CR bound \eqref{sldCR_1para} is uniformly attained by some $\theta$-independent estimator 
$\hat{\Pi}$, if and only if two conditions i), ii) are satisfied. 
i) The model is parametrized in terms of $\xi\in\Xi$ as 
\[
\rho_{\xi}=\Exp{\frac12[ \xi F-\psi(\xi)]}\rho_0 \Exp{\frac12[ \xi F-\psi(\xi)]},
\] 
where $\psi(\xi)$ is a function of $\xi$, $F$ is an Hermitian operator on $\cH$, 
and $\rho_0$ is a $\xi$-independent state on $\cH$. 
ii) The parameter to be estimated is expressed as $\theta= \tr{F\rho_\xi}$. 
\end{theorem} 

This theorem states the necessary and sufficient condition for the existence of an efficient estimator attaining the SLD CR bound uniformly: i) A model is the quantum exponential family and ii) the parameter to be estimated is 
the expectation value of the observable $F$. Geometrically speaking, this is equivalent to 
the three conditions: The model is quasi-classical and e-autoparallel, and $\theta$ is an m-affine parameter. 
The above theorem can be generalized to the multiparameter setting as well. See \cite[Theorem 7.6]{ANbook}.

\subsection{One-parameter model with nuisance parameters}
Next, we provide a known result for a one-parameter estimation problem 
in the presence of nuisance parameters; see, for example, Chapter 7 of \cite{ANbook}. 
Consider a $d$-parameter model with $d-1$ nuisance parameters, 
i.e., ${\theta}_{\mathrm{I}}=\theta_1$ and $\bm{\theta}_{\mathrm{N}}=(\theta_2,\theta_3,\dots,\theta_d)$.  
Denote this model as 
\be\nonumber
\cM=\{ \rho_{{\bm{\theta}}=(\theta_1,\bm{\theta}_{\mathrm{N}})}\,|\,\theta_1\in\Theta_1,\bm{\theta}_{\mathrm{N}}\in\Theta_\mathrm{N}\}. 
\ee
We note that this model $\cM$ is reduced to the single parameter model $\cM'$ 
if all nuisance parameters are completely known. 
We stress that there are no general formulas for achievable bounds for this class of general models. 

A key result is now given for the one-parameter estimation problem in the presence of nuisance parameter(s). 
The following fundamental theorem also establishes the optimality of the SLD quantum Fisher information matrix \cite[equation (7.93)]{ANbook}. 
\begin{theorem}\label{thmAN}
Given a $d$-parameter regular model $\cM$, for each $d$-dimensional (column) vector $\Vec{v}\in\bbr^d$, 
the infimum of the MSE matrix in the direction of $\Vec{v}$ is 
\be \label{vCRbound}
\inf_{ \hat{\Pi}\mathrm{\,:l.u.at\,}{\bm{\theta}} } \Vec{v}^{\mathrm{T}} V_{\bm{\theta}}[\hat{\Pi}] \Vec{v}= \Vec{v}^{\mathrm{T}} (\sldQFI)^{-1}\Vec{v},
\ee
where $\sldQFI$ is the SLD quantum Fisher information matrix. 
An optimal measurement is given by a projection measurement 
about the linear combination of the SLD operators:
\be\label{1paraoptPOVM}
L^{\mathrm{S}}_{{\bm{\theta}};\Vec{v}}=\sum_{i,j=1}^d v_i 
J^{{\rm S};i,j}_{{\bm{\theta}}} \SLD{j}
= \sum_{i=1}^d v_i \SLDdual{i}
\ee 
\end{theorem}

We have several remarks regarding this theorem, although some of them are already discussed in the earlier sections. 
First, infimum is taken over all possible estimators which are locally unbiased for all parameters ${\bm{\theta}}=({\theta}_{\mathrm{I}},\bm{\theta}_{\mathrm{N}})$ at ${\bm{\theta}}$. 
Second, the optimal estimator $\hat{\Pi}_{\mathrm{opt}}$ in Theorem \ref{thmAN} depends on 
this particular direction $\Vec{v}$ in general. 
Third, this optimal estimator $\hat{\Pi}_{\mathrm{opt}}$ 
may not be locally unbiased for both the parameter of interest $\theta_1$ and the nuisance parameters $\Vec{\theta}_{\rm N}=(\theta_2,\ldots,\theta_d)$ at $\Vec{\theta}$. 
%\sout{\blue{but rather only about $\theta_{\Vec{v}}:=\sum_{i}v_i\theta_i$.}} 
%\dgreen{\sout{Therefore,}} 
In general, the bound \eqref{vCRbound} can be achieved by an adaptive strategy [see strategy (A2) in section \ref{sec:Qnui1-2}]. 
As a special case, a repetitive strategy [strategy (A1) in section \ref{sec:Qnui1-2}] can attain this bound when the optimal PVM is independent of $\Vec{\theta}$ (See Subsection \ref{subsec-example-ob} for an example.).  
%\sout{\blue{That is why the infimum rather than the minimum in Eq.~\eqref{vCRbound}.}} 
Fourth, the remaining $d-1$ parameters other than  $\theta_{\Vec{v}}:=\sum_{i}v_i\theta_i$ 
are to be regarded as nuisance parameters in this setting. 
Fifth, Theorem \ref{thmAN} can be understood as the rank-1 limit of 
the positive weight matrix as $W\to \Vec{v}\Vec{v}^{\mathrm{T}}$, which was discussed in Section \ref{sec:Qnui2}. 
Last, this theorem establishes the optimality of the SLD quantum Fisher information matrix 
for each direction given by $\Vec{v}$, and this provides an operational meaning of the SLD quantum Fisher information. 

The special case $\Vec{v}=(1,0,\dots,0)^{\mathrm{T}}$ is of particular importance when dealing with the one-parameter estimation problem 
in the presence of $d-1$ nuisance parameters. The impact of these nuisance parameters on the estimation is made apparent by comparing Eq.\ (\ref{sldCR_1para}) with Eq.\ (\ref{vCRbound}). 
The above theorem at first sight completely solves this case by providing 
an optimal estimator $\hat{\Pi}_{\mathrm{opt}}$. That is, the CR bound is 
the $(1,1)$ component of the inverse of SLD Fisher information matrix: $\sldQFIinv{1,1}$. 
However, 
there remains a question of achievability of this bound, 
since ${\Pi}_{\mathrm{opt}}$ depends on the unknown parameter ${\bm{\theta}}$ in general.  

Recently, there has been a growing trend in studying multiphase estimation and, in particular, distributed quantum metrology \cite{humphreys2013quantum,pezze2017optimal,altenburg2018multi,proctor2018multiparameter,eldredge2018optimal,ge2018distributed,swd19}.
A typical scenario is to consider estimating a unitary process on a network of $d$ spatially separate nodes, each described by a local unitary with an unknown parameter. Denoting by ${\bm{\theta}}=(\theta_1,\dots,\theta_d)$ the vector of all unknown parameters, the whole process is  described by the unitary $U({\bm{\theta}}):=\exp\left\{-i\sum_{k=1}^d\theta_k  H_k\right\}$,
where $\{H_k\}$ are the local generators of the evolution. 
The goal is to estimate the parameter $\theta_{\bf v}:={\bf v}\cdot{\bm{\theta}}$ that is a weighted sum of $\{\theta_k\}$, using a suitable probe state $\ket{\Psi_0}$. In such a setting, the relevant state model is $\{\ket{\Psi_{\bm{\theta}}}:=U({\bm{\theta}})\ket{\Psi_0}\}$, which can be characterized by one parameter of interest ($\theta_{\bf v}$) and $d-1$ nuisance parameters. Theorem \ref{thmAN} can be readily applied to obtain the estimation precision as well as the optimal measurement.
Similar as in other areas of quantum metrology, the main interest is whether the estimation precision can be enhanced when there is entanglement over different sites. For instance, it was shown in \cite{proctor2018multiparameter} that how big the advantage depends on the number of nuisance parameters. If there is only one parameter of interest, then it is often desired to use entangled probes.

\subsection{A refined version of Theorem \ref{thmAN}}\label{sec:1para2}
We can now prove that the bound in Theorem \ref{thmAN} can be achieved by a locally unbiased estimator for the parameter of interest 
corresponding to $\Vec{v}=(1,0,\dots,0)^{\mathrm{T}}$. 
Thereby, we obtain the precision limit for the single-copy setting. 
Note that we don't need to use the weight matrix for the parameter of interest in this special case, 
since we are minimizing a scalar quantity. By setting $W_{\rm I}=1$, we have the following theorem\footnote{To our knowledge, this theorem appears for the fist time in the context of the nuisance parameter problem.}.  
\begin{theorem}\label{thmAN2}
Given a $d$-parameter regular model $\cM$, suppose that 
we are interested in estimating the parameter $\theta_1$ in the presence of the nuisance parameters 
$\bm{\theta}_{\mathrm{N}}=(\theta_2,\dots,\theta_d)$. The achievable lower bound 
for the MSE about the parameter of interest 
$V_{\bm{\theta};\mathrm{I}}[\hat{\Pi}_{\mathrm{I}}]$ is given by 
\be \label{vCRbound2}
C_{\bm{\theta};\mathrm{I}}[W_{\rm I}=1,\cM]:=
\min_{\hat{\Pi}_{\mathrm{I}}\mathrm{\,:l.u.\,at\,}\bm{\theta} \mathrm{\, for \,} \bm{\theta}_{\mathrm{I}}} 
V_{\bm{\theta};\mathrm{I}}[\hat{\Pi}_{\mathrm{I}}] =\sldQFIinv{1,1}=\left(
J_{\bm{\theta}}({\mathrm{I}}|{\mathrm{N}})\right)^{-1}, 
\ee
where the minimization is taken over all locally unbiased estimators $\hat{\Pi}_{\mathrm{I}}$ for the parameter of interest at ${\bm{\theta}}$, $\sldQFIinv{1,1}$ is the $(1,1)$-th element of the inverse SLD matrix, and $
J_{\bm{\theta}}({\mathrm{I}}|{\mathrm{N}})$ is the partial SLD Fisher information \eref{qpSLDinfo}. 
An optimal measurement is given by a projection measurement about the operator:
\be\label{1paraoptPOVM}
\SLDdual{1}=\sum_{j=1}^d \sldQFIinv{j,1} \SLD{j}. 
\ee
\end{theorem}

We remark that this is a stronger variant of Theorem \ref{thmAN}. 
In the previous discussion, it was proven only for the infimum 
of the MSE about the parameter of interest $V_{\bm{\theta}_{\mathrm{I}}}[\hat{\Pi}]$ under 
the condition of locally unbiased estimators for all parameters ${\bm{\theta}}=({\theta}_{\mathrm{I}},\bm{\theta}_{\mathrm{N}})$. 
In Theorem \ref{thmAN2}, the condition is relaxed to unbiasedness for the parameter of interest [see \eqref{lu_cond}]. 
The proof for this theorem is given in Appendix \ref{sec:AppPr4}. With this theorem, we conclude that the partial SLD Fisher information 
$J^{\mathrm{S}}_{\bm{\theta}}({\mathrm{I}}|{\mathrm{N}})$ is the relevant quantity for the single-copy setting. 
 
From Theorem \ref{thmAN2}, we see that 
the case of one-parameter estimation problem with nuisance parameters is {\it essentially} a
one-parameter problem. The only difference here is that 
the partial SLD Fisher information 
$J^{\mathrm{S}}_{\bm{\theta}}({\mathrm{I}}|{\mathrm{N}})$ 
plays the fundamental role for the ultimate precision limit. 
This fact is understood transparently if we apply the parameter orthogonalization method. 
In the new parametrization, the parameter of interest is made orthogonal to the rest 
globally with respect to the SLD quantum Fisher information matrix. 
Hence, the nuisance parameters do not affect the precision limit. 
However, the optimal PVM attaining this limit {\it does} depend on the nuisance parameters in general. 
This means that the effects of the nuisance parameters should not be completely ignored. 
This point becomes significantly important for the finite sample case. 

Following the discussion in section \ref{sec:Qfunction}, we can extend our argument 
to derive the achievable bound upon estimating a scalar function of parameters $g(\bm{\theta})$. 
Given a smooth function $g$, define a column vector, 
\be
\bm{v}_{\bm \theta;g}:=\left(\frac{\del g(\bm{\theta})}{\del \theta_1},\frac{\del g(\bm{\theta})}{\del \theta_2},\ldots,\frac{\del g(\bm{\theta})}{\del \theta_d}\right)^{\rm T} .
\ee 
Then, we have the following result. See \cite[Section 9]{YCH18} for a rigorous proof.  
\begin{corollary}
The achievable precision limit for estimating $g(\bm{\theta})$ is given by
\begin{align}
C_{\bm{\theta};g}[W_{g}=1,\cM]&:=
\min_{\hat{\Pi}_{g}\mathrm{\,:l.u.\,at\,}\bm{\theta} \mathrm{\, for \,} g} 
V_{\bm{\theta};g}[\hat{\Pi}_{g}] \\
&=\bm{v}_{\bm \theta;g}^{\rm T} (\sldQFI)^{-1} \bm{v}_{\bm \theta;g}. 
\end{align}
An optimal estimator, which is locally unbiased at $\bm{\theta}$ for $g$, is given by the PVM about 
\be
L^{\mathrm{S}}_{{\bm{\theta}};\bm{v}_{\bm \theta;g}}=\sum_{i,j=1}^d \frac{\del g(\bm{\theta})}{\del \theta_i}
J^{{\rm S};i,j}_{{\bm{\theta}}} \SLD{j}
= \sum_{i=1}^d {v}_{\bm \theta;g,i} \SLDdual{i},
\ee 
where ${v}_{\bm \theta;g,i}$ denotes the $i$th component of the vector $\bm{v}_{\bm \theta;g}$. 
\end{corollary}
 
\subsection{Multi-copy setting}
Finally, let us discuss the one-parameter estimation problem in the presence of nuisance parameters 
for the multi-copy setting. This sets the ultimate precision limit upon estimating the parameter of interest. e
The RHS of Theorem \ref{thmAN2} of the $n$-copy case 
is just the $n$ times of the RHS of the one-copy case. 
This property shows that
any collective POVM on the $n$-copy case does not improve the bound 
$n C_{\bm{\theta};\mathrm{I}}[W_{\rm I}=1,\cM]$.
To see this property in a different viewpoint, we can also
explicitly evaluate the minimization in the Holevo bound \eref{N-holevo-bound} in the presence of nuisance parameters. 
Since there is no imaginary part appearing in this expression, we only need to evaluate the minimization
\be
C_{\bm{\theta};\mathrm{I}}^H[W_{\rm I}=1,\cM]:=\min_{X} \Tr{\rho_{\bm{\theta}}X^2 },
\ee
over all Hermitian matrices $X$ that satisfy: 
i) $\tr{\del_1 \rho_{\bm{\theta}}X }=1$ and ii) $\tr{\del_i \rho_{\bm{\theta}}X }=0$ ($i=2,3,\ldots,d$). 
This minimization can be solved explicitly as
\be
C_{\bm{\theta};\mathrm{I}}^H[W_{\rm I}=1,\cM]=\sldQFIinv{1,1}. 
\label{e92}
\ee
To show this, first substitute $X=\SLDdual{1}+M$ with an Hermitian matrix $M$ 
satisfying $\tr{\del_i \rho_{\bm{\theta}} M }=0$ ($i=1,2,\ldots,d$). 
The function to be minimized is then 
\be\label{Eq:optM}
\Tr{\rho_{\bm{\theta}}(\SLDdual{1}+M)^2 }=\Tr{\rho_{\bm{\theta}}(\SLDdual{1})^2}+\Tr{\rho_{\bm{\theta}}(M)^2}. 
\ee
Here, the cross terms vanish due to the condition 
$\tr{\del_i \rho_{\bm{\theta}} M }=\tr{\rho_{\bm{\theta}}\SLD{i}M+\SLD{i}\rho_{\bm{\theta}}M }=0$ 
and $\SLDdual{1}=\sum_{j=1}^d \sldQFIinv{1,j} \SLD{j}$. 
Therefore, the above minimization \eref{Eq:optM} yields 
$C_{\bm{\theta};\mathrm{I}}^H[W_{\rm I}=1,\cM]=\tr{\rho_{\bm{\theta}} (\SLDdual{1})^2}=\sldQFIinv{1,1}$ 
with the minimizer $M=0$. 

%%%%%%%%%%%%%%%%%%%%%%

\subsection{Special case}\label{sec:1para5}
We analyze the optimal POVM in Theorem \ref{thmAN2} 
and compare it with the optimal one for the case of without the nuisance parameters. 
Consider the spectral decomposition of two operators, $\SLD{1}$ and $\SLDdual{1}$: 
\begin{align}
\SLD{1}&=\ds\sum_{x\in\cX_1}\lambda_{\bm{\theta};1}(x) E_{\bm{\theta};1}(x),\\
\SLDdual{1}&=\ds\sum_{x\in\cX^1}\lambda_{\bm{\theta}}^{1}(x) E_{\bm{\theta}}^1(x). 
\end{align} 
Define the following projections measurements: 
\begin{align} 
\Pi_* &=\ds\{ E_{\bm{\theta};1}(x)\}_{x\in\cX_1},\\
\Pi_{\theta_{\rm{I}}}^*&=\ds\{ E_{\bm{\theta}}^1(x)\}_{x\in\cX^1}. 
\end{align}
%Since all the SLD operators $\SLD{i}$ are linearly independent under our assumptions, 
In generally, the optimal PVM $\Pi_* $ is no longer optimal for estimating $\theta_1$ in the presence of the nuisance parameters. This is because one faces two-different parametric models. 
It is straightforward to see that two measurements $\Pi_*$ and $\Pi_{\theta_{\rm{I}}}^*$ become identical 
at $\bm{\theta}$, if and only if $\SLD{1}$ and $\SLDdual{1}$ commute with each other. 
When the model is quasi-classical, a stronger commutation relation 
$[L_{\bm{\theta};1}^{\rm{S}}\,,\, L_{\bm{\theta'}}^{\rm{S};1}]=0$ holds 
for all $\bm{\theta},\bm{\theta}'$, since all SLDs commute with each other. 
Furthermore, SLDs are $\bm{\theta}$-independent. 
A non-trivial example, which is important, is when the SLD Fisher information matrix is block diagonal 
with respect to the partition $(\theta_1,\bm{\theta}_{\rm{N}})$.  
When this global parameter orthogonality condition is satisfied, ${(\sldqfi{\bm{\theta};1,1})}^{-1}=\sldQFIinv{1,1}$ holds. 

First, when all the nuisance parameters are known, we can perform 
the optimal PVM $\Pi_*$ whose Fisher information satisfies $J_{\bm{\theta}}[\Pi_*]=\sldqfi{\bm{\theta};1,1}$. 
Therefore, we can attain the SLD CR bound. 
In the presence of nuisance parameters, however, this PVM is no longer optimal in general. 
The optimal PVM for estimating the parameter of interest is $\Pi_{\theta_{\rm{I}}}^*$ according to Theorem \ref{thmAN2}. 
Since we have information loss \eref{qinfoloss} due to the nuisance parameters as
\begin{equation}\label{qinfoloss_1para}
\Delta C_{\bm{\theta};\mathrm{I}}[W_{\mathrm I}=1|\bm{\theta}_{\rm N}]
=\Delta C^H_{\bm{\theta};\mathrm{I}}[W_{\mathrm I}=1|\bm{\theta}_{\rm N}]
%[W_{\mathrm{I}}|\bm{\theta}_{\mathrm{N}}]
%\Delta C_{\theta_{\rm{I}}}
=\sldQFIinv{1,1}-(\sldqfi{\bm{\theta};1,1})^{-1}\ge0, 
\end{equation}
where the equality holds if and only if $\theta_{\rm{I}}=\theta_1$ and $\bm{\theta}_{\rm{N}}$ 
are orthogonal with respect to the SLD Fisher information matrix, the effect of nuisance parameters is not negligible. 

Applying the global parameter orthogonalization method in section \ref{sec:QpoGlobal}, 
we can always make $\theta_{\rm{I}}=\theta_1$ orthogonal to the rest $\bm{\xi}_{\rm{N}}$. 
Thus, by combining properties of global parameter orthogonalization method, we can show that 
the inverse of the partial SLD Fisher information $J^{\mathrm{S}}_{\bm{\theta}}(\mathrm{I}|{\mathrm{N}})$ is the precision limit. 
This also shows an alternative proof for Theorem \ref{thmAN2}. 

When we further consider three different estimation strategies discussed in section \ref{sec:Qnui1-2}, 
parameter dependence on the optimal PVM should also be examined. 
It is clear that the following sufficient condition suppresses effects of the nuisance parameters completely. 
In this case, the optimal estimation strategy is the repetitive one. 
\begin{align} \nonumber
\mathrm{(i)}&\ \theta_{\rm{I}}\mbox{ is globally orthogonal to } \bm{\theta}_{\rm{N}}\\ \label{pocond1}
\mathrm{(ii)}&\ \Pi_{\theta_{\rm{I}}}^*\mbox{ is independent of }\bm{\theta}\\ \nonumber
&\Leftrightarrow \forall x\in\cX^1, E_{\bm{\theta}}^1(x)\mbox{ is independent of }\bm{\theta}; 
\end{align}
To demonstrate usefulness of the global parameter orthogonalization method, 
consider the quantum exponential family \eqref{Qe-family}. 
The SLDs are calculated as
\be
\SLD{i}=F_i-\del_i\psi_{\bm{\theta}} I. 
\ee
From this expression, we see that the projector onto a subspace of each spectrum is independent of the parameters, 
since $F_i$ is $\bm{\theta}$-independent and the second term is irrelevant. 
All $F_i$ are mutually commutative by definition, and hence they can be simultaneously diagonalizable. 
Any linear combination of the SLDs is also $\bm{\theta}$-independent. 
Note that the SLD operator for the parameter of interest in the new parametrization is expressed as a linear combination 
of the SLDs \eqref{sldxi}.  
We thus see that the optimal PVM about the parameter of interest is independent of the parameters $\bm{\theta}$. 
Therefore, the above two conditions in \eqref{pocond1} are satisfied to conclude that we can attain the precision limit set by 
the partial SLD Fisher information within the repetitive strategy (A1). 

\subsection{Related works} 
Finally, we conclude with a brief discussion on another method of treating one-parameter estimation. 
In \cite{wsu10}, Watanabe {\it et al.}\ showed an optimal estimation strategy for estimating the expectation value 
of an arbitrary observable in the presence of a non-parametric quantum noise. 
A crucial assumption in their work is that one has no prior information about the state under consideration. 
This is to consider a full parameter model ($d=(\dim \cH)^2-1$) as a parametric model, which is D-invariant. 
In addition, the noise model was assumed to be known and was not treated as nuisance. 
Then, the problem can be formulated as estimating a single parameter of interest, which is a linear combination 
of these parameters, whereas the rest of the parameters are nuisance parameters.   
Compared with \cite{wsu10}, our formulation is more general and is applicable for arbitrary model (see the example in Subsection \ref{subsec-example-ob}). 
Another observation is that \cite{wsu10} only proves optimality within separable POVMs. 
In fact, this optimality can also be shown within all possible POVMs as we proved in this review. 

In a recent work \cite{tsang19}, Tsang proposed a framework called quantum semiparametric estimation, 
which offers an alternative approach to determine the precision bound of estimating a single parameter in the presence of (infinitely) many nuisance parameters. 
Compared to the Cram\'er-Rao approach in our review paper, the semiparametric estimation approach does not follow the procedure based on the inverse matrix of the quantum Fisher information matrix.
He derived a lower bound for mean square error 
with under the unbiased condition 
from a geometrical viewpoint.
Although his obtained bound (see \cite[Theorem 6]{tsang19}) is the same as our bound \eqref{vCRbound} in Section 5.2,
his achievement is different from ours in the following way.
(i) For estimators, he imposed the unbiased condition, which are rather unrealistic as already mentioned in Section \ref{sec:Qnui1}
while we consider the locally unbiased condition.
Hence, he did not show the achievability 
nor how to construct the optimal measurement 
while we show the achievability
with the construction of the optimal one under the locally unbiased condition.
(ii) While he characterized the lower bound, he did not gave an explicit form of the lower bound.
But, we give a concrete calculation formula for the lower bound.
(iii) His method can be applied to the case with
infinitely many nuisance parameters while our method can be applied to the case with 
a finite number of nuisance parameters. 
Nevertheless, the semiparametric approach has indeed brought new insights into quantum estimation 
in the presence of nuisance parameters and is worthy of more investigation. 
For instance, a hybrid approach combining the advantages of both aforementioned approaches would definitely be desired in many applications.
 
%%%%%%%%%%%%%%%%%%%%%%%%%%%%%%%%%%%

\section{Examples}\label{sec:example}
In the following section, we give examples to show the effects of nuisance parameters in quantum estimation, and show how to derive quantum CR bounds in the presence of nuisance parameters.
\subsection{A noisy qubit clock.}\label{subsec-example-qubit}
We first revisit the example in the introduction and show how it can be tackled using the results in one-parameter estimation with nuisance parameters.
Recall that the task is to estimate time from identical copies of a two-level atom with known Hamiltonian, 
which is assumed to be $\sigma_z/2=-(1/2)(\ket{0}\bra{0}-\ket{1}\bra{1})$ for simplicity. In the meantime, the qubit also suffers from dephasing noise, and thus its state at time $t$ is 
$\rho_{t,\gamma}=e^{-\gamma t}\ket{\psi_t}\bra{\psi_t}+\left(1-e^{-\gamma t}\right)\frac{I}{2}$
where  $\ket{\psi_t}:=(1/\sqrt{2})(\ket{0}+e^{-it}\ket{1})$ and $\gamma\ge 0$ is the decay rate.  
For the state $\rho_{t,\gamma}$, the SLD quantum Fisher information matrix can be evaluated as
\begin{align}
J_{t,\gamma}^{\rm S}=\left(\begin{array}{cc}e^{-2\gamma t}+\frac{\gamma^2}{e^{2\gamma t}-1} & \frac{\gamma t}{e^{2\gamma t}-1}\\
\frac{\gamma t}{e^{2\gamma t}-1} & \frac{t^2}{e^{2\gamma t}-1}\end{array}\right).
\end{align}
According to Theorem \ref{thmAN2}, the optimal measurement has an error equal to
\begin{align}
J_{t,\gamma}^{{\rm S};1,1}=e^{2\gamma t}.
\end{align}
One can see from this example the effect of the nuisance parameter $\gamma$, since this value is strictly larger than the inverse of $\left(J^{\rm S}_{t,\gamma}\right)_{1,1}$. In addition, we note that the choice of the nuisance parameter is not unique. Indeed, we can perform the change of variables $(t,\gamma)\to(t,p)$, where $p:=(1+e^{-\gamma t})/2$ is the mixedness of the qubit. In the new coordinate, the qubit state becomes $\rho_{t,p}=p\ket{\psi_t}\bra{\psi_t}+\left(1-p\right)\ket{\psi_t^\perp}\bra{\psi_t^\perp}$, 
with $\ket{\psi_t^\perp}$ being orthogonal to $\ket{\psi_t}$. The SLD quantum Fisher information for $(t,p)$ can be evaluated as
\begin{align}
J_{t,p}^{\rm S}=\left(\begin{array}{cc} (2p-1)^2 & 0\\
0& \frac{1}{p(1-p)}\end{array}\right).
\end{align}
One can see from the above matrix that this choice of the nuisance parameter makes it orthogonal to the parameter of interest, as discussed in Section \ref{sec:QpoLocal}.
One can also easily check that 
\begin{align}
J^{{\rm S},1,1}_{t,p}=J^{{\rm S},1,1}_{t,\gamma},
\end{align}
since the choice of nuisance parameters does not affect the precision bound.

\subsection{Estimating a generic observable of a $d$-dimensional system}\label{subsec-example-ob}
The next example is to estimate an observable of a generic qudit state, which has been analyzed by Watanabe et al. in \cite{wsu10,wsu11}. A generic qudit state can be expressed as
\begin{align}
\rho_{\bm{\bm{\theta}}}=\frac{I}{d}+{\bm{\bm{\theta}}}\cdot{\bm H},
\end{align}
where ${\bm H}=(H_1,H_2,\dots,H_{d^2-1})^{\rm T}$ is a vector of traceless Hermitians satisfying $\Tr{H_iH_j}=\delta_{i,j}$ for any $i$ and $j$.
The Hermitians $\{H_j\}$ form a basis for traceless operators, and a generic observable $A$ to measure, assumed without loss of generality to be traceless \footnote{Notice that measuring $A$ is essentially the same as measuring $A-\Tr AI$.}, can be thus expressed as $A={\bm v}\cdot{\bm H}$. 
The parameter of interest  is then the expectation $\expectn{A}:=\sum_i v_i\Tr{\rho_{\bm{\theta}}H_i}$ of $A$ with respect to the qudit state $\rho_{\bm{\bm{\theta}}}$.

For the generic qudit model, the inverse SLD quantum Fisher information matrix can be evaluated as
\begin{align}\label{qudit-SLD}
J_{{\bm{\bm{\theta}}}}^{{\rm S};i,j}=\left\langle\frac{H_i H_j+H_j H_i}{2}\right\rangle-\expectn{H_i}\expectn{H_j}.
\end{align}
\if0
On the other hand, the inverse RLD quantum Fisher information matrix can be evaluated as
\begin{align} 
J_{{\bm{\bm{\theta}}}}^{{\rm R};i,j}=\left\langle H_i H_j\right\rangle-\expectn{H_i}\expectn{H_j}.
\end{align}
Defining the quantum Fisher information difference as the Hermitian matrix  
\begin{align}\label{qudit-QFIdiff}
\left(\Delta J^{-1}_{\bm {\bm{\theta}}}\right)_{i,j}:=-\frac{i}{2}\expectn{[H_i,H_j]},
\end{align}
we can rewrite the RLD quantum Fisher information matrix as
\begin{align}\label{qudit-RLD}
J_{{\bm{\bm{\theta}}}}^{{\rm R};i,j}=J_{{\bm{\bm{\theta}}}}^{{\rm S};i,j}+i\left(\Delta J^{-1}_{\bm {\bm{\theta}}}\right)_{i,j}.
\end{align}
\fi
With the above discussion, we can now analyze the performance of different estimation strategies (see Section \ref{sec:Qnui1-2}). 
 If one adopts the adaptive strategy [see strategy (A2) in Section \ref{sec:Qnui1-2}], the minimum achievable MSE is given by the SLD bound in Theorem \ref{thmAN}, which reads 
\begin{align}
\min_{\hat{\Pi}\in A2}\lim_{n\to\infty}nV_{\expectn{A}}[\hat{\Pi}^{\otimes n}]=
{\bm v}^{\rm T}\left(J_{{\bm{\bm{\theta}}}}^{\rm S}\right)^{-1}{\bm v}
\end{align} 
where the minimization taken over all adaptive strategies (i.e.\ the set $A2$).
As shown in \eqref{e92},
this value equals
the minimum achievable MSE over all strategies, including those that require collective measurements on all copies of the state [see strategy (A3) in section\ \ref{sec:Qnui1-2}].
Fortunately, the variance of the observable $A$, i.e.\ $\expectn{A^2}-\expectn{A}^2$,
equals ${\bm v}^{\rm T}\left(J_{{\bm{\bm{\theta}}}}^{\rm S}\right)^{-1}{\bm v}$.
Hence, when we repeat the measurement of the observable $A$,
this bound can be attained. That is, 
this bound can be attained even in repetitive strategy 
[see strategy (A1) in section\ \ref{sec:Qnui1-2}].
That is, there is no difference among these three settings. 

\if0
On the other hands, the generic qudit model is D-invariant, and thus the Holevo bound equals the RLD bound. 

The minimum achievable MSE over all strategies, including those that require collective measurements on all copies of the state [see strategy (A3) in section\ \ref{sec:Qnui1-2}], is given by the RLD bound (\ref{Eq:rld_crbound}), which reads  

\begin{align}
\min_{\hat{\Pi}^{(n)}\in A3}\lim_{n\to\infty}nV_{\expectn{A}}[\hat{\Pi}^{(n)}]&={\bm v}^{\rm T}\Re\left(J_{{\bm{\bm{\theta}}}}^{\rm R}\right)^{-1}{\bm v}+\left|{\bm v}^{\rm T}\Im\left(J_{{\bm{\bm{\theta}}}}^{\rm R}\right)^{-1}{\bm v}\right|\nonumber\\
&={\bm v}^{\rm T}\left(J_{{\bm{\bm{\theta}}}}^{\rm S}\right)^{-1}{\bm v}+\left|{\bm v}^{\rm T}\Delta J^{-1}_{\bm {\bm{\theta}}}{\bm v}\right|.
\end{align}

Therefore, the gap between the repetitive strategy and the optimal collective strategy is 
\begin{align}
\Delta_{\expectn{A}}:=\min_{\hat{\Pi}\in A1}\lim_{n\to\infty}nV_{\expectn{A}}[\hat{\Pi}^{\otimes n}]-\min_{\hat{\Pi}^{(n)}\in A3}\lim_{n\to\infty}nV_{\expectn{A}}[\hat{\Pi}^{(n)}]%\nonumber\\
=\left|{\bm v}^{\rm T}\Delta J^{-1}_{\bm {\bm{\theta}}}{\bm v}\right| .
\end{align}  
One can see from the above expression that usually $\Delta_{\expectn{A}}>0$, and there is a non-trivial advantage of adapting a collective strategy.
\fi

\subsection{Multiparameter estimation with nuisance parameters: a qubit case.}
As the last example we consider estimation of two parameters of a qubit state. Since a qubit model consists of three real parameters, the last one of them should be regarded as a nuisance parameter. For more cases regarding qubit estimation with nuisance parameters, we refer the readers to \cite{jsNuipaper}.

A generic qubit can be expressed as $\rho_{\bm{\bm{\theta}}}=\frac{1}{2}(I+\theta_1\sigma_1+\theta_2\sigma_2+\theta_3\sigma_3)$
with $\{\sigma_i\}_{i=1,2,3}$ being the Pauli matrices and the vector of parameters satisfying the constraint $\|\bm{\bm{\theta}}\|\le 1$.
We consider the first two parameters as parameters of interest, i.e.\ ${\bm{\bm{\theta}}}_{\rm I}=(\theta_1,\theta_2)$. 
%Physically, since $\theta_3$ captures the mixedness of the state, this case is relevant to the scenario of pure qubit state tomography under depolarizing noise.

The inverse of the SLD Fisher information matrix can be evaluated as
\begin{align}
\left(J^{\rm S}_{\bm{\bm{\theta}}}\right)^{-1}=\left(\begin{array}{ccc}1-\theta_1^2 & -\theta_1\theta_2 & -\theta_1\theta_3\\ -\theta_1\theta_2 &1-\theta_2^2 & -\theta_2\theta_3 \\ -\theta_1\theta_3 & -\theta_2\theta_3 & 1-\theta_3^2\end{array}\right).
\end{align}
Next, we perform parameter orthogonalization by switching to a new coordinate $\bm{\bm{\xi}}$, defined by $\theta_1={\xi}_1$, $\theta_2={\xi}_2$, and $\theta_3={\xi}_3\sqrt{1-{\xi}_1^2-{\xi}_2^2}$. Under the new coordinate $\bm {\bm{\xi}}$, the inverse of the SLD Fisher information matrix has the form
\begin{align}
\left(J^{\rm S}_{\bm{\bm{\xi}}}\right)^{-1}=\left(\begin{array}{ccc}1-{\xi}_1^2 & -{\xi}_1{\xi}_2 & 0\\ -{\xi}_1{\xi}_2 &1-{\xi}_2^2 & 0 \\ 0 & 0 & \frac{1-c(\xi_3)^2}{\dot{c}(\xi_3)^2(1-(\xi_1)^2-(\xi_2)^2)}\end{array}\right),
\end{align}
where $c(\xi_3)$ is an arbitrary differentiable function satisfying the condition $\forall\xi_3,\dot{c}(\xi_3):=dc(\xi_3)/d\xi_3\neq0$. 
After the parameter orthogonalization, the estimation precision for the parameters of interest 
the depends only on the following submatrix
\begin{align}
J^{{\rm S};{\rm I},{\rm I}}_{\bm{\xi}}=
\left(\begin{array}{cc}1-{\xi}_1^2 & -{\xi}_1{\xi}_2  \\ -{\xi}_1{\xi}_2 &1-{\xi}_2^2    \end{array}\right)
=J^{{\rm S};{\rm I},{\rm I}}_{\bm{\theta}},
\end{align}
and the precision limit under any separable measurement (A1 and A2) can be obtained by 
setting the weight matrix in the Gill-Masser bound \eqref{qcrbound3} as 
\[
W=\left(\begin{array}{cc}W_{\rm I} & \begin{array}{c}0 \\0\end{array} \\ 
\begin{array}{cc}0 & 0\end{array} & 0\end{array}\right). 
\]
Owing to the global parameter orthogonalization, this is equivalent to substituting the partial SLD Fisher information matrix 
$J^{\mathrm{S}}_{\bm{\theta}}({\mathrm{I}}|{\mathrm{N}})=(J^{{\rm S};{\rm I},{\rm I}}_{\bm{\theta}})^{-1}$ into 
the Nagaoka bound [$\dim \cH=2$ and $d=2$ in \eqref{qcrbound3}]. 
Explicitly, we have 
\begin{equation}
C^{\rm{N}}_{\bm{\theta};\rm{I}}[W_{\rm{I}},\cM]=\Tr{W_{\rm{I}} 
\left(J^{\mathrm{S}}_{\bm{\theta}}({\mathrm{I}}|{\mathrm{N}})\right)^{-1}}
+2\sqrt{\det W_{\rm{I}}} \sqrt{\det\left(J^{\mathrm{S}}_{\bm{\theta}}({\mathrm{I}}|{\mathrm{N}})\right)^{-1}} . 
\end{equation}

If we consider all possible POVMs to attain the ultimate precision limit for estimating $\theta_1,\theta_2$, 
we can show that the Holevo bound for the parameters of interest \eqref{N-holevo-bound} is given by
\begin{equation}
C_{\bm{\theta};\mathrm{I}}^H[W_{\mathrm{I}},\cM]=
\Tr{W_{\rm{I}} J^{{\rm S};{\rm I},{\rm I}}_{\bm{\theta}}} +2\sqrt{\det W_{\rm{I}}} |\theta_3| .
\end{equation}
Therefore, estimation error by performing collective POVMs (A3) can be lowered by the amount 
\be
C^{\rm{N}}_{\bm{\theta};\rm{I}}[W_{\rm{I}},\cM]- C_{\bm{\theta};\mathrm{I}}^H[W_{\mathrm{I}},\cM]
= 2\sqrt{\det W_{\rm{I}}}\left(  \sqrt{\det J^{{\rm S};{\rm I},{\rm I}}_{\bm{\theta}}}-  |\theta_3|\right). 
\ee
This is positive, since $\sqrt{\det J^{{\rm S};{\rm I},{\rm I}}_{\bm{\theta}}}=\sqrt{1-(\theta_1)^2-(\theta_2)^2}> |\theta_3|$ 
holds for any mixed-state model.  

Last, we discuss information loss \eqref{qinfoloss} in the presence of the nuisance parameter $\theta_{\rm N}=\theta_3$. 
The Holevo bound for the general two-parameter qubit-state model $\cM'=\{\rho_{\rm \theta}\in\cM|\bm{\theta}=(\theta_1,\theta_2)\in\Theta' \}$ 
for a fixed $\theta_3$ is given by the formula \eqref{Hbound_qubit}. 
For the specific parametrization under consideration, 
we can use results in \cite[section V C]{js16} to evaluate information loss as 
\begin{align}
&\Delta C^H_{\bm{\theta};\mathrm{I}}[W_{\mathrm{I}}|\bm{\theta}_{\mathrm{N}}]
=C^H_{\bm{\theta};\mathrm{I}}[W_{\mathrm{I}},\cM]-C^H_{\bm{\theta};\mathrm{I}}[W_{\mathrm{I}},\cM']\\
&=\begin{cases}
(\theta_1\,\theta_2)\left[ \frac{\Tr{W_{\rm I}}}{1-\theta_3^2}I- W_{\rm I} \right]\binom{\theta_1}{\theta_2}
+2\sqrt{\det W_{\rm I}}|\theta_3|\left[1- \sqrt{\frac{1-s_{\bm \theta}^2}{1-\theta_3^2}}\right] & 
\mbox{ for }B_{\bm \theta}[W_{\rm I}]\ge0\\[2ex]
\frac{\theta_3^2}{1-\theta_3^2}(\theta_1\,\theta_2)W_{\rm I} \binom{\theta_1}{\theta_2}
+2|\theta_3|\sqrt{\det W_{\rm I}} \left[ 1-\frac12\,\frac{(1-s_{\bm \theta}^2)|\theta_3|\sqrt{\det W_{\rm I}}}{(\theta_1\,\theta_2)\left[ \Tr{W_{\rm I}}I- W_{\rm I} \right]\binom{\theta_1}{\theta_2}} \right]& \mbox{ for }B_{\bm \theta}[W_{\rm I}]<0
\end{cases},  
\end{align}
where $s_{\bm \theta}^2=\theta_1^2+\theta_2^2+\theta_3^2$ is the square of the Bloch vector and 
$B_{\bm \theta}[W_{\rm I}]$ is defined by
\[
B_{\bm \theta}[W_{\rm I}]:=- \frac{1}{1-\theta_3^2} (\theta_1\,\theta_2)\left[ \Tr{W_{\rm I}}I- W_{\rm I} \right]\binom{\theta_1}{\theta_2}+\sqrt{\frac{1-s_{\bm \theta}^2}{1-\theta_3^2}}|\theta_3|\sqrt{\det W_{\rm I}}.
\] 
In contrast to the single parameter estimation problem in the presence of nuisance parameters \eqref{qinfoloss_1para}, 
information loss is much complex even in this simple qubit model. 
It is worth exploring the structure of information loss to gain a deeper insight into effects of the nuisance parameters 
in quantum estimation theory.

\section{Conclusion and open questions}\label{sec:conclusion}

As discussed in this review, the nuisance parameter problem is a common and practical problem. 
We have derived the ultimate precision limit for the parameters of interest in the presence of the nuisance parameters. 
This bound is not expressed in a closed form except when there is only one single parameter of interest, thus it is hard to understand 
the effects of the nuisance parameters in a simple picture. 
An important concrete question is to derive the necessary and sufficient condition for the zero loss of information \eqref{qinfoloss}. 
Classically, this condition is expressed as the orthogonality condition with respect to the classical Fisher information matrix. 
The quantum case, on the other hand, is much more complicated and deserves further exploration. 

Another important aspect of the nuisance parameter problem is the trade-off between the error of estimating the parameters of interest and the error of estimating the nuisance parameters. This trade-off relation 
is particularly important when dealing with the finite sample case \cite{jsNuipaper}. 
Noting that an optimal POVM minimizing the mean-square error for the parameters of interest depends on 
unknown nuisance parameters, we cannot completely neglect the error of estimating the nuisance parameters. 
The question is then how much knowledge one should acquire on the value of the nuisance parameters for a given sample size.
The nuisance parameter problem also appears in other statistical inference problems and quantum control theory \cite{d2007introduction,wiseman_milburn_2009,jacobs2014quantum}. 
Proper extensions of statistical methods known in classical statistics will be needed to address these problems.

In this review, we have introduced the framework and tools of treating nuisance parameters in quantum state estimation. On this basis, it is natural to consider nuisance parameters in quantum metrology \cite{giovannetti2006quantum,glm11}, which is a vigorous research direction concerning estimating parameters from physical processes instead of quantum states. 
In quantum metrology, the parameters to be estimated can be encoded in physical processes ranging from multiple uses of noiseless \cite{giovannetti2006quantum} and noisy \cite{huelga1997improvement,escher2011general,ddkg12} gates to complex processes with memories \cite{huelga1997improvement,matsuzaki2011magnetic,chin2012quantum,macieszczak2015zeno,bpla19,yang2019}.
Practically, all of these processes are, to some extent, subject to noises depending on unknown parameters that can be treated as nuisance parameters. Suitable extensions of the tools presented in this review will, therefore, be able to quantify the effects of nuisance parameters in quantum metrology. 
Researches in this direction will be timely and promising, as quantum metrology is likely to become one of the earliest applicable quantum technologies.

\section*{Notes added in this version}
After our accepted paper went through the proof, we noticed that \cite{tsang19} was updated as \cite{tsang_v6} with additional results. 
To give a comparison to their latest results, we added two remarks, Remark \ref{remark1} and Remark \ref{remark2}, in this version. 
Newly added Sec. VIII of \cite{tsang_v6} corresponds to section \ref{sec:Qfunction} of this paper in the parametric case. 
For the sake of completeness, Theorem 9 of \cite{tsang_v6} was added as a side remark (Remark \ref{remark1}) to show that it is a simple consequence of our formalism based on the standard argument. 
We also provide additional supplement as Remark \ref{remark2} for the existence 
of the minimum in \eqref{N-holevo-bound_function} for readers' convenience. 
We would like to thank Dr.~Mankei Tsang for additional remarks on our results. 

%%%%%%%%%%%%%%%%%%%%%%%%%%%%%%%%%%%%%%%%%
\section*{Acknowledgement}
MH is grateful to Mr.\ Daiki Suruga and Mr.\ Seunghoan Song for providing helpful comments for this paper. 
JS is partly supported by JSPS Grant-in-Aid for Scientific Research (C) No. 17K05571. 
YY is supported by the Swiss National Science Foundation via the National Center for Competence in Research ``QSIT" as well as via project No.\ 200020\_165843.
MH was supported in part by JSPS Grant-in-Aid for Scientific Research (A) No.17H01280 and for Scientific Research (B) No.16KT0017, and Kayamori Foundation of Informational Science Advancement.

\section*{References}
\bibliographystyle{iopart-num}
\bibliography{ref_arxiv}

\providecommand{\newblock}{}
\begin{thebibliography}{100}
\expandafter\ifx\csname url\endcsname\relax
  \def\url#1{{\tt #1}}\fi
\expandafter\ifx\csname urlprefix\endcsname\relax\def\urlprefix{URL }\fi
\providecommand{\eprint}[2][]{\url{#2}}

\bibitem{fisher35}
Fisher R~A 1935 {\em Journal of the Royal Statistical Society\/} {\bf 98}
  39--82

\bibitem{amari}
Amari S~I 1985 {\em Differential-Geometrical Methods in Statistics\/}
  (Springer-Verlag)

\bibitem{lc}
Lehmann E~L and Casella G 2006 {\em Theory of point estimation\/} (Springer
  Science \& Business Media)

\bibitem{bnc}
Barndorff-Nielsen O~E and Cox D~R 1994 {\em Inference and asymptotics\/}
  (Chapman$\backslash$\& Hall)

\bibitem{ANbook}
Amari S~I and Nagaoka H 2007 {\em Methods of information geometry\/} (American
  Mathematical Soc.)

\bibitem{basu77}
Basu D 1977 {\em Journal of the American Statistical Association\/} {\bf 72}
  355--366

\bibitem{ka84}
Kumon M and Amari S~I 1984 {\em Biometrika\/} {\bf 71} 445--459

\bibitem{rc87}
Cox D~R and Reid N 1987 {\em Journal of the Royal Statistical Society: Series B
  (Methodological)\/} {\bf 49} 1--18

\bibitem{ak88}
Amari S~I and Kumon M 1988 {\em The Annals of Statistics\/} {\bf 16} 1044--1068

\bibitem{bs94}
Bhapkar V~P and Srinivasan C 1994 {\em Annals of the Institute of Statistical
  Mathematics\/} {\bf 46} 593--604

\bibitem{zr94}
Zhu Y and Reid N 1994 {\em Canadian Journal of Statistics\/} {\bf 22} 111--123

\bibitem{gardiner2004quantum}
Gardiner C and Zoller P 2004 {\em Quantum noise: a handbook of Markovian and
  non-Markovian quantum stochastic methods with applications to quantum
  optics\/} (Springer-Verlag)

\bibitem{huelga1997improvement}
Huelga S~F, Macchiavello C, Pellizzari T, Ekert A~K, Plenio M~B and Cirac J~I
  1997 {\em Physical Review Letters\/} {\bf 79} 3865

\bibitem{ych17}
Yang Y, Chiribella G and Hayashi M 2018 {\em Proceedings of the Royal Society
  A: Mathematical, Physical and Engineering Sciences\/} {\bf 474} 20170773

\bibitem{Demkowicz-Dobrzanski2020}
Demkowicz-Dobrzanski R, Gorecki W and Guta M 2020 Multi-parameter estimation
  beyond quantum fisher information (\textit{Preprint} \eprint{2001.11742})

\bibitem{Albarelli2020}
Albarelli F, Barbieri M, Genoni M~G and Gianani I 2020 {\em Physics Letters
  A\/} {\bf 384} 126311

\bibitem{at95}
Akahira M and Takeuchi K 2012 {\em Non-regular statistical estimation\/}
  (Springer Science \& Business Media)

\bibitem{lancaster}
Lancaster T 2000 {\em Journal of Econometrics\/} {\bf 95} 391--413

\bibitem{bhatia}
Bhatia R 2009 {\em Positive definite matrices\/} (Princeton university press)

\bibitem{rm97}
Reuven I and Messer H 1997 {\em IEEE Transactions on Information Theory\/} {\bf
  43} 1084--1093

\bibitem{jsNuipaper}
Suzuki J 2019 Nuisance parameter problem in quantum estimation theory: General
  formulation and qubit examples (\textit{Preprint} \eprint{1905.04733})

\bibitem{amari_comment}
Amari S in Discussion of the paper by Cox and Reid \cite{rc87}

\bibitem{helstrom}
Helstrom C~W 1976 {\em Quantum detection and estimation theory\/} (Academic
  press)

\bibitem{holevo}
Holevo A~S 2011 {\em Probabilistic and statistical aspects of quantum theory\/}
  (Edizioni della Normale)

\bibitem{hayashi}
Masahito H (ed) 2005 {\em Asymptotic theory of quantum statistical inference:
  selected papers\/} (World Scientific)

\bibitem{petz}
Petz D 2007 {\em Quantum information theory and quantum statistics\/} (Springer
  Science \& Business Media)

\bibitem{fn95}
Fujiwara A and Nagaoka H 1995 {\em Physics Letters A\/} {\bf 201} 119--124

\bibitem{fn99}
Fujiwara A and Nagaoka H 1999 {\em Journal of Mathematical Physics\/} {\bf 40}
  4227--4239

\bibitem{H-group}
Hayashi M 2017 {\em A Group Theoretic Approach to Quantum Information\/}
  (Springer)

\bibitem{fedorov}
Fedorov V~V 1972 {\em Theory of optimal experiments\/} (Academic Press)

\bibitem{pukelsheim}
Pukelsheim F 2006 {\em Optimal design of experiments\/} (SIAM)

\bibitem{fh97}
Fedorov V~V and Hackl P 2012 {\em Model-oriented design of experiments\/}
  (Springer Science \& Business Media)

\bibitem{pp13}
Pronzato L and P{\'a}zman A 2013 {\em Design of experiments in nonlinear
  models\/} (Springer \& Business Media)

\bibitem{fl14}
Fedorov V~V and Leonov S~L 2013 {\em Optimal design for nonlinear response
  models\/} (CRC Press)

\bibitem{gns19}
Gazit Y, Ng H~K and Suzuki J 2019 {\em Physical Review A\/} {\bf 100} 012350

\bibitem{holevo1973statistical}
Holevo A~S 1973 {\em Journal of multivariate analysis\/} {\bf 3} 337--394

\bibitem{ozawa1980optimal}
Ozawa M 1980 {\em Reports on Mathematical Physics\/} {\bf 18} 11--28

\bibitem{personick71}
Personick S 1971 {\em IEEE Transactions on Information Theory\/} {\bf 17}
  240--246

\bibitem{whth07}
Wang X~B, Hiroshima T, Tomita A and Hayashi M 2007 {\em Physics reports\/} {\bf
  448} 1--111

\bibitem{tanaka2006}
Tanaka F 2006 Generalized bayesian predictive density operators
  (\textit{Preprint} \eprint{0602072})

\bibitem{teklu2009bayesian}
Teklu B, Olivares S and Paris M~G 2009 {\em Journal of Physics B: Atomic,
  Molecular and Optical Physics\/} {\bf 42} 035502

\bibitem{teklu2010phase}
Teklu B, Genoni M~G, Olivares S and Paris M~G 2010 {\em Physica Scripta\/} {\bf
  2010} 014062

\bibitem{brivio2010experimental}
Brivio D, Cialdi S, Vezzoli S, Gebrehiwot B~T, Genoni M~G, Olivares S and Paris
  M~G 2010 {\em Physical Review A\/} {\bf 81} 012305

\bibitem{blume2010optimal}
Blume-Kohout R 2010 {\em New Journal of Physics\/} {\bf 12} 043034

\bibitem{christandl2012reliable}
Christandl M and Renner R 2012 {\em Physical Review Letters\/} {\bf 109} 120403

\bibitem{tsang2012ziv}
Tsang M 2012 {\em Physical review letters\/} {\bf 108} 230401

\bibitem{koyama2017minimax}
Koyama T, Matsuda T and Komaki F 2017 {\em Entropy\/} {\bf 19} 618

\bibitem{teo2018bayesian}
Teo Y~S, Oh C and Jeong H 2018 {\em New Journal of Physics\/} {\bf 20} 093009

\bibitem{oh2018bayesian}
Oh C, Teo Y~S and Jeong H 2018 {\em New Journal of Physics\/} {\bf 20} 093010

\bibitem{quadeer2019minimax}
Quadeer M, Tomamichel M and Ferrie C 2019 {\em Quantum\/} {\bf 3} 126

\bibitem{yl73}
Yuen H and Lax M 1973 {\em IEEE Transactions on Information Theory\/} {\bf 19}
  740--750

\bibitem{KM}
Matsumoto K 2002 {\em Journal of Physics A: Mathematical and General\/} {\bf
  35} 3111

\bibitem{nagaoka89}
Nagaoka H 1989 {\em IEICE Tech Report\/} {\bf IT 89-42} 9--14 (Reprinted in
  \cite{hayashi})

\bibitem{sm01}
Stoica P and Marzetta T~L 2001 {\em IEEE Transactions on Signal Processing\/}
  {\bf 49} 87--90

\bibitem{GM00}
Gill R~D and Massar S 2000 {\em Physical Review A\/} {\bf 61} 042312

\bibitem{hayashi97}
Hayashi M 1997 A linear programming approach to attainable cramer-rao type
  bound {\em Quantum Communication, Computing, and Measurement\/} ed Hirota O,
  Holevo A~S and Caves C~M (Plenum, New York)

\bibitem{PRA.75.042108}
\v{R}eh\'a\v{c}ek J, Hradil Z, Knill E and Lvovsky A~I 2007 {\em Physical
  Review A\/} {\bf 75} 042108

\bibitem{PRA.90.012115}
Zhu H 2014 {\em Physical Review A\/} {\bf 90} 012115

\bibitem{npj.3.44}
Bolduc E, Knee G~C, Gauger E~M and Leach J 2017 {\em npj Quantum Information\/}
  {\bf 3} 44

\bibitem{QSEbook}
Paris M~G~A and \v{R}eh\'a\v{c}ek J~E 2004 {\em Quantum State Estimation\/}
  (Springer)

\bibitem{H98}
Hayashi M 1998 {\em Journal of Physics A: Mathematical and General\/} {\bf 31}
  4633

\bibitem{LFGKC}
Li N, Ferrie C, Gross J~A, Kalev A and Caves C~M 2016 {\em Physical Review
  Letters\/} {\bf 116} 180402

\bibitem{ZH18}
Zhu H and Hayashi M 2018 {\em Physical Review Letters\/} {\bf 120} 030404

\bibitem{nagaoka89-2}
Nagaoka H 2005 On the parameter estimation problem for quantum statistical
  models {\em Asymptotic Theory Of Quantum Statistical Inference: Selected
  Papers\/} ed Hayashi M (World Scientific) pp 125--132

\bibitem{HM98}
Hayashi M and Matsumoto K 1998 Statistical model with measurement degree of
  freedom and quantum physics {\em Surikaiseki Kenkyusho Kokyuroku\/} vol 1055
  p~96 (English translation available in \cite{hayashi})

\bibitem{BNG00}
Barndorff-Nielsen O and Gill R 2000 {\em Journal of Physics A: Mathematical and
  General\/} {\bf 33} 4481

\bibitem{fujiwara06}
Fujiwara A 2006 {\em Journal of Physics A: Mathematical and General\/} {\bf 39}
  12489

\bibitem{stm12}
Sugiyama T, Turner P~S and Murao M 2012 {\em Physical Review A\/} {\bf 85}
  052107

\bibitem{oioyift12}
Okamoto R, Iefuji M, Oyama S, Yamagata K, Imai H, Fujiwara A and Takeuchi S
  2012 {\em Physical Review Letters\/} {\bf 109} 130404

\bibitem{mrdfbks13}
Mahler D, Rozema L~A, Darabi A, Ferrie C, Blume-Kohout R and Steinberg A 2013
  {\em Physical Review Letters\/} {\bf 111} 183601

\bibitem{ksrhhk13}
Kravtsov K, Straupe S, Radchenko I, Houlsby N, Husz{\'a}r F and Kulik S 2013
  {\em Physical Review A\/} {\bf 87} 062122

\bibitem{hzxlg16}
Hou Z, Zhu H, Xiang G~Y, Li C~F and Guo G~C 2016 {\em npj Quantum
  Information\/} {\bf 2} 16001

\bibitem{ooyft17}
Okamoto R, Oyama S, Yamagata K, Fujiwara A and Takeuchi S 2017 {\em Physical
  Review A\/} {\bf 96} 022124

\bibitem{zlwjn17}
Zhang J, Liu Y~X, Wu R~B, Jacobs K and Nori F 2017 {\em Physics Reports\/} {\bf
  679} 1--60

\bibitem{H-book}
Hayashi M 2017 {\em Quantum Information Theory\/} (Springer)

\bibitem{YCH18}
Yang Y, Chiribella G and Hayashi M 2019 {\em Communications in Mathematical
  Physics\/} {\bf 368} 223--293

\bibitem{js18_clmodel}
Suzuki J 2019 {\em Entropy\/} {\bf 21} 703

\bibitem{js16}
Suzuki J 2016 {\em Journal of Mathematical Physics\/} {\bf 57} 042201

\bibitem{HM08}
Hayashi M and Matsumoto K 2008 {\em Journal of Mathematical Physics\/} {\bf 49}
  102101

\bibitem{nagaoka87}
Nagaoka H 2005 On fisher information of quantum statistical models {\em
  Asymptotic Theory Of Quantum Statistical Inference: Selected Papers\/} ed
  Hayashi M (World Scientific) pp 113--124

\bibitem{rjdd16}
Ragy S, Jarzyna M and Demkowicz-Dobrza{\'n}ski R 2016 {\em Physical Review A\/}
  {\bf 94} 052108

\bibitem{PhysRevA.97.012106}
Bradshaw M, Lam P~K and Assad S~M 2018 {\em Physical Review A\/} {\bf 97}
  012106

\bibitem{afd19}
Albarelli F, Friel J~F and Datta A 2019 {\em Physical Review Letters\/} {\bf
  123} 200503

\bibitem{js15}
Suzuki J 2015 {\em International Journal of Quantum Information\/} {\bf 13}
  1450044

\bibitem{tsang19}
Tsang M 2019 Quantum semiparametric estimation (\textit{Preprint}
  \eprint{1906.09871v5})

\bibitem{gross2020one}
Gross J~A and Caves C~M 2020 One from many: Estimating a function of many
  parameters (\textit{Preprint} \eprint{2002.02898})

\bibitem{Carollo_2019}
Carollo A, Spagnolo B, Dubkov A~A and Valenti D 2019 {\em Journal of
  Statistical Mechanics: Theory and Experiment\/} {\bf 2019} 094010
  \urlprefix\url{https://doi.org/10.1088%2F1742-5468%2Fab3ccb}

\bibitem{Carollo_2020}
Carollo A, Spagnolo B, Dubkov A~A and Valenti D 2020 {\em Journal of
  Statistical Mechanics: Theory and Experiment\/} {\bf 2020} 029902
  \urlprefix\url{https://doi.org/10.1088%2F1742-5468%2Fab6f5e}

\bibitem{tsang2019_trivial}
Tsang M 2019 The holevo cramér-rao bound is at most thrice the helstrom
  version (\textit{Preprint} \eprint{1911.08359})

\bibitem{albarelli2019upper}
Albarelli F, Tsang M and Datta A 2019 Upper bounds on the holevo cramér-rao
  bound for multiparameter quantum parametric and semiparametric estimation
  (\textit{Preprint} \eprint{1911.11036})

\bibitem{tsang_v6}
Tsang M, Albarelli F and Datta A 2019 Quantum semiparametric estimation
  (\textit{Preprint} \eprint{1906.09871v6})

\bibitem{young}
Young T~Y 1975 {\em Information Sciences\/} {\bf 9} 25--42

\bibitem{bc94}
Braunstein S~L and Caves C~M 1994 {\em Physical Review Letters\/} {\bf 72} 3439

\bibitem{humphreys2013quantum}
Humphreys P~C, Barbieri M, Datta A and Walmsley I~A 2013 {\em Physical Review
  Letters\/} {\bf 111} 070403

\bibitem{pezze2017optimal}
Pezz{\`e} L, Ciampini M~A, Spagnolo N, Humphreys P~C, Datta A, Walmsley I~A,
  Barbieri M, Sciarrino F and Smerzi A 2017 {\em Physical Review Letters\/}
  {\bf 119} 130504

\bibitem{altenburg2018multi}
Altenburg S and W{\"o}lk S 2018 {\em Physica Scripta\/} {\bf 94} 014001

\bibitem{proctor2018multiparameter}
Proctor T~J, Knott P~A and Dunningham J~A 2018 {\em Physical Review Letters\/}
  {\bf 120} 080501

\bibitem{eldredge2018optimal}
Eldredge Z, Foss-Feig M, Gross J~A, Rolston S~L and Gorshkov A~V 2018 {\em
  Physical Review A\/} {\bf 97} 042337

\bibitem{ge2018distributed}
Ge W, Jacobs K, Eldredge Z, Gorshkov A~V and Foss-Feig M 2018 {\em Physical
  Review Letters\/} {\bf 121} 043604

\bibitem{swd19}
Sekatski P, W{\"o}lk S and D{\"u}r W 2019 Optimal distributed sensing in noisy
  environments (\textit{Preprint} \eprint{1905.06765})

\bibitem{wsu10}
Watanabe Y, Sagawa T and Ueda M 2010 {\em Physical Review Letters\/} {\bf
  104}(2) 020401

\bibitem{wsu11}
Watanabe Y, Sagawa T and Ueda M 2011 {\em Physics Review A\/} {\bf 84} 042121

\bibitem{d2007introduction}
d'Alessandro D 2007 {\em Introduction to quantum control and dynamics\/} (CRC
  press)

\bibitem{wiseman_milburn_2009}
Wiseman H~M and Milburn G~J 2009 {\em Quantum Measurement and Control\/}
  (Cambridge University Press)

\bibitem{jacobs2014quantum}
Jacobs K 2014 {\em Quantum measurement theory and its applications\/}
  (Cambridge University Press)

\bibitem{giovannetti2006quantum}
Giovannetti V, Lloyd S and Maccone L 2006 {\em Physical Review Letters\/} {\bf
  96} 010401

\bibitem{glm11}
Giovannetti V, Lloyd S and Maccone L 2011 {\em Nature Photonics\/} {\bf 5} 222

\bibitem{escher2011general}
Escher B, de~Matos~Filho R and Davidovich L 2011 {\em Nature Physics\/} {\bf 7}
  406

\bibitem{ddkg12}
Demkowicz-Dobrza{\'n}ski R, Ko{\l}ody{\'n}ski J and Gu{\c{t}}{\u{a}} M 2012
  {\em Nature Communications\/} {\bf 3} 1063

\bibitem{matsuzaki2011magnetic}
Matsuzaki Y, Benjamin S~C and Fitzsimons J 2011 {\em Physical Review A\/} {\bf
  84} 012103

\bibitem{chin2012quantum}
Chin A~W, Huelga S~F and Plenio M~B 2012 {\em Physical Review Letters\/} {\bf
  109} 233601

\bibitem{macieszczak2015zeno}
Macieszczak K 2015 {\em Physical Review A\/} {\bf 92} 010102

\bibitem{bpla19}
Bai K, Peng Z, Luo H~G and An J~H 2019 {\em Physical Review Letters\/} {\bf
  123} 040402

\bibitem{yang2019}
Yang Y 2019 {\em Physical Review Letters\/} {\bf 123} 110501

\bibitem{yamagata}
Yamagata K 2011 {\em International Journal of Quantum Information\/} {\bf 9}
  1167--1183

\bibitem{adaptive-channel}
Hayashi M 2009 {\em IEEE Transactions on Information Theory\/} {\bf 55}
  3807--3820

\end{thebibliography}
%%%%%%%%%%%%%%%%%%%%%%%%%%%%%%%%%%%%%%%%%%
\appendix %%%%%%%%%%%%%%%%%%%%%%%%%%%%%%%%%%%%%%%%%%
%%%%%%%%%%%%%%%%%%%%%%%%%%%%%%%%%%%%%%%%%%

\section{Supplemental materials for classical statistics}\label{sec:AppCstat}
\subsection{Locally unbiased estimators} \label{sec:AppCstat1}
In this appendix, we give more detail discussions on the locally unbiased estimators and the Cram\'{e}r-Rao (CR) inequality. 
For a given $d$-parameter model $\cM=\{p_{\bm{\theta}}\,|\,{\bm{\theta}}\in\Theta\}$, consider its 
$n$th iid extension. 
When considering the asymptotic theory of parameter estimation problems, 
one often considers the asymptotically unbiased estimators. This is defined by 
requiring a sequence of estimators to be the locally unbiased in the asymptotic limit $n\to\infty$. 
Importantly, there always exists such an asymptotically unbiased estimator, e.g., the MLE. 

One of the most fundamental results in the parameter estimation theory is the following 
generalized CR inequality: 
Given an i.i.d. (regular) model, the MSE matrix 
of {\bf any estimator} ${\hat{{\bm{\theta}}}}$ obeys the matrix inequality 
\be\label{crgen}
V_{\bm{\theta}}^{(n)}[{\hat{{\bm{\theta}}}}]\ge \frac{1}{n} B_{\bm{\theta}} J_{\bm{\theta}}^{-1} (B_{\bm{\theta}})^{\mathrm T}+ \bm{b}_{\bm{\theta}} (\bm{b}_{\bm{\theta}})^{\mathrm T},
\ee
where $J_{\bm \theta}[\cM]$ is the Fisher information matrix about the model $\cM$ and 
\begin{align*}
B_{\bm{\theta}}[{\hat{{\bm{\theta}}}}]:= \left[\del_j \Eof{{\hat{\theta}_i}(X^n)}\right],\quad
\bm{b}_{\bm{\theta}}[{\hat{{\bm{\theta}}}}]:=\left[ \Eof{{\hat{\theta}_i}(X^n)}-\theta_i \right]^{\mathrm T},
\end{align*}
are called a derivative of bias, or bias matrix, ($d\times d$ matrix) and a bias (vector), respectively.  
Importantly, the biased terms depend on the estimation error in general. 
Since the locally unbiased estimators satisfy $B_{\bm{\theta}}[{\hat{{\bm{\theta}}}}]=I_d$ (the identity matrix) 
and $\bm{b}_{\bm{\theta}}[{\hat{{\bm{\theta}}}}]=\bm{0}$, the CR inequality 
simplifies to $V_{\bm{\theta}}^{(n)}[{\hat{{\bm{\theta}}}}]\ge (J_{\bm{\theta}}[\cM])^{-1} /n$ for any locally unbiased estimator. 

We next turn our attention to the locally unbiasedness condition for the nuisance parameter problems, 
i.e., ${\bm{\theta}}=(\bm{\theta}_{\mathrm{I}},\bm{\theta}_{\mathrm{N}})$. 
Since we are only interested in estimating the parameter of interest $\bm{\theta}_{\mathrm{I}}$, 
we should only require the locally unbiasedness condition for $\bm{\theta}_{\mathrm{I}}$, 
that is defined as follows. See definition \eqref{lu_cond_class}. 
An estimator $\hat{{\bm{\theta}}}_{\mathrm{I}}$ is locally unbiased estimator for $\bm{\theta}_{\mathrm{I}}$, if 
\be \label{app_lucond}
\Eof{{\hat{\theta}}_i(X^n)}=\theta_i\ \mathrm{and}\ \del_j \Eof{{\hat{\theta}}_i(X^n)}=\delta_{i,j}
\ee
are satisfied for $\forall i\in \{1,\dots,\dI\}$ and $\forall j\in \{1,\dots,d\}$ at a point ${\bm{\theta}}$. 
These conditions are expressed in terms of the biased matrix and bias vector as
\[
B_{\bm{\theta}}=\left(\begin{array}{cc}{I_{\dI}} & {0} \\B_1 & B_2\end{array}\right),\quad
\ \bm{b}_{\bm{\theta}}= \left(\begin{array}{c}{0}\\{ b_3}\end{array}\right),
\]
with $B_1,B_2$ nonzero matrices and $b_3$ a non-zero vector in general.   

Set the bias matrix and vector as the following block forms:
\[
B_{\bm{\theta}}=\left(\begin{array}{cc}
B_{\mathrm{I}} &B_{\mathrm{I,N}}  \\
B_{\mathrm{N,I}} & B_{\mathrm{N}}\end{array}\right),\quad
\ \bm{b}_{\bm{\theta}}= \left(\begin{array}{c} \bm{b}_{\mathrm{I}}\\ \bm{b}_{\mathrm{N}}\end{array}\right),
\]
and define the projector onto the subspace of the parameter of interest by 
\[
P_{\mathrm{I}}=\left(\begin{array}{cc}I_{\dI} &0 \\0&0\end{array}\right). 
\]
The CR inequality \eqref{crgen} after projecting onto the relevant subspace becomes
\begin{align}
V_{\bm{\theta};{\mathrm{I}}}^{(n)}[{\hat{{\bm{\theta}}}}]&=P_{\mathrm{I}} V_{\bm{\theta};{\mathrm{I}}}^{(n)}[{\hat{{\bm{\theta}}}}] P_{\mathrm{I}}%\nonumber\\
\ge \frac{1}{n} P_{\mathrm{I}} B_{\bm{\theta}} J_{\bm{\theta}}[p_{\bm{\theta}}]^{-1} (B_{\bm{\theta}})^{\mathrm T}P_{\mathrm{I}}+ P_{\mathrm{I}} \bm{b}_{\bm{\theta}} (\bm{b}_{\bm{\theta}})^{\mathrm T}P_{\mathrm{I}} \nonumber\\
&=\frac{1}{n}\{ B_{\mathrm{I}} 
J_{\bm{\theta}}^{\mathrm{I},\mathrm{I}}B_{\mathrm{I}}^{\mathrm T} +B_{\mathrm{I,N}}J_{\bm{\theta}}^{\mathrm{N},{\mathrm{I}}}B_{\mathrm{I}}^{\mathrm T}
+B_{\mathrm{I}}J_{\bm{\theta}}^{\mathrm{I},{\mathrm{N}}}B_{\mathrm{I,N}}^{\mathrm T} 
+ B_{\mathrm{I,N}}J_{\bm{\theta}}^{\mathrm{N},{\mathrm{N}}}B_{\mathrm{I,N}}^{\mathrm T}\} +\bm{b}_{\mathrm{I}} (\bm{b}_{\mathrm{I}})^{\mathrm T}, \label{mse_genineq}
\end{align}
where the same partitions for the MSE matrix and the inverse of the Fisher information matrix are used. 
Therefore, if we consider the locally unbiased estimator for $\bm{\theta}_{\mathrm{I}}$, i.e.\ $B_{\rm I}=I$, $B_{\mathrm{I},\mathrm{N}}=0$ and $\bm{}b_{\rm I}=0$, we get the result: 
\[
V_{\bm{\theta};{\mathrm{I}}}^{(n)}[{\hat{{\bm{\theta}}}}]\ge\frac{1}{n} J_{\bm{\theta}}^{\mathrm{I},{\mathrm{I}}}. 
\] 

Finally, we mention an important property of locally unbiasedness condition. 
Since we are only interested in estimating parameters of interest $\bm{\theta}_{\mathrm{I}}$, 
it should not matter how we reparametrize the nuisance parameters. 
Consider the following transformation (See also the method of parameter orthogonalization in section \ref{sec:NuiPO}):
\begin{align}
\bm{\theta}_{\mathrm{I}}&=(\theta_i({\bm{\xi}}))=({\xi}_i)\ \mbox{for $i=1,2,\dots,\dI$},\\
\bm{\theta}_{\mathrm{N}}&=(\theta_j({\bm{\xi}}))\quad\mbox{for $j=\dI+1,\dI+2,\dots,d$}. 
\end{align}
With this parametrization, we can show that locally unbiasedness condition \eqref{app_lucond} 
remains unchanged. In other words, arbitrary reparametrization of the nuisance parameter 
does not affect the locally unbiasedness condition for the parameter of interest (Lemma \ref{lem_lucond}). 

\subsection{Three interpretations of classical CR bound} \label{sec:AppCstat2}
In this appendix, we give three different derivations of the classical results \eqref{ccrineq1} and \eqref{ccrineq2}. 
The first one is given in the main text. We try to estimate all parameters under the locally unbiasedness, 
since this is what we can do best. 

The second interpretation is due to Bhapker and others. See \cite{bs94} and references therein. 
This derivation is based on evaluating a Fisher-like information quantity by finding the worst case tangent space. 
The tangent space of the statistical model (manifold) is a vector space spanned by the score functions $\del_i\log p_{\bm{\theta}}(x)$. 
Under the same assumptions and setting as before, we define an information matrix
\begin{align}
J_{\bm{\theta};{\mathrm{I}}}(M)&:=
\big[E_{\bm{\theta}}[u_{\bm{\theta};i}(X|M)\,u_{\bm{\theta};j}(X|M) ] \big]_{i,j\in\{1,2,\dots,\dI\}},\\ 
u_{\bm{\theta};i}(x|M)&:=\del_i \log p_{\bm{\theta}}(x)-\sum_{j=\dI+1}^d  m_{i,j}\del_j\log p_{\bm{\theta}}(x). \label{eqApp:effscore}
\end{align}
Here, $M=[m_{i,j}]$ is a $\dI\times \dN$ real matrix ($\dN=d-\dI$), which can depends on both ${\bm{\theta}}$ and $x$, 
and $u_{\bm{\theta};i}(x|M)$ represents an effective score function in the presence of the nuisance parameter. 
We next define the Fisher information matrix for the parameter of interest by 
minimizing the above information matrix over all possible rectangular matrices $M$: 
\be \nonumber
J_{\bm{\theta};{\mathrm{I}}}:=
\min_{M} \{ J_{\bm{\theta};{\mathrm{I}}}(M) \}, 
\ee
where the minimization is understood in the sense of a matrix inequality. 
Working out some algebras, we can show 
\be
J_{\bm{\theta};{\mathrm{I}}}=J_{\bm{\theta};{\mathrm{I}},{\mathrm{I}}}-J_{\bm{\theta};{\mathrm{I}},{\mathrm{N}}}
\left(J_{\bm{\theta};{\mathrm{N}},{\mathrm{N}}}\right)^{-1}
J_{\bm{\theta};{\mathrm{N}},{\mathrm{I}}}, 
\ee
with the optimal 
$M_*= \arg\min_{M}  J_{\bm{\theta};{\mathrm{I}}}(M)= J_{\bm{\theta};{\mathrm{I}},{\mathrm{N}}}
(J_{\bm{\theta};{\mathrm{N}},{\mathrm{N}}})^{-1}$. 
This is exactly the same as the partial Fisher information 
$J_{\bm{\theta}}({\mathrm{I}}|{\mathrm{N}})$ [\eqref{cpFI}]. 

Note that this method can be extended to a singular model as well. Suppose that 
nuisance parameters are not linearly independent. This situation results in a singular Fisher 
information matrix and we cannot invert the matrix $J_{\bm{\theta};{\mathrm{N}},{\mathrm{N}}}$. 
However, one can use any generalized inverse of $J_{\bm{\theta};{\mathrm{N}},{\mathrm{N}}}$ to 
define the above effective logarithmic likelihood function for the parameter of interest. 

The third derivation is based on the projection method, which is intimately related to 
a geometrical aspect of parameter estimation problems \cite{amari,ANbook}, see also \cite{zr94}. 
Note that the tangent space 
$T_{\bm{\theta}}(\cM)=\mathrm{Span}\{\del_i\log p_{\bm{\theta}}\}_{i=1,\dots,d}$ 
at ${\bm{\theta}}=(\bm{\theta}_{\mathrm{I}},\bm{\theta}_{\mathrm{N}})$ 
cannot be expressed as a direct sum of two tangent spaces,  
$T_{\bm{\theta};{\mathrm{I}}}(\cM)=\mathrm{Span}\{\del_i\log p_{\bm{\theta}}\}_{i=1,\dots,\dI}$ and 
$T_{\bm{\theta};{\mathrm{N}}}(\cM)=\mathrm{Span}\{\del_i\log p_{\bm{\theta}}\}_{i=\dI+1,\dots,d}$, 
unless two parameter are orthogonal with respect to the Fisher information. 
As we discussed in Appendix \ref{sec:AppCstat1}, the reparametrization of the nuisance parameter 
does not matter as long as we wish to estimate the parameter of interest under the locally unbiasedness condition. 
We can always find a new coordinate system such that $T_{\bm{\theta}}(\cM)=T_{\bm{\theta};{\mathrm{I}}}(\cM)\oplus 
T_{\bm{\theta};{\mathrm{N}}}(\cM)$. 
Geometrically speaking, we are introducing a foliation structure for the statistical model \cite{amari,ANbook}. 
Owing to this geometrical structure, the nuisance parameter degree of freedom can be 
used to define an ancillary submanifold. A condition $\bm{\theta}_{\mathrm{I}}=c_{\mathrm{I}}$ (constant) defines 
a submanifold of $\cM$ for each $\bm{\theta}_{\mathrm{I}}$. 
The problem is then equivalent to inferring statistical submodels defined by $\bm{\theta}_{\mathrm{I}}=c_I$ condition. 
A standard orthogonalization is given by \eqref{localortho}, and it is straightforward 
to see that the Fisher information matrix in the new coordinate system becomes
identical to the partial Fisher information matrix \eqref{cpFI}. Achievability and 
efficiency can also be easily analyzed in the language of information geometry \cite{ANbook,amari}.  

\subsection{Parameter transformation and estimating a function of parameters}\label{sec:AppCstat3}
In this appendix, we summarize how the change of parameters reflects the CR inequality 
and its application to estimate a function of parameters. For simplicity, we only concern the case 
when the sample size is one without loss of generality. 

Let us start with a statistical model with $d$ parameters. 
$\cM=\{ p_{\bm \theta}\,|\, \bm{\theta}\in\Theta\}$. 
If we transform the parameter $\bm{\theta}=(\theta_1,\theta_2,\ldots,\theta_d)$ to 
a new parameter $\bm{\xi}=(\xi_1,\xi_2,\ldots,\xi_d)$, the model is now parametrized as
$\cM=\{ p_{\bm \xi}\,|\, \bm{\xi}\in\Xi\}$. Geometrically speaking, this corresponds to introduce a new coordinate system to a point $p\in\cM$. To have a well-defined parametrization in the new parameter $\bm{\xi}$, 
we need impose several conditions. 
Among them, the mapping $\bm{\theta}\mapsto\bm{\xi}$ needs to be $C^{r}$ diffeomorphism for sufficiently large $r$. 
In other words, it is a one-to-one mapping, and 
each function $\xi_\alpha(\bm{\theta})$ for $\alpha=1,2,\ldots,d$ is $C^{r}$-class. 
Further, its inverse function $\theta_i(\bm{\xi})$ for $i=1,2,\ldots,d$ is also $C^{r}$-class. 
Important consequence of this requirement is that the Jacobi matrix for this transformation is full rank and is invertible. 
Here, the Jacobi matrix is defined by
\be
\frac{\del \bm{\theta}}{\del \bm{\xi}}:= \left[ \frac{\del \theta_i}{\del \xi_\alpha} \right],
\ee
where $i$ and $\alpha$ correspond to the row and column indices, respectively. 
Its inverse matrix is 
\be
\frac{\del \bm{\xi}}{\del \bm{\theta}}:= \left[ \frac{\del \xi_\alpha}{\del \theta_i} \right],
\ee
with the column index $i$ and the row index $\alpha$. 

Under this transformation, the partial derivatives $\frac{\del }{\del \theta_i}$ and $\frac{\del }{\del \xi_\alpha}$ are transformed as
\begin{align}
\frac{\del }{\del \xi_\alpha}&= \sum_{i=1}^d \frac{\del \theta_i}{\del \xi_\alpha} \frac{\del }{\del \theta_i},\\
\frac{\del }{\del \theta_i}&= \sum_{\alpha=1}^d \frac{\del \xi_\alpha}{\del \theta_i} \frac{\del }{\del \xi_\alpha}. 
\end{align}
Accordingly, the classical Fisher information matrix is transformed as
\begin{equation}
J_{\bm \xi}= \frac{\del \bm{\theta}}{\del \bm{\xi}} \ J_{\bm \theta}\: \left(\frac{\del \bm{\theta}}{\del \bm{\xi}}\right)^{\rm T}.
\end{equation}

Let $\hat{\bm{\xi}}=(\hat{\xi}_1,\hat{\xi}_2,\ldots,\hat{\xi}_d)$ be an estimator for the new parameter. 
One naively expects that a good estimator for $\bm{\theta}$ is also a good estimator for $\bm{\xi}$ 
when it is transformed . 
However, this statement is true only in the asymptotic limit. 
Importantly, the unbiasedness condition in the $\bm{\xi}$ parametrization takes a different form as
\be
E_{\bm{\xi}}[ \hat{\xi}_\alpha(X)]=\xi_\alpha, 
\ee 
for $\alpha=1,2,\ldots,d$. 
As a consequence, an unbiased estimator $\hat{\bm{\theta}}$ for $\bm{\theta}$ is no longer unbiased for $\bm{\xi}$ 
when transformed into the new parametrization, i.e., 
the estimator $\bm{\xi}\circ\hat{\bm{\theta}}$ is biased. 
There are several methods known in statistics to remove bias \cite{lc,bnc,ANbook}. 
Owing to the continuous mapping theorem \cite{lc,bnc}, if $\hat{\bm{\theta}}$ is weakly consistent, 
$\bm{\xi}( \hat{\bm{\theta}})$ converges to $\bm{\xi}(\bm{\theta})$ in probability. 
The above statement about non-invariance of unbiasedness also holds for the locally unbiasedness condition. 

The CR inequality for estimating the new parameter $\bm{\xi}$ is expressed as
\be\label{cr_xi}
V_{\bm{\xi}}[\hat{{\bm{\xi}}}]\ge (J_{\bm \xi})^{-1}
= \left(\frac{\del \bm{\xi}}{\del \bm{\theta}}\right)^{\rm T}\, (J_{\bm \theta})^{-1}\;\frac{\del \bm{\xi}}{\del \bm{\theta}} ,
\ee
for all locally unbiased estimators at $\bm{\xi}$. We can also derive the generalized version of the CR inequality \eqref{crgen} 
when $\hat{\bm{\xi}}$ is not (locally) unbiased. 
From this expression, it holds that the weighted trace of the MSE matrix is bounded by
\begin{align}
\Tr{W V_{\bm{\xi}}[{\hat{{\bm{\xi}}}}]}&\ge \Tr{ \frac{\del \bm{\xi}}{\del \bm{\theta}}\;W\left(\frac{\del \bm{\xi}}{\del \bm{\theta}}\right)^{\rm T}\, (J_{\bm \theta})^{-1}}\\
&= \Tr{\widetilde{W}  (J_{\bm \theta})^{-1}}. 
\end{align}
Thereby, we immediately see that the parameter transformation corresponds to the change in the weight matrix 
$\widetilde{W}:=  \frac{\del \bm{\xi}}{\del \bm{\theta}}\;W(\frac{\del \bm{\xi}}{\del \bm{\theta}})^{\rm T}$. 
This fact is an important property of the parameter transformation in the context of quantum state estimation. 

Suppose we are interested in estimating a vector-valued function of $\bm{\theta}$,
\be
\bm{g}(\bm{\theta}):=\left(g_1(\bm{\theta}), g_2(\bm{\theta}),\ldots, g_K(\bm{\theta})\right), 
\ee
where $K$ should be smaller or equal to the number of parameters $d$ for mathematical convenience. 
$g_k(\bm{\theta})$ ($k=1,2,\ldots,K$) are also assumed to be differentiable and continuous. 
Define a rectangular matrix 
\be
G_{\bm \theta}:=\left[ \frac{\del g_k(\bm{\theta})}{\del \theta_i} \right], 
\ee
where $k=1,2,\ldots,K$ is the row index and $i=1,2,\ldots,d$ is the column index. 
Let $\hat{\bm{g}}=(\hat{g}_1, \hat{g}_2,\ldots, \hat{g}_K)$ be an estimator estimating the vector-valued function. 
We can define the locally unbiasedness condition at $\bm{\theta}$ by 
\be
E_{\bm{\theta}}[ \hat{g}_k(X)]=g_i(\bm{\theta}),\ \mathrm{and}\ \frac{\del}{\del \theta_i} E_{\bm{\theta}}[\hat{g}_k(X)]
=\frac{\del g_k}{\del \theta_i}, 
\ee
for $\forall k=1,2,\ldots,K$ and $\forall i=1,2,\ldots,d$. 
Following the same argument to derive \eqref{cr_xi}, 
it is straightforward to derive the CR inequality for a locally unbiased estimator as
\be
V_{\bm{\theta}}[\hat{\bm{g}}]\ge \left(G_{\bm \theta}\right)^{\rm T}\, (J_{\bm \theta})^{-1}\;G_{\bm \theta}, 
\ee
where $V_{\bm{\theta}}[\hat{\bm{g}}]:=\left[E_{\bm \theta}[\left(\hat{g}_k(X)-g_k(\bm{\theta})\right) \left(\hat{g}_{k'}(X)-g_{k'}(\bm{\theta})\right) ]  \right]$ denotes the MSE matrix for estimating the vector-valued function. 

In general, it is not easy to construct a locally unbiased estimator from $\hat{\bm \theta}$. 
For biased estimators instead, we obtain the following CR inequality for any estimators $\hat{\bm{g}}$. 
\be\label{cr_function}
V_{\bm{\theta}}[\hat{\bm{g}}]\ge 
 \left(B_{\bm \theta}\right)^{\rm T}\, (J_{\bm \theta})^{-1}\;B_{\bm \theta}+ \bm{b}_{\bm{\theta}} (\bm{b}_{\bm{\theta}})^{\mathrm T} .
\ee
Here, the $K\times d$ rectangular matrix $B_{\bm{\theta}}$ and the $K$-column vector $\bm{g}$ (the bias vector) are defined by
\begin{align}
B_{\bm{\theta}}[\hat{\bm{g}}]&:= \left[\frac{\del}{\del \theta_j} E_{\bm \theta}[\hat{g}_k(X)]\right]
=\left[\frac{\del}{\del \theta_j} b_{\bm{\theta};k}[\hat{\bm{g}}]\right]+G_{\bm \theta},\\
b_{\bm{\theta};k}[\hat{\bm{g}}] &:=E_{\bm{\theta}}[ \hat{g}_k(X)]-g_k(\bm{\theta}). 
\end{align}
In contrast to the usual CR inequality, the achievability of this bound depends on 
the nature of the vector-valued function $\bm{g}$. 
See for example \cite{rm97}. 
Note that the right hand side of the above CR inequality \eqref{cr_function} still depends on the estimator $\hat{\bm{g}}$, 
unless it is unbiased. 

As a special case, consider a scalar function $g(\bm{\theta})$. 
Then, the CR inequality for estimating $g(\bm{\theta})$ is expressed as
\be\label{cr_singlefunction}
V_{\bm{\theta}}[\hat{g}]\ge 
 \bm{v}_{\bm \theta}^{\rm T}\, (J_{\bm \theta})^{-1}\; \bm{v}_{\bm \theta},
\ee
for any locally unbiased estimator $\hat{g}$ with $\bm{v}_{\bm \theta}^{\rm T}:=\left(\frac{\del g(\bm{\theta})}{\del \theta_1},\frac{\del g(\bm{\theta})}{\del \theta_2},\ldots,\frac{\del g(\bm{\theta})}{\del \theta_d}\right) $. 

\section{Supplemental materials for quantum statistics}\label{sec:AppQstat}

\subsection{CR inequality}\label{PfCR}
This subsection shows the SLD and RLD CR inequalities
\eqref{CRSLDF} and \eqref{CRRLDF}.
Also, it
derives the equality condition of the SLD CR inequality \eqref{CRSLDF}.

First, we show the SLD CR inequality \eqref{CRSLDF}.
Let $\hat{\Pi}$ be
a locally unbiased estimator at ${\bm{\theta}}$.
For two $d$-dimensional real vectors ${\bm a},{\bm b}$, 
we show the following inequality
\be
\Big({\bm b}^{\rm T} V_{\bm{\theta}}[\hat{\Pi}] {\bm b}\Big) 
\Big({\bm a}^{\rm T} J_{{\bm{\theta}}}^{\rm S} {\bm a}\Big)
\ge
({\bm b}^{\rm T} {\bm a})^2.\label{EH1}
\ee
Define the Hermitian matrices 
$O_{\bm b}:=
\sum_{x \in {\cal X}} 
({\bm b}^{\rm T}(\hat{{\bm{\theta}}}(x)-{\bm{\theta}})) \Pi_x$
and 
$L_{\bm a}:=\sum_{j=1}^d a_j L_{{\bm{\theta}};j}^{\rm S}$.
The relations
\begin{align}
&{\bm b}^{\rm T} {\bm a}= 
\frac{1}{2}\tr {O_{\bm b} (\rho_{\bm{\theta}} L_{\bm a} + L_{\bm a} \rho_{\bm{\theta}} )},
\quad
{\bm a}^{\rm T} J_{{\bm{\theta}}}^{\rm S} {\bm a}
=
\frac{1}{2}\tr {L _{\bm a}(\rho_{\bm{\theta}} L_{\bm a} + L_{\bm a} \rho_{\bm{\theta}} )} ,\label{EH2}\\
&
\begin{aligned}
{\bm b}^{\rm T} V_{\bm{\theta}}[\hat{\Pi}] {\bm b}
-\tr {\rho_{\bm{\theta}} O_{\bm b}^2}
&=\sum_{x \in {\cal X}} 
\tr {\rho_{\bm{\theta}}
(O_{\bm b}-({\bm b}^{\rm T} (\hat{{\bm{\theta}}}(x)-{\bm{\theta}}))I)
\Pi_x
(O_{\bm b}-({\bm b}^{\rm T} (\hat{{\bm{\theta}}}(x)-{\bm{\theta}}))I)} \\
&\ge 0 
\end{aligned}
\label{EH3}
\end{align}
hold.
We apply Schwartz inequality for the inner product $X,Y\mapsto 
\langle X, Y \rangle_{\rho_{\bm{\theta}}}^{\rm S}:=
\frac{1}{2}\tr {X (\rho_{\bm{\theta}} Y +
Y \rho_{\bm{\theta}} )}$
to the case with $X= O_{\bm b}$ and $Y=L_{\bm a}$.
The combination with \eqref{EH2} and \eqref{EH3}
implies \eqref{EH1}.

The substitution of ${\bm a}=(J_{{\bm{\theta}}}^{\rm S})^{-1} {\bm b}$ into 
\eqref{EH1}
yields the inequality
\be
{\bm b}^{\rm T} V_{\bm{\theta}}[\hat{\Pi}] {\bm b} \ge
{\bm b}^{\rm T} (J_{{\bm{\theta}}}^{\rm S})^{-1} {\bm b}.\label{EH4}
\ee
Since ${\bm b}$ is an arbitrary $d$-dimensional real vector, we obtain
\eqref{CRSLDF} \cite{helstrom}.
The RLD CR inequality \eqref{CRRLDF} can be shown as follows.
Replacing 
$J_{{\bm{\theta}}}^{\rm S}$, $L_{{\bm{\theta}};j}^{\rm S}$, 
and the inner product $\langle X, Y \rangle_{\rho_{\bm{\theta}}}^{\rm S}$
by $J_{{\bm{\theta}}}^{\rm{R}}$, $L_{{\bm{\theta}};j}^{\rm R}$, and the inner product
$\langle X, Y \rangle_{\rho_{\bm{\theta}}}^{\rm R}:=
\tr X^\dagger \rho_\theta Y$
and extending the range of vectors 
${\bm a}$ and ${\bm b}$ to $d$-dimensional complex vectors, we have 
the inequality
\be
\bar{{\bm b}}^{\rm T} V_{\bm{\theta}}[\hat{\Pi}] {\bm b} \ge
\bar{{\bm b}}^{\rm T} (J_{{\bm{\theta}}}^{\rm R})^{-1} {\bm b}.\label{EH4RLD}
\ee
because the components of $J_{{\bm{\theta}}}^{\rm{R}}$ have complex numbers in general \cite{holevo}.
Since ${\bm b}$ is an arbitrary $d$-dimensional complex vector, we obtain
\eqref{CRRLDF} \cite{holevo}. 

Next, we show the equality condition of \eqref{CRSLDF}.
In the following, we denote 
$\sum_{j=1}^d(J_{{\bm{\theta}}}^{\rm S})^{-1}_{i,j} L_{{\bm{\theta}};j}^{\rm S}$ 
by $L_{{\bm{\theta}}}^{\mathrm{S};i}$.
The equality in the above application of Schwartz inequality
holds iff
$O_{\bm b}$ is a constant times of $L_{(J_{{\bm{\theta}}}^{\rm S})^{-1} {\bm b}}$
for any vector ${\bm b}$. 
The combination of this equality condition and 
the locally unbiased condition implies that 
$O_{\bm b}=L_{(J_{{\bm{\theta}}}^{\rm S})^{-1} {\bm b}}$ with any vector ${\bm b}$,
i.e.,
$\sum_{x \in {\cal X}} 
(\hat{\theta}_i(x)-\theta_i) \Pi_x=L_{{\bm{\theta}}}^{\mathrm{S};i}$ with $i=1,\ldots,d$.
Therefore, 
the equality in \eqref{CRSLDF} holds if and only if
(i) the equality in \eqref{EH3} 
holds for any vector ${\bm b}$,
and (ii) the equation
$\sum_{x \in {\cal X}} 
(\hat{\theta}_i(x)-\theta_i) \Pi_x=L_{{\bm{\theta}}}^{\mathrm{S};i}$ 
holds with $i=1,\ldots,d$.

When we can choose SLDs $L_{{\bm{\theta}};i}^{\rm S}$ for $i=1, \ldots, d$ 
such that these SLDs $L_{{\bm{\theta}};i}^{\rm S}$ are commutative with each other,
we choose a POVM $\hat{\Pi}$ as 
the simultaneous spectral decomposition of
$L_{{\bm{\theta}}}^{\mathrm{S};i}+\theta_i$ with $i=1,\ldots,d$.
Then, the condition (ii) holds.
Since 
${\Pi}_x$ is the projection to 
the common eigenspaces of 
$L_{{\bm{\theta}}}^{\mathrm{S};1}, \ldots,L_{{\bm{\theta}}}^{\mathrm{S};d}$,
the equality in \eqref{EH3} holds for any vector ${\bm b}$.
Hence, the equality in \eqref{CRSLDF} holds \cite{ANbook}.

Conversely, we assume that 
a locally unbiased estimation $\hat{\Pi}$ satisfies the equality in \eqref{CRSLDF}.
Then, the equation
$\sum_{x \in {\cal X}} 
(\hat{\theta}_i(x)-\theta_i) \Pi_x=L_{{\bm{\theta}}}^{\mathrm{S};i}$ 
holds with $i=1,\ldots,d$.
When $\rho_{\bm{\theta}} $ is strictly positive, 
the equality in \eqref{EH3} 
implies the relation
$(O_{\bm b}-({\bm b}^{\rm T}(\hat{{\bm{\theta}}}(x)-{\bm{\theta}}))I)
\Pi_x
(O_{\bm b}-({\bm b}^{\rm T} (\hat{{\bm{\theta}}}(x)-{\bm{\theta}}))I)
=0$.
This relation with any vector ${\bm b}$ guarantees that 
${\Pi}_x$ is the projection to 
the common eigenspaces of $L_{{\bm{\theta}}}^{\mathrm{S};1}, \ldots,L_{{\bm{\theta}}}^{\mathrm{S};d}$ \cite{ANbook}.

However, when $\rho_{\bm{\theta}} $ is not strictly positive, 
the relation
$(O_{\bm b}-({\bm b}^{\rm T}(\hat{{\bm{\theta}}}(x)-{\bm{\theta}}))I)
\Pi_x
(O_{\bm b}-({\bm b}^{\rm T}(\hat{{\bm{\theta}}}(x)-{\bm{\theta}}))I)
=0$ does not hold in general.
Hence, we cannot apply this discussion to the equality in \eqref{EH3}.
However, we can say the following even in this case.
The equality holds in \eqref{CRSLDF}
iff we can choose SLDs $L_{{\bm{\theta}};i}^{\rm S}$ on a sufficiently large extended Hilbert space ${\cal H}'$
for $i=1, \ldots, d$ 
such that these SLDs $L_{{\bm{\theta}};i}^{\rm S}$ are commutative with each other. 

To show the above equivalence relation, we assume that a locally unbiased estimator $\hat{\Pi}$ 
satisfies the equality in \eqref{CRSLDF}.
Then, we extend the Hilbert space 
${\cal H}'$ with a projection $P$ to the original space ${\cal H}$
to satisfy the following conditions.
There exists a locally unbiased estimator $\hat{\Pi}'$ 
on ${\cal H}'$ such that $\hat{\Pi}'_x$ is a projection and 
$P{\Pi}'_x P={\Pi}_x$ for $x \in {\cal X}$.
Notice that ${\Pi}_x$ is not a projection in general.
In the following, our discussion is made on the larger Hilbert space ${\cal H}'$
based on the following equivalent class;
Two Hermitian matrices $X$ and $Y$ are equivalent 
when the norm $\|X-Y\|_{\bm{\theta}}:=\sqrt{\langle X-Y, X-Y \rangle_{\rho_{\bm{\theta}}}^{\rm S}}$
is zero.
Since the equality in \eqref{EH3} holds,
the equality condition of Schwartz inequality guarantees that
$\sum_{x \in {\cal X}} 
(\hat{\theta}_i(x)-\theta_i) \Pi_x'$ equals $L_{{\bm{\theta}}}^{\mathrm{S};i}$ 
in the sense of the above equivalent class 
for $i=1,\ldots,d$.
Thus,
$
{L_{{\bm{\theta}};j}^{\rm S}}':=
\sum_{i=1}^d
J_{{\bm{\theta}};i,j}^{\rm S}\sum_{x \in {\cal X}} 
(\hat{\theta}_i(x)-\theta_i) \Pi_x'$ equals $L_{{\bm{\theta}};j}^{\rm S}$ 
in the sense of the above equivalent class 
for $i=1,\ldots,d$.
${L_{{\bm{\theta}};1}^{\rm S}}',\ldots {L_{{\bm{\theta}};d}^{\rm S}}'$
 are commutative with each other,
 and can be regarded as SLDs.
Thus, 
we can conclude the following.
When the equality in \eqref{CRSLDF} holds,
we can choose SLDs $L_{{\bm{\theta}};i}^{\rm S}$ for $i=1, \ldots, d$ 
such that these SLDs $L_{{\bm{\theta}};i}^{\rm S}$ are commutative with each other.

\subsection{Useful lemmas}
In this subsection, we give several known lemmas concerning the theory of quantum state estimation.  

First lemma concerns the classical Fisher information about 
a projection measurement for a linear combination of the SLD operators, see for example \cite{GM00,yamagata}. 
Given an $d$-parameter model $\cM=\{\rho_{\bm{\theta}}|{\bm{\theta}}\in\Theta\}$, 
consider a set of the SLD operators $\{\SLD{1},\SLD{2},\dots,\SLD{d} \}$, 
and define the following Hermitian operator. 
\be\label{def_Lv}
L_{\bm{v}}:=\sum_{i=1}^d v_i\SLD{i},
\ee
where $\bm{v}=(v_1,\ldots, v_d)^{\rm T}\in\bbr^d$ is an arbitrary $d$-dimensional vector. 
Then, we can consider a projection measurement $\Pi_v$ defined by this observable. 
\begin{lemma}\label{lem_optPVM}
Let $\bm{v}=(v_1,\ldots, v_d)^{\rm T}$ be an arbitrary $d$-dimensional real vector, 
and define the operator $L_{\bm{v}}$ by \eqref{def_Lv}. 
The Fisher information matrix about a projection measurement $\Pi_{\bm{v}}$ 
for the operator $L_{\bm{v}}$ satisfies 
\be
\bm{v}^{\rm T}J_{\bm{\theta}}[\Pi_v] \bm{v}= \bm{v}^{\rm T}\sldQFI \bm{v}.
%\sum_{i,j=1}^d v_i J_{\bm{\theta}}[\Pi_v] v_j= \sum_{i,j=1}^d v_i \sldQFI v_j. 
\ee
\end{lemma}
\begin{proof}
Let $\del_i:=\del/\del\theta_i$ be the partial derivative with respect to $\theta_i$. 
For a POVM $\Pi=\{\Pi_x\}_{x\in\cX}$, let $p_{\bm{\theta}}(x|\Pi)=\tr{\rho_{\bm{\theta}}\Pi_x}=\sldin{I_d}{\Pi_x}$ 
be the probability distribution, and the $i$th score function 
$u_{{\bm{\theta}};i}(x)=\del_i \log p_{\bm{\theta}}(x)$ is expressed as
\be
u_{{\bm{\theta}};i}(x)=\frac{\tr{\del_i\rho_{\bm{\theta}}\Pi_x}}{p_{\bm{\theta}}(x|\Pi)}=\frac{\sldin{\SLD{i}}{\Pi_x}}{\sldin{I_d}{\Pi_x}}, 
\ee
where the relation $\tr{\del_i\rho_{\bm{\theta}} X}=\sldin{\SLD{i}}{X}$ for $X\in\lofh$ holds from definition of the SLD operator. 
Using this representation of the classical score function, we can express the Fisher information matrix as
\begin{align}
J_{{\bm{\theta}};i,j}[\Pi]=\sum_{x\in\cX} p_{\bm{\theta}}(x|\Pi) u_{{\bm{\theta}};i}(x) u_{{\bm{\theta}};j}(x)
=\sum_{x\in\cX}\frac{\sldin{\SLD{i}}{\Pi_x}\sldin{\SLD{j}}{\Pi_x}}{\sldin{I_d}{\Pi_x}}. 
\end{align}
Now, set $\Pi=\Pi_{\bm{v}}$ and denote the spectral decomposition of the operator $L_{\bm{v}}$ by $L_{\bm{v}}=\sum_{x\in\cX}\lambda_x E_x$. 
Then, we have 
\begin{align*}
&\sum_{i,j=1}^d v_i (J_{\bm \theta}[\Pi_{\bm v}])_{i,j} v_j=\sum_{x\in\cX}\frac{\sldin{L_{\bm{v}}}{E_x}\sldin{L_{\bm{v}}}{E_x}}{\sldin{I_d}{E_x}}% \\
=\sum_{x\in\cX}\frac{\lambda_x^2(\sldin{I_d}{E_x})^2}{\sldin{I_d}{E_x}}  \\
=&\sum_{x\in\cX}\lambda_x^2\sldin{I_d}{E_x}
=\sldin{I_d}{L_{\bm{v}}^2}=\sldin{L_{\bm{v}}}{L_{\bm{v}}}
=\sum_{i,j=1}^d v_i (J_{\bm \theta}^{\rm S})_{i,j} v_j. %\qquad\square
\end{align*}
\end{proof}

We remark that the measurement $\Pi_{\bm{v}}$ depends on the choice of the vector ${\bm{v}}$, 
and hence this lemma does not prove the relationship $J_{\bm{\theta}}[\Pi]=\sldQFI$ 
unless all SLD operators commute with each other. 
In the context of the nuisance parameter problem, an important case is 
to estimate the single parameter $\theta_1$ in the presence of the nuisance parameters $\bm{\theta}_{\rm{N}}$. 
By setting ${\bm{v}}=\bm{e}_1=(1,0,\dots,0)$, we immediately obtain $J_{{\bm{\theta}};1,1}[\Pi_{\bm{e}_1}]=J_{{\bm{\theta}};1,1}^{\rm S}$. 
Similarly, we can show that the Fisher information for each diagonal element 
can attain $J_{{\bm{\theta}};ii}$ by the projection measurement $L_{\bm{e}_i}$ 
with the standard basis vector $\bm{e}_i$. 

Next corollary shows that the MSE matrix is bounded by the quantum MSE matrix. 
For a POVM $\Pi=\{\Pi_x\}_{x\in\cX}$ and an estimator $\hat{{\bm{\theta}}}=(\hat{\theta}_1,\ldots, \hat{\theta}_d):\cX\to\Theta\subset\bbr^d$, 
we introduce a $d$-valued observable:
\be
\Pi(\hat{{\bm{\theta}}}):=\sum_{x\in\cX} \hat{{\bm{\theta}}}(x)\Pi_x,  
\ee
and denote its $i$th element by $\hat{\Pi}_i=\sum_{x\in\cX} {\hat{\theta}_i}(x)\Pi_x$. 
The symmetrized quantum covariance matrix \cite{ANbook,petz} is defined by 
\be
V_{\bm{\theta}}^Q[\Pi(\hat{{\bm{\theta}}})]:= \big[ \sldin{\hat{\Pi}_i -\theta_iI_d}{\hat{\Pi}_j-\theta_jI_d} \big]. 
\ee
It is easy to verify that this matrix is a $d\times d$ real positive-semidefinite matrix. 

\if0
The following lemma due to Holevo is fundamental \cite{holevo}.
\begin{lemma}\label{lem_holevo}
Given a POVM $\Pi=\{\Pi_x\}_{x\in\cX}$ on $\cH$ and a real value function $f:\cX\to\bbr$, 
the operator inequality 
\be
\sum_{x\in\cX}f(x)^2 \Pi_x\ge {\Pi}_f^2,
\ee
holds, where ${\Pi}_f:=\sum_{x\in\cX}f(x)\Pi_x$ is Hermitian but is not positive in general. 
\end{lemma}
\begin{proof}
With these notations, $[{\Pi}_f-f(x)I_d]\Pi_x[{\Pi}_f-f(x)I_d]$ 
is a positive operator. Summing over the index $x\in\cX$ gives 
\be
\sum_{x\in\cX}[{\Pi}_f-f(x)I_d]\Pi_x[{\Pi}_f-f(x)I_d]\ge0 
\lra 
\sum_{x\in\cX}f(x)^2 \Pi_x- {\Pi}_f^2\ge0. %\quad\square
\ee
\end{proof}
\fi

Now, we recall \eqref{EH3}.
Since
${\bm b}^{\rm T} V_{\bm{\theta}}^Q[\Pi(\hat{{\bm{\theta}}})] {\bm b}
=\tr {\rho_{\bm{\theta}} O_{\bm b}^2}$, \eqref{EH3} yields the following corollary. 

\begin{corollary} \label{cor_holevo}
Given a quantum parametric model, the MSE matrix and SLD covariance matrix satisfy 
the matrix inequality for any estimator $\hat{\Pi}=(\Pi,\hat{{\bm{\theta}}})$: 
\be
V_{\bm{\theta}}[\hat{\Pi}]\ge V_{\bm{\theta}}^Q[\Pi(\hat{{\bm{\theta}}})]. 
\ee 
\end{corollary}
\if0
\begin{proof}
Let $v=(v_i)\in\bbr^d$ be arbitrary real vector and applying Lemma \ref{lem_holevo} 
together with the formula $X\ge 0\Rightarrow \tr{X\rho}$, we get
\begin{align}
\sum_{i,j=1}^d v_i V_{{\bm{\theta}};i,j}[\hat{\Pi}]v_j
&=\sum_{x\in\cX} [\sum_i v_i({\hat{\theta}_i}(x)-\theta_i)]^2\tr{\rho_{\bm{\theta}}\Pi_x}\\
&\ge \tr{\rho_{\bm{\theta}} \big[\sum_{x\in\cX} \sum_i v_i({\hat{\theta}_i}(x)-\theta_i)\Pi_x \big]^2}\\
&=\tr{\rho_{\bm{\theta}} \big[ \sum_i v_i(\hat{\Pi}_i-\theta_iI_d) \big]^2}\\
&=\sldin{\sum_i v_i(\hat{\Pi}_i-\theta_iI_d)}{\sum_j v_j(\hat{\Pi}_j-\theta_jI_d)}\\
&=\sum_{i,j=1}^d v_i  V_{{\bm{\theta}};ij}^Q[\Pi(\hat{{\bm{\theta}}})] v_j. 
\end{align}
\end{proof}
\fi
We remark here that the above construction for the SLD covariance matrix 
can also be extended to other inner products on $\cH$, known as the quantum covariance matrix \cite{petz}. 
In particular, the RLD covariance matrix can be defined similarly by using the right inner product. 

The last lemma is also well known. 
\begin{lemma} \label{lemm_crgen}
The SLD quantum covariance matrix for arbitrary estimator $\hat{\Pi}$ 
satisfies the generalized SLD quantum CR inequality:
\be\label{sldCR_gen}
V_{\bm{\theta}}^Q[\Pi(\hat{{\bm{\theta}}})]\ge B_{\bm{\theta}}[\hat{\Pi}] (\sldQFI)^{-1} B_{\bm{\theta}}[\hat{\Pi}]^{\mathrm{T}},  
\ee
where $d\times d$ matrix $B_{\bm{\theta}}[\hat{\Pi}]$ denotes the bias matrix, which is defined by
\be
B_{\bm{\theta}}[\hat{\Pi}]:=\big[ \sldin{\SLD{i}}{\hat{\Pi}_j} \big]. 
\ee
\end{lemma}
\begin{proof}
Here, we follow \cite{petz} that utilizes a property of positive matrix theory. 
Define the following $2d\times 2d$ block matrix based on the symmetric inner product. 
\be
M:=\left(\begin{array}{cc} V_{\bm{\theta}}^Q[\Pi(\hat{{\bm{\theta}}})] & B_{\bm{\theta}}[\hat{\Pi}]  \\ B_{\bm{\theta}}[\hat{\Pi}]^{\mathrm{T}} & \sldQFI\end{array}\right). 
\ee
Noting that this matrix is also expressed as
\begin{align}
M&=\big[ \sldin{m_a}{m_b} \big]_{a,b=1,2,\dots,2d},\\
m_a&=\begin{cases}
\hat{\Pi}_a-\theta_a I_d&\ (a=1,2,\dots,d),\\
\SLD{a-n}&\ (a=d+1,\dots,2d)
\end{cases},
\end{align}
we see $M$ is a positive semi-definite matrix. 
From the standard argument in the positive matrix theory, 
$M$ is positive if and only if 
$V_{\bm{\theta}}^Q[\Pi(\hat{{\bm{\theta}}})]- B_{\bm{\theta}}[\hat{\Pi}] (\sldQFI)^{-1} B_{\bm{\theta}}[\hat{\Pi}]^{\mathrm{T}}$ 
is positive (see, for example, Theorem 1.3.3 of \cite{bhatia}). 
\end{proof}

\subsection{Characterization of quantum parametric models} \label{sec:Qstat_model}
In this Appendix, we discuss characterization of the D-invariant model and the asymptotically classical model. 
%Details of proofs can be found in \cite{js16,js18_clmodel}. 
Consider a $d$-parameter model $\cM=\{\rho_{\bm{\theta}}\,|\,{\bm{\theta}}\in\Theta\}$ satisfying regularity conditions. 
The following lemma is known \cite{js16,js18_clmodel}.
% that the Holevo bound coincides with the RLD CR bound if and only if the model is D-invariant. 

\begin{lemma}\label{DRH}
A model is D-invariant 
if and only if
the Holevo bound $C_{\bm{\theta}}^H[W,\cM]$ is identical to the RLD-CR bound $C_{\bm{\theta}}^{\rm{R}}[W,\cM]$
for any weight matrix $W>0$.
\end{lemma}

The next fact concerns the asymptotic achievability of the SLD CR bound \cite{js16,rjdd16,js18_clmodel}: 

\begin{lemma}\label{DSH}
A model is asymptotically classical 
if and only if
the Holevo bound $C_{\bm{\theta}}^H[W,\cM]$ is identical to the SLD-CR bound $C_{\bm{\theta}}^{\rm{S}}[W,\cM]$
for any weight matrix $W>0$.
\end{lemma} 

Below, we shall list several equivalent conditions for the D-invariant model and the asymptotically classical model. 

First, let $\cD{\rho_{\bm{\theta}}}$ be the commutation operator at ${\bm{\theta}}$ defined by \eqref{def:CommOp}. 
Then, we have the equivalent relations for the D-invariant model. 
\begin{enumerate}
\item The model $\cM$ is D-invariant at ${\bm{\theta}}$. 
\item $\forall i$, $\cD{\rho_{\bm{\theta}}}(\SLD{i})\in T_{\bm{\theta}}(\cM)$. (Definition)
\item $\forall i$, $\cD{\rho_{\bm{\theta}}}(\RLD{i})\in \tilde{T}_{\bm{\theta}}(\cM)$.
\item $\forall i$, $\SLDdual{i}=\RLDdual{i}$. 
\item $(\rldQFI)^{-1}=Z^{\mathrm{S}}_{\bm{\theta}}$. 
\item $(\sldQFI)^{-1}=Z^{\mathrm{R}}_{\bm{\theta}}$. 
\item $\forall i,j$, $\cD{\rho_{\bm{\theta}}}(\SLD{i})\bot \SLDdual{j}-\RLDdual{j}\mbox{ with respect to }\sldin{\cdot}{\cdot}$.
\item $\forall i,j$, $\cD{\rho_{\bm{\theta}}}(\RLD{i})\bot \SLDdual{j}-\RLDdual{j}\mbox{ with respect to }\sldin{\cdot}{\cdot}$.
\end{enumerate}
In the above result, we use the following definitions: 
The RLD tangent space at ${\bm{\theta}}$ is defined by the complex span of the RLD operators: 
\begin{equation}\label{rldT}
\tilde{T}_{\bm{\theta}}(\cM)=\mathrm{span}_\bbc \{\RLD{i}\}_{i=1}^d. 
\end{equation}
The Hermitian complex matrices $Z^{\mathrm{S}}_{\bm{\theta}}$ and $Z^{\mathrm{R}}_{\bm{\theta}}$ are defined by 
\begin{align}
Z^{\mathrm{S}}_{\bm{\theta}}:=\Big[ \rldin{\SLDdual{i}}{\SLDdual{j}} \Big],\quad 
Z^{\mathrm{R}}_{\bm{\theta}}:=\Big[ \sldin{\RLDdual{i}}{\RLDdual{j}} \Big].  
\end{align}

Next, we list equivalent conditions for the asymptotically classical model. 
\begin{enumerate}
\item The model $\cM$ is asymptotically classical at ${\bm{\theta}}$. 
\item $\forall i,j$, $\tr{\rho_{\bm{\theta}} [\SLD{i}\,,\,\SLD{j}]}=0$. 
\item $\Im Z^{\mathrm{S}}_{\bm{\theta}} =0$. 
\item $(\sldQFI)^{-1}=Z^{\mathrm{S}}_{\bm{\theta}}$. 
\item $\exists W_0>0$, $C_{\bm{\theta}}^H[W_0,\cM]= C_{\bm{\theta}}^{\rm S}[W_0,\cM]$. 
\item $\forall i,j$, $\cD{\rho_{\bm{\theta}}}(\SLD{i})\bot \SLD{j}\mbox{ with respect to }\sldin{\cdot}{\cdot}$. 
\item $\forall i,j$, $\cD{\rho_{\bm{\theta}}}(\SLD{i})\bot \RLD{j}\mbox{ with respect to }\rldin{\cdot}{\cdot}$. 
\item $\forall i,j$, $\cD{\rho_{\bm{\theta}}}(\SLD{i})\bot \SLD{j}\mbox{ with respect to }\rldin{\cdot}{\cdot}$. 
\item $\forall i,j$, $\cD{\rho_{\bm{\theta}}}(\RLD{i})\bot \SLD{j}\mbox{ with respect to }\rldin{\cdot}{\cdot}$. 
\end{enumerate}

\section{Proofs}\label{sec:AppPr}

\subsection{Proof of Lemma \ref{lem_lucond}} \label{sec:AppPr1}
First, it is straightforward to see that the transformation \eqref{nui_change0} 
preserves the first condition of \eqref{lu_cond}. 
By elementary calculus, we can show that the partial derivatives are transformed as 
\begin{align} \label{del_change1}
\frac{\del}{\del {\xi}_i}&= \sum_{j=1}^d\frac{\del\theta_j}{\del{\xi}_i}\frac{\del}{\del \theta_j} 
=\frac{\del}{\del \theta_i}-\sum_{j=\dI+1}^d\frac{\del\theta_j}{\del{\xi}_i}\frac{\del}{\del \theta_j}  \quad(i=1,2,\dots,\dI),\\ %\nonumber
\frac{\del}{\del {\xi}_j}&= \sum_{i=1}^d\frac{\del\theta_i}{\del{\xi}_j}\frac{\del}{\del \theta_i} 
=\frac{\del}{\del \theta_j}\quad(j=\dI+1,\dots,d). 
\label{del_change2}
\end{align} 
The second condition of \eqref{lu_cond} is verified as follows. 
For $i,j=1,2,\ldots,\dI$, using \eqref{del_change1} reads
\begin{align}
\frac{\del}{\del{\xi}_j}E_{\bm{\xi}}[{\hat{\theta}_i}(X)]&=
\frac{\del}{\del{\theta}_j}E_{\bm{\theta}}[{\hat{\theta}_i}(X)]- \sum_{k=\dI+1}^d\frac{\del\theta_k}{\del{\xi}_j}\frac{\del}{\del \theta_k}E_{\bm{\theta}}[{\hat{\theta}_i}(X)]\\
&=\frac{\del}{\del{\theta}_j}E_{\bm{\theta}}[{\hat{\theta}_i}(X)]=\delta_{i,j}. 
\end{align}
The second term in the first line vanishes because of the assumption of locally unbiasedness. 
For $i=1,2,\ldots,\dI$ and $j=\dI+1,\dI+2,\ldots,d$, we can directly check 
\be
\frac{\del}{\del\xi_j}E_{\bm{\xi}}[{\hat{\theta}_i}(X)]=\frac{\del}{\del\theta_j}E_{\bm{\theta}}[{\hat{\theta}_i}(X)]=0= \delta_{i,j}. 
\ee
Therefore, we prove the relation 
$\frac{\del}{\del{\xi}_j}E_{\bm{\xi}}[{\hat{\theta}_i}(X)]=0$ for $i=1,2,\dots,\dI$ and $j=1,2,\dots, d$. 

\subsection{Derivation of expression \eqref{MICRbound2}}\label{sec:AppPr2}
Let us define another bound by
\[
\overline{C}_{\bm{\theta};{\mathrm{I}}}[W_{\mathrm{I}},\cM]:=\min_{\Pi\mathrm{: POVM}} 
\Tr{W_{\mathrm{I}} J_{\bm \theta}^{\rm{I},\rm{I}}[\Pi]},
%J^{\bm{\theta};{\mathrm{I}},{\mathrm{I}}}[\Pi] },
\] 
then, we will prove $
\overline{C}_{\bm{\theta};{\mathrm{I}}}[W_{\mathrm{I}},\cM]
={C}_{\bm{\theta};{\mathrm{I}}}[W_{\mathrm{I}},\cM]$. 
The proof here is almost same line of argument as \cite{nagaoka89}. 
Using the CR inequality for any locally unbiased estimator for $\bm{\theta}_{\mathrm{I}}$, 
we have 
\begin{align*}
C_{\bm{\theta};{\mathrm{I}}}[W_{\mathrm{I}},\cM]= \min_{\hat{\Pi}_{\mathrm{I}}\mathrm{\,:l.u.\,at \,}\bm{\theta}} 
\Tr{W_{\mathrm{I}}V_{\bm{\theta};{\mathrm{I}}}[\hat{\Pi}_{\mathrm{I}}]}
\ge \Tr{W_{\mathrm{I}} 
J_{\bm{\theta}}^{\mathrm{I},{\mathrm{I}}}[\Pi] }
\end{align*}
Since this is true for all POVMs, we have the relation $C_{\bm{\theta};{\mathrm{I}}}[W_{\mathrm{I}},\cM]
\ge \overline{C}_{\bm{\theta};{\mathrm{I}}}[W_{\mathrm{I}},\cM]$. 
To prove the other direction, note that given a POVM $\Pi$ we can always construct a locally unbiased estimator 
at $\bm{\theta}_0=\big(\theta_1(0),\dots,\theta_d(0)\big)$. 
For example, 
\be\label{lu_est}
\hat{\theta}_i(X)= \theta_i(0)+\sum_{j=1}^d \left((J_{\bm{\theta}_0}[\Pi])^{-1}\right)_{ji}  \left.\frac{\del \log p_{\bm{\theta}}(X|\Pi)}{\del\theta_j}\right|_{\bm{\theta}_0}. 
\ee
Since this estimator is also locally unbiased for the parameter of interest $\bm{\theta}_{\mathrm{I}}$, 
and the MSE matrix about $\hat{\Pi}=(\Pi,\hat{{\bm{\theta}}})$ satisfies $V_{{\bm{\theta}}}[\hat{\Pi }]=J_{\bm{\theta}}^{-1}[\Pi]$. 
In turn, we have a relationship 
$V_{\bm{\theta};{\mathrm{I}}}[\hat{\Pi}_{\mathrm{I}}]
=J_{\bm{\theta}}^{\mathrm{I},{\mathrm{I}}}[\Pi]$ for the parameter of interest. 
Thus, we obtain $C_{\bm{\theta};{\mathrm{I}}}[W_{\mathrm{I}},\cM]\le \overline{C}_{\bm{\theta};{\mathrm{I}}}[W_{\mathrm{I}},\cM]$. 
This proves 
$C_{\bm{\theta};{\mathrm{I}}}[W_{\mathrm{I}},\cM]= 
\overline{C}_{\bm{\theta};{\mathrm{I}}}[W_{\mathrm{I}},\cM]$.

\subsection{Proofs for properties in section \ref{sec:QpoLocal}} \label{sec:AppPO}
\noindent{\it Property 1: The partial SLD Fisher information matrix under parameter change.}\\
The partial SLD Fisher information defined by \eqref{qpSLDinfo}: 
\[
J^{\mathrm{S}}_{\bm{\theta}}({\mathrm{I}}|{\mathrm{N}})
=\sldqfi{\bm{\theta};{\mathrm{I}},{\mathrm{I}}}-
J_{\bm{\theta};\rm{I},\rm{N}}^{\rm S}
%\sldqfi{\bm{\theta}_{\mathrm{I}}\bm{\theta}_{\mathrm{N}}} 
{\big(J^{\mathrm{S}}_{\bm{\theta};{\mathrm{N}},{\mathrm{N}}}\big)}^{-1}
\sldqfi{\bm{\theta};{\mathrm{N}},{\mathrm{I}}}
\]
is  invariant under any reparametrization of the nuisance parameters of the form \eref{nui_change0}
and is transformed as the same manner as the usual Fisher information matrix.\\
\begin{proof}
Let us consider the following change of the parameters, 
\be\label{ParaChange}
{\bm{\theta}}=(\bm{\theta}_{\mathrm{I}},\bm{\theta}_{\mathrm{N}})\mapsto{\bm{\xi}}=\big({\bm{\xi}}_{\mathrm{I}}(\bm{\theta}_{\mathrm{I}}),{\bm{\xi}}_{\mathrm{N}}(\bm{\theta}_{\mathrm{I}},\bm{\theta}_{\mathrm{N}}) \big). 
\ee 
The Jacobian matrix for this coordinate transformation is block diagonal as 
\[
T=\frac{\del{\bm{\theta}}}{\del{\bm{\xi}}}=\left(\begin{array}{cc} \frac{\del\bm{\theta}_{\mathrm{I}}}{\del{\bm{\xi}}_{\mathrm{I}}}&
 \frac{\del\bm{\theta}_{\mathrm{N}}}{\del{\bm{\xi}}_{\mathrm{I}}} \\[1ex]
0 & \frac{\del\bm{\theta}_{\mathrm{N}}}{\del{\bm{\xi}}_{\mathrm{N}}}\end{array}\right)
=: \left(\begin{array}{cc} t_\mathrm{I}& t  \\[1ex]
0 & t_\mathrm{N}\end{array}\right). 
\]
Let us express the SLD Fisher information matrix as
\be\nonumber 
{J}_{\bm{\xi}}^\mathrm{S}=\left(\begin{array}{cc}
J^{\mathrm{S}}_{{\bm{\xi}};{\mathrm{I}},{\mathrm{I}}}& 
\sldqfi{{\bm{\xi}};{\mathrm{I}},{\mathrm{N}}} \\[0ex] 
\sldqfi{{\bm{\xi}};{\mathrm{N}},{\mathrm{I}}}
& \sldqfi{{\bm{\xi}};{\mathrm{N}},{\mathrm{N}}} \end{array}\right), 
\ee
and use the transformation relation 
\begin{align*}
{J}_{\bm{\xi}}^\mathrm{S}&=T \sldQFI T^{\mathrm{T}}\\
&=\left(\begin{array}{cc} 
t_\mathrm{I}\sldqfi{\bm{\theta};{\mathrm{I}},{\mathrm{I}}}\left(t_\mathrm{I}\right)^{\mathrm{T}} 
+t_\mathrm{I}\sldqfi{\bm{\theta};{\mathrm{I}},{\mathrm{N}}} \left(  t \right)^{\mathrm{T}} 
+ t\sldqfi{\bm{\theta};{\mathrm{N}},{\mathrm{I}}}\left( t_\mathrm{I} \right)^{\mathrm{T}} 
+t\sldqfi{\bm{\theta};{\mathrm{N}},{\mathrm{N}}}\left( t \right)^{\mathrm{T}}  &
\left( t_\mathrm{I} \sldqfi{\bm{\theta};{\mathrm{I}},{\mathrm{N}}} 
+  tJ^{\mathrm{S}}_{\bm{\theta};{\mathrm{N}},{\mathrm{N}}} \right) \left( t_\mathrm{N} \right)^{\mathrm{T}} \\[1ex]
t_\mathrm{N}\left( \sldqfi{\bm{\theta};{\mathrm{N}},{\mathrm{I}}}
\left( t_\mathrm{I} \right)^{\mathrm{T}}
+ \sldqfi{\bm{\theta};{\mathrm{N}},{\mathrm{I}}}
\left( t \right)^{\mathrm{T}} \right)&
t_\mathrm{N}J^{\mathrm{S}}_{\bm{\theta};{\mathrm{N}},{\mathrm{N}}}\left( t_\mathrm{N} \right)^{\mathrm{T}}  
\end{array}\right). 
\end{align*}
Then, by the direct calculation, we obtain
\begin{align*}
{J}^{\mathrm{S}}_{\bm{\xi}}({\mathrm{I}}|{\mathrm{N}})
&= \sldqfi{{\bm{\xi}};{\mathrm{I}},{\mathrm{I}}}-
\sldqfi{{\bm{\xi}};{\mathrm{I}},{\mathrm{N}}} 
\left(J^{\mathrm{S}}_{{\bm{\xi}};{\mathrm{N}},{\mathrm{N}}}\right)^{-1}
\sldqfi{{\bm{\xi}};{\mathrm{N}},{\mathrm{I}}}\\
&=
t_\mathrm{I}\sldqfi{\bm{\theta};{\mathrm{I}},{\mathrm{I}}}\left(t_\mathrm{I}\right)^{\mathrm{T}} 
+t_\mathrm{I}\sldqfi{\bm{\theta};{\mathrm{I}},{\mathrm{N}}} \left(  t \right)^{\mathrm{T}} 
+ t\sldqfi{\bm{\theta};{\mathrm{N}},{\mathrm{I}}}\left( t_\mathrm{I} \right)^{\mathrm{T}} 
+t\sldqfi{\bm{\theta};{\mathrm{N}},{\mathrm{N}}}\left( t \right)^{\mathrm{T}}\\
&\quad-\left( t_\mathrm{I} 
\sldqfi{\bm{\theta};{\mathrm{I}},{\mathrm{N}}} 
+  tJ^{\mathrm{S}}_{\bm{\theta};{\mathrm{N}},{\mathrm{N}}} \right)
\left(  J^{\mathrm{S}}_{\bm{\theta};{\mathrm{N}},{\mathrm{N}}}\right)^{-1} 
\left( \sldqfi{\bm{\theta};{\mathrm{N}},{\mathrm{I}}}
\left( t_\mathrm{I} \right)^{\mathrm{T}}
+ \sldqfi{\bm{\theta};{\mathrm{N}},{\mathrm{I}}}
\left( t \right)^{\mathrm{T}} \right)\\
&=t_\mathrm{I}  
\left( \sldqfi{\bm{\theta};{\mathrm{I}},{\mathrm{I}}}-
\sldqfi{\bm{\theta};{\mathrm{I}},{\mathrm{N}}} 
{\big(J^{\mathrm{S}}_{\bm{\theta};{\mathrm{N}},{\mathrm{N}}}\big)}^{-1}
\sldqfi{\bm{\theta};{\mathrm{N}},{\mathrm{I}}}\right) \left( t_\mathrm{I} \right)^{\mathrm{T}}\\
&=\frac{\del\bm{\theta}_{\mathrm{I}}}{\del{\bm{\xi}}_{\mathrm{I}}}  
J^{\mathrm{S}}_{\bm{\theta}}({\mathrm{I}}|{\mathrm{N}})
\left(\frac{\del\bm{\theta}_{\mathrm{I}}}{\del{\bm{\xi}}_{\mathrm{I}}}\right)^\mathrm{T}. 
\end{align*}
This shows the statement. 
\end{proof}

{\it Property 2:  After parameter orthogonalization, the SLD operator about the parameter of interest in the new parametrization 
is expressed as }
\be\label{sldxi-theta}
L^{\mathrm{S}}_{{\bm{\xi}};1}=(\sldQFIinv{1,1})^{-1}\SLDdual{1}.
\ee
\begin{proof} 
Inserting $\SLD{i}=\sum_j \sldqfi{{\bm{\theta}};j,i}\SLDdual{j}$ into expression \eqref{sldxi} 
with ${\alpha}=1$, we have $L^{\mathrm{S}}_{{\bm{\xi}};1}=\sum_{i,j}\frac{\del \theta_i}{\del {\xi}_1}\sldqfi{{\bm{\theta}};j,i}\SLDdual{j}$, where the summation over the index $i$ vanishes except for $i=1$ due to assumption \eqref{qdiffeq}. 
Then, we get $L^{\mathrm{S}}_{{\bm{\xi}};1}=\sum_{i}\frac{\del \theta_i}{\del {\xi}_1}\sldqfi{{\bm{\theta}};1,i}\SLDdual{1}$, and thus $L^{\mathrm{S}}_{{\bm{\xi}};1}$ is proportional to $\SLDdual{1}$. The proportionality factor is 
determined by $\sldin{L^{\mathrm{S}}_{{\bm{\xi}};1}}{\SLDdual{1}}=\frac{\del \theta_1}{\del {\xi}_1}=1$. 
\end{proof}

\noindent
{\it Property 3: The partial SLD Fisher information of the parameter of interest after the parameter orthogonalization is preserved.}\\
Although the parameter orthogonalization method enables us to have the relation 
$J^{\mathrm{S};1,1}_{\bm{\xi}}=(\sldqfi{{\bm{\xi}};1,1})^{-1}$ in the new parameterization, 
it preserves the partial SLD Fisher information for the parameter of interest as 
\be
J^{\mathrm{S};1,1}_{\bm{\xi}}=\sldQFIinv{1,1}. 
\ee
That is, the precision limit for the parameter of interest does not change as should be. 
(See also Theorem \ref{thmAN2} in section \ref{sec:1para2}.) 
\begin{proof}
The simplest way to show this relation is to 
compute $\sldqfi{{\bm{\xi}};1,1}=\langle L^{\mathrm{S}}_{{\bm{\xi}};1},L^{\mathrm{S}}_{{\bm{\xi}};1}\rangle_{\rho_{\bm{\xi}}}$. 
Using the property 1, we have $\sldqfi{{\bm{\xi}};1,1}=(\sldQFIinv{1,1})^{-2}\sldin{\SLDdual{1}}{\SLDdual{1}}
=(\sldQFIinv{1,1})^{-2}\sldQFIinv{1,1}=(\sldQFIinv{1,1})^{-1}$. 
\end{proof}

\subsection{Proof for Theorem \ref{thmAN2}} \label{sec:AppPr4}
We first show that the inequality $V_{\bm{\theta};{\mathrm{I}}}[\hat{\Pi}] \ge (\sldQFI)^{-1}_{1,1}$ holds for all locally unbiased 
estimators for $\bm{\theta}_{\mathrm{I}}=\theta_1$. We then show its achievability by constructing an optimal estimator explicitly. 

Given an $\hat{\Pi}_{\mathrm{I}}=(\Pi,\hat{{\bm{\theta}}}_{\mathrm{I}})$ for the parameter of interest, 
we add an arbitrary estimator for the nuisance parameter $\hat{{\bm{\theta}}}_{\mathrm{N}}$, 
for example $\hat{{\bm{\theta}}}_{\mathrm{N}}$ can be a constant function. 
Suppose an estimator $\hat{\Pi}=(\Pi,\hat{{\bm{\theta}}})$ for the parameter ${\bm{\theta}}=(\bm{\theta}_{\mathrm{I}},\bm{\theta}_{\mathrm{N}})$ is locally unbiased 
for $\bm{\theta}_{\mathrm{I}}$ at ${\bm{\theta}}$, the bias matrix defined in Lemma \ref{lemm_crgen} takes of the form:
\be
B_{\bm{\theta}}[\hat{\Pi}]=\left(\begin{array}{cc}{I_{\dI}} & {0} \\B_1 & B_2\end{array}\right), 
\ee 
with some matrices $B_1,B_2$. 
The generalized SLD quantum CR inequality in Lemma \ref{lemm_crgen} then gives 
the $\dI\times \dI$ block matrix of 
\begin{align}
P_{\mathrm{I}}V_{\bm{\theta}}^Q[\Pi(\hat{{\bm{\theta}}})]P_{\mathrm{I}}&\ge P_{\mathrm{I}}  B_{\bm{\theta}}[\hat{\Pi}] (\sldQFI)^{-1} B_{\bm{\theta}}[\hat{\Pi}]^{\mathrm T}P_{\mathrm{I}}
=P_{\mathrm{I}} (\sldQFI)^{-1} P_{\mathrm{I}}=
J_{\bm{\theta}}^{\mathrm{S};\rm{I,I}},
\end{align}
where $P_{\mathrm{I}}=\sum_{i=1}^{\dI} e_ie_i^{\mathrm T}$ denotes the projector onto a subspace of the first $k$ element, 
i.e., the subspace for the parameters of interest. 
Combining this with Corollary \ref{cor_holevo}, we show that any locally unbiased estimator for the parameters of interest 
satisfies the matrix inequality $V_{\bm{\theta};{\mathrm{I}}}[\hat{\Pi}]\ge 
J_{\bm{\theta}}^{\mathrm{S};\rm{I,I}}$.
In particular, by letting $\bm{\theta}_{\mathrm{I}}=\theta_1$ and $\bm{\theta}_{\mathrm{N}}=(\theta_2,\dots,\theta_d)$, we obtain the converse part of this theorem. 

To make our discussion clear, we perform the parameter orthogonalization method with respect to the 
SLD Fisher information matrix. Then, the model in the new parametrization ${\bm{\xi}}=({\bm{\xi}}_{1},{\bm{\xi}}_{\mathrm{N}})$ is an orthogonal model 
according to this partition. Let us consider a projection measurement $\Pi^*=\{\Pi_x\}$ 
composed of the spectral decomposition of the SLD operator $L^{\mathrm{S}}_{{\bm{\xi}};1}$ and 
an estimator 
\be
\hat{{\bm{\xi}}}_1(x)={\xi}_1+g_{\bm{\xi}}^{1,1}\frac{\del \ell_{{\bm{\xi}}}(x)}{\del{\xi}_1},
\ee
with $\ell_{{\bm{\xi}}}(x)=\log\tr{\rho_{\bm{\xi}}\Pi_x}$. 
It is straightforward to show that this estimator ${\hat{\Pi}}^*_{\mathrm{I}}=(\Pi^*,\hat{\theta}_{1})$ for 
the parameter of interest is locally unbiased for $\xi_{\rm I}=\xi_{1}=\theta_1$ at ${\bm{\theta}}$. 
The $(1,1)$ component of the MSE matrix is easily computed by 
\begin{align}
V_{\bm{\theta};{\mathrm{I}}}[{\hat{\Pi}}^*_{\mathrm{I}}]=\sum_{x\in\cX} (J_{\bm{\xi}}^{\mathrm{S};1,1})^2 \big(\frac{\del \ell_{{\bm{\xi}}}(x)}{\del{\xi}_1}\big)^2\tr{\rho_{\bm{\xi}}\Pi_x}
= (J_{\bm{\xi}}^{\mathrm{S};1,1})^2J_{{\bm{\xi}};1,1}[\Pi^*]
= (J_{\bm{\xi}}^{\mathrm{S};1,1})^2 \sldqfi{{\bm{\xi}};1,1}=J_{\bm{\xi}}^{\mathrm{S};1,1},  
\end{align}
where Lemma \ref{lem_optPVM} is used to get the second equality. 
The last line follows from the fact that the model is orthogonal. 
Therefore, using property 3) of section \ref{sec:QpoGlobal}, we obtain 
\be
V_{\bm{\theta};{\mathrm{I}}}[{\hat{\Pi}}^*_{\mathrm{I}}]=J_{\bm{\xi}}^{\mathrm{S};1,1}=\sldQFIinv{1,1}. 
\ee
The statement about the optimal estimator is also immediate if we use property 2) of section \ref{sec:QpoGlobal}. 

\subsection{Proof of Ineq. \eqref{WVCW}}\label{AC-3}
Assume that a sequence of estimators $\{\Pi^{(n)}\}_{n=0}^{\infty}$ satisfies the local asymptotic covariance condition at ${\bm{\theta}}$.
We denote the limiting distribution family and the Fisher information matrix by $\{P_{t}\}_{t}$ and $ J$. We will show the following.
There exists a POVM $M_0$ on ${\cal H}$ such that
\begin{align}
 J\le J_{\bm{\theta}}^{M_0}.
\label{HJ10}
\end{align}
Here, $J_{\bm \theta}^M=J_{\bm \theta}[M]$ is the Fisher information matrix of the family of 
the resultant distributions when the measurement corresponding to $M$ is applied.
The local asymptotic covariance condition guarantees the unbiased condition for the family of the distributions 
$\{P_{{\bm t}}\}_{{\bm t}}$ on $\mathbb{R}^d$ when 
the variable on $\mathbb{R}^d$ is considered as an estimator of ${\bm t}$.
Therefore, Cram\'{e}r-Rao inequality shows that
$V_{\bm{\theta}_0}[\{\Pi^{(n)}\}_{n=0}^{\infty}]
\ge (J_{\bm{\theta}}^{M_0})^{-1}$.
Hence, using \eqref{MICRbound} and this inequality, we obtain \eqref{WVCW}. 

In the following, we show \eqref{HJ10}.
Define
\begin{align}
P_{\bm{\theta}_0,{\bm t}}^{(n)}(B):=
\tr {\rho_{\bm{\theta}_0+\frac{{\bm t}}{\sqrt{n}}} \Pi( \{ \hat{{\bm{\theta}}} | 
(\hat{{\bm{\theta}}}-\bm{\theta}_0)\sqrt{n}- {\bm t} \in B
 \}  ) }.
\end{align}
Let $F(P,Q)$ be the fidelity between two distributions $P$ and $Q$.
Lemma 20 of \cite{YCH18} shows that
\begin{align}
F(P_{\bm{\theta}_0,0},P_{\bm{\theta}_0,{\bm t}})
\ge
\limsup_{n\to \infty}
F(P_{\bm{\theta}_0,0}^{(n)},P_{\bm{\theta}_0,{\bm t}}^{(n)}).
\end{align}

Let ${\cal M}$ be the set of extremal points in the set of POVMs on ${\cal H}$.
Let ${\cal P}({\cal M})$ be the set of probability distributions on 
${\cal M}$.
Hence, any POVM can be written as an element of ${\cal P}({\cal M})$.
We denote the set of outcomes of POVM by ${\cal Y}$.
An adaptive measurement can be written as 
a set of $\{f_k\}_{k=1}^n$
functions $f_k: {\cal Y}^{k-1}\to{\cal P}({\cal M})$.
Assume that the norm of $t$ is smaller than a certain value $R$.

Let $P_{{\bm t},1}^{(n),k} $ be the distribution of 
the initial $k$ outcomes when the true state is $\rho_{{\bm{\theta}}+{\bm t}/\sqrt{n}}$.
Let $F_{{\bm t},M}$ be the fidelity between 
$P^M_{{\bm{\theta}}} $ and $P^M_{{\bm{\theta}}+{\bm t}} $,
where $P^M_{{\bm{\theta}}} $ is the output distribution with the POVM $M$ and the state $\rho_{\bm{\theta}}$.
We inductively define $P_{{\bm t},2}^{(n),k}$
as
\be
P_{{\bm t},2}^{(n),k}(d y^{k-1}):=
\sqrt{P_{0,1}^{(n),k}(d y^{k-1})} \sqrt{P_{{\bm t},1}^{(n),k}(dy^{k-1})}/
\prod_{k'=1}^k \int_{{\cal Y}^{k'-1}} F_{{\bm t}/\sqrt{n},f_{k'}(y^{k'-1})} 
P_{{\bm t},2}^{(n),k'}(d y^{k'-1}).
\ee
These definitions are quite similar to the definitions in \cite{adaptive-channel}.
The fidelity $F(P_{\bm{\theta}_0,0}^{(n)},P_{\bm{\theta}_0,{\bm t}}^{(n)})$
equals 
\begin{align}
\prod_{k=1}^n \int_{{\cal Y}^{k-1}} F_{{\bm t}/\sqrt{n},f_k(y^{k-1})} 
P_{{\bm t},2}^{(n),k}(d y^{k-1}).
\end{align}

Since the set ${\cal M}$ and the range of ${\bm t}$ are compact, 
the difference $8(1- F_{{\bm t}/\sqrt{n},M})n -
{\bm t}^Y  J_{{\bm{\theta}}}^M{\bm t}$
converges to zero uniformly with respect to $M$ and ${\bm t}$.
That is, the difference is uniformly upper bounded by $a_n$, and $a_n$ goes to zero.
Hence, 
\begin{align}
\log \prod_{k=1}^n \int_{{\cal Y}^{k-1}} F_{{\bm t}/\sqrt{n},f_k(y^{k-1})} 
P_{{\bm t},2}^{(n),k}(d y^{k-1}) %\nonumber\\
=&
\sum_{k=1}^n 
\log (\int_{{\cal Y}^{k-1}} F_{{\bm t}/\sqrt{n},f_k(y^{k-1})} 
P_{{\bm t},2}^{(n),k}(d y^{k-1})) \label{NN1}\\
\cong &
-\sum_{k=1}^n 
\int_{{\cal Y}^{k-1}}
\frac{1}{8n }{\bm t}^{\rm T} J_{{\bm{\theta}}}^{f_k(y^{k-1})}{\bm t}
P_{{\bm t},2}^{(n),k}(d y^{k-1})\label{NN2}
\end{align} 
where the difference between \eqref{NN1} and \eqref{NN2}
is upper bounded by $a_n$.
When ${\bm t}$ goes to zero, 
\begin{align}
\max_{k}1- F(P_{{\bm t},2}^{(n),k},P_{0,1}^{(n),k})
\le
1-F(\rho_{{\bm{\theta}}}^{\otimes n},
\rho_{{\bm{\theta}}+\epsilon {\bm t}/\sqrt{n}}^{\otimes n})
\cong
\frac{1}{8} {\bm t}^{\rm T} J_{{\bm{\theta}}}^{\rm S} {\bm t} \epsilon^2 .
\end{align}
This value goes to zero as $\epsilon \to 0$.
This fact means that
the difference between 
$P_{{\bm t},2}^{(n),k}$ and $P_{0,1}^{(n),k}$
is upper bonded uniformly with respect to $k$.
Therefore,
\begin{align}
&\frac{1}{\epsilon^2}\log \prod_{k=1}^n \int_{{\cal Y}^{k-1}} 
F_{\epsilon {\bm t}/\sqrt{n},f_k(y^{k-1})} 
P_{{\bm t},2}^{(n),k}(d y^{k-1})\nonumber \\
\cong &
-\sum_{k=1}^n 
\int_{{\cal Y}^{k-1}}
\frac{1}{8n } {\bm t}^{\rm T} J_{{\bm{\theta}}}^{f_k(y^{k-1})}{\bm t}
P_{0,1}^{(n),k}(d y^{k-1})\label{NN3}.
\end{align}
Now, we define
$M^{(n)}:=
\frac{1}{n}\sum_{k=1}^n 
\int_{{\cal Y}^{k-1}}
f_k(y^{k-1})
P_{0,1}^{(n),k}(d y^{k-1})$.
Notice that $f_k(y^{k-1})$ expresses a distribution on ${\cal M}$. 
\begin{align}
&\frac{1}{8}{\bm t}^{\rm T} J {\bm t}
=
-\lim_{\epsilon \to 0}
\log F(P_{\bm{\theta}_0,0},P_{\bm{\theta}_0,\epsilon {\bm t}})
\ge
-\lim_{\epsilon \to 0}
\log
\limsup_{n\to \infty}
F(P_{\bm{\theta}_0,0}^{(n)},P_{\bm{\theta}_0,{\bm t}}^{(n)}) \nonumber\\
=&
\lim_{n \to \infty}\frac{1}{8 } {\bm t}^{\rm T} J_{{\bm{\theta}}}^{M^{(n)}}{\bm t}
\label{NN3}.
\end{align}
Since the set ${\cal M}$ is compact,
there exist a POVM $M_0$ and a subsequence $\{M^{(n_l)}\}$
such that $M^{(n_l)}\to M_0$.
Hence, we obtain 
\begin{align}
&\frac{1}{8}{\bm t}^{\rm T} J{\bm t}
\ge \frac{1}{8 }{\bm t}^{\rm T} J_{{\bm{\theta}}}^{M_0}{\bm t}
\label{NN4}.
\end{align}
Since ${\bm t}$ is an arbitrary, we obtain \eqref{HJ10}.

\subsection{Proofs of \eqref{Ho-ine} and \eqref{Ho-ine-2}}\label{AC-4}
Let $\Pi$ be a locally unbiased estimator at ${\bm{\theta}}$.
We choose the operator $X_i:=\int (x_i-\theta_i) \Pi(dx)$.
The locally unbiased condition for $\Pi$ implies the condition
$\tr{\frac{\del}{\del\theta_j}\rho_{{\bm{\theta}}}X_i}=\delta_{i,j}$ for 
$i,j=1, \ldots, d$.

Next, we show the matrix inequality
\be\label{Hol3}
V_{\bm{\theta}}[\Pi] \ge Z_{\bm{\theta}}({\bm X}).
\ee

Let ${\bm a}=(a_i)$ be an arbitrary vector in $\mathbb{C}^d$.
We choose the operator 
$A:=\int \sum_i \bar{a}_i(x_i-\theta_i) \Pi(dx)$.
Then, in the same way as \eqref{EH3}, we have
\if0
\begin{align}
& \int |\sum_i \bar{a}_i(x_i-\theta_i)|^2 \Pi(dx)- AA^\dagger \nonumber\\
=&
\int ( (\sum_i \bar{a}_i(x_i-\theta_i))-A)
( (\sum_i \bar{a}_i(x_i-\theta_i))-A)^\dagger   \Pi(dx)
\ge 0.
\end{align}
\fi
\begin{align}
{\bm a}^\dagger V_{\bm{\theta}}[\Pi] {\bm a}
\ge
\tr{\rho_{\bm{\theta}} AA^\dagger}
={\bm a}^\dagger Z_{\bm{\theta}}({\bm X}){\bm a}.
\end{align}
Since ${\bm a}$ is an arbitrary vector in $\mathbb{C}^d$, we obtain 
\eqref{Hol3}.

Using \eqref{Hol3}, we have
\be
W^{1/2}V_{\bm{\theta}}[\Pi]W^{1/2} \ge W^{1/2}Z_{\bm{\theta}}({\bm X})W^{1/2}.
\ee
Thus,
\be
W^{1/2}V_{\bm{\theta}}[\Pi]W^{1/2} - 
W^{1/2} (\Re Z_{\bm{\theta}}({\bm X}))W^{1/2}
\ge 
W^{1/2} (\Im Z_{\bm{\theta}}({\bm X}))W^{1/2}.
\ee
Given an antisymmetric matrix $C$,
the minimum of the trace of symmetric matrices $B$ to satisfy 
the matrix inequality $B \ge i C$ is $\Tr{ |C|}$.
Hence, we have
\be
\Tr{W^{1/2}V_{\bm{\theta}}[\Pi]W^{1/2} - 
W^{1/2} (\Re Z_{\bm{\theta}}({\bm X}))W^{1/2}}
\ge 
\Tr{|W^{1/2} (\Im Z_{\bm{\theta}}({\bm X}))W^{1/2}|},
\ee
which implies
\be \label{Hol-5}
\Tr{ WV_{\bm{\theta}}[\Pi]} \ge
\Tr{W^{1/2} (\Re Z_{\bm{\theta}}({\bm X}))W^{1/2}}
+\Tr{|W^{1/2} (\Im Z_{\bm{\theta}}({\bm X}))W^{1/2}|}.
\ee
Taking the minimum with respect to ${\bm X}$,
we obtain \eqref{Ho-ine}.

We can show the inequality \eqref{Ho-ine-2} in the same way as 
\eqref{Ho-ine}.
Consider the model with nuisance parameters.
Let $\Pi$ be a locally unbiased estimator at ${\bm{\theta}}$.
We choose the operator $X_i:=\int (x_i-\theta_i) \Pi(dx)$
for $i=1, \ldots, \dI$.
Then, 
the vector ${\bm X}=(X_1,\ldots, X_{\dI})$ satisfies 
the condition
$\tr{\frac{\del}{\del\theta_j}\rho_{{\bm{\theta}}}X_i}=\delta_{i,j}$ for 
$i=1, \ldots, \dI$ and $j=1,\ldots, d$.
Since we have \eqref{Hol-5} in the same way,
taking the minimum with respect to ${\bm X}$,
we obtain \eqref{Ho-ine-2}.

%%%%%%%%%%%%%%%%%%%%%%%%%%%%%%%%%%%%%%%

\end{document}